\documentclass[acmsmall,screen,nonacm]{acmart}
\usepackage{bm}
\usepackage{color}
\usepackage{graphicx}
\usepackage{booktabs}
\usepackage{multirow}
\usepackage{threeparttable}
\usepackage{enumerate}
\usepackage[english]{babel}
\usepackage{algorithmicx}
\usepackage{algorithm}
\usepackage{algpseudocode}
\usepackage{nomencl}
\usepackage{multicol}
\usepackage{commath}
\usepackage{makecell}
\usepackage{blindtext}
\usepackage{soul}
\usepackage{comment}
\usepackage{braket}
\newtheorem{theorem}{Theorem}[section]
\newtheorem{lemma}[theorem]{Lemma}
\newtheorem{proposition}[theorem]{Proposition}

\theoremstyle{remark}
\newtheorem{remark}{Remark}[section]

\newcommand{\cutval}{\ensuremath{c_f}}

\makenomenclature
\graphicspath{{Figures/}{ieee_bus_systems_no_legends_large_font_outputs/pngs/}}
\AtBeginDocument{%
  }

\renewcommand\footnotetextcopyrightpermission[1]{}

\begin{document}
\title[PACE-QAOA]{PACE-QAOA: Physics-Constrained Quantum Optimization for Qubit-Efficient Power System Islanding}

\author{Yuqi Jiang}
\affiliation{%
  \department{Department of Electrical Engineering}
  \institution{Pennsylvania State University}
  \city{University Park}
  \state{Pennsylvania}
  \country{USA}}
\email{yzj5282@psu.edu}

\author{Zhiding Liang}
\affiliation{%
  \department{Department of Computer Science}
  \institution{Rensselaer Polytechnic Institute}
  \city{Troy}
  \state{New York}
  \country{USA}}

\author{Qiang Guan}
\affiliation{%
  \department{Department of Computer Science}
  \institution{Kent State University}
  \city{Kent}
  \state{Ohio}
  \country{USA}}

\author{Yan Li}
\affiliation{%
  \department{Department of Electrical Engineering}
  \institution{Pennsylvania State University}
  \city{University Park}
  \state{Pennsylvania}
  \country{USA}}
\email{yql5925@psu.edu}

\author{Ganesh Kumar Venayagamoorthy}
\affiliation{%
  \department{Real-Time Power and Intelligent Systems Laboratory, Holcombe Department of Electrical and Computer Engineering}
  \institution{Clemson University}
  \city{Clemson}
  \state{South Carolina}
  \country{USA}}
\affiliation{%
  \institution{University of Pretoria}
  \city{Pretoria}
  \country{South Africa}}

\renewcommand{\shortauthors}{Jiang et al.}

\begin{abstract}
Increasing renewable-energy penetration heightens power-system variability and complicates disturbance containment. Controlled islanding mitigates cascading failures by partitioning a stressed network to limit disrupted power transfer while preserving each island's operational integrity, but this constrained partitioning problem is NP-hard. Although QAOA offers a complementary search strategy, limited near-term qubit capacity restricts conventional formulations. This paper presents a qubit-efficient hybrid quantum framework combining a physics-informed compact encoding with Lagrangian constraint handling and classical feasibility refinement. The encoding exploits grid structure while formally preserving the original feasible solution space and objective. For a fixed island count on sparse working graphs, the formulation reduces phase-separator and per-layer gate complexity from quadratic to linear scaling with system size. Tests on eight IEEE systems ranging from 9 to 89 buses and multiple quantum-provider backends produce feasible, high-quality islanding solutions under practical circuit and sampling budgets. Factorial ablation attributes resource and runtime improvements to the complementary effects of compact encoding and qubit-efficient constraint handling. Noise analysis shows stable solution quality under tested device noise, while landscape diagnostics reveal smoother, more consistently scaled QAOA cost surfaces and improved parameter-optimization behavior. These results offer a transferable approach for scaling constrained quantum optimization toward larger real-world applications on near-term hardware.
\end{abstract}



\maketitle

\section{Introduction}
Quantum computing has developed rapidly in both hardware accessibility and algorithmic design, opening a new computational paradigm for large-scale combinatorial optimization problems whose exponentially growing solution spaces place them beyond the reach of exact classical methods \cite{preskill2018quantum,bharti2022nisq,chicano2024combinatorial}. In the noisy intermediate-scale quantum (NISQ) era, hybrid quantum-classical algorithms are especially important because they use parameterized quantum circuits together with classical feedback, rather than relying on fully fault-tolerant machines \cite{cerezo2021variational,gemeinhardt2024nisq}. This progress has positioned quantum optimization as a promising route to tackle hard combinatorial problems across science and engineering, from the natural sciences to large-scale infrastructure, where discrete choices, complex interdependencies, and problem-specific constraints jointly render the search space intractable for classical solvers at scale.

The Quantum Approximate Optimization Algorithm (QAOA) is a leading NISQ algorithm for combinatorial optimization because it alternates problem-dependent phase operators with mixing operators and uses classical optimization to tune circuit parameters \cite{farhi2014qaoa,blekos2024review}. Recent evidence of scaling advantage on a classically intractable benchmark has further strengthened interest in QAOA as a practical quantum-optimization component \cite{shaydulin2024scaling}. These properties have attracted attention across a wide range of practical domains, provided that formulations are adapted to near-term hardware constraints.

In finance, Zaman \emph{et al.}~\cite{zaman2024po} use QAOA to solve portfolio selection. In chemistry, Mustafa \emph{et al.}~\cite{mustafa2022variational} apply variational quantum algorithms to molecular ground-state estimation, and Robert \emph{et al.}~\cite{robert2021protein} solve lattice protein folding on superconducting hardware. In logistics, Harwood \emph{et al.}~\cite{harwood2021routing} apply QAOA to vehicle routing with time windows. In transportation, Villanueva \emph{et al.}~\cite{villanueva2025hybrid} tackle traffic-congestion minimization with a hybrid quantum-classical solver. Across these domains, quantum optimization has demonstrated tangible potential as a general-purpose framework for real-world combinatorial problems.

Within this broad application landscape, power systems present a distinctive optimization setting in which large discrete decision spaces are tightly coupled with physical and operational constraints. These challenges are becoming more acute as the growing penetration of distributed and renewable energy resources reshapes grid operation, replacing predictable centralized generation with intermittent, environment-dependent supply distributed across many buses \cite{dincer2000renewable}. This transition tightens operating margins and increases the risk that a local disturbance escalates into a widespread cascading outage \cite{guo2017critical}. Controlled islanding is a key resilience mechanism against such events, deliberately partitioning a stressed network into smaller self-sustaining islands that isolate faults, preserve critical loads, and support subsequent restoration. Determining such partitions requires a discrete network split that simultaneously satisfies generation, load, and connectivity constraints. Because this constrained combinatorial optimization problem contains weighted minimum bisection as a special case, it is NP-hard and becomes increasingly difficult to solve as the network grows. This computational challenge motivates investigating quantum-optimization methods as an alternative solution approach.

A similar motivation has emerged in power-system operations, where many planning and operational tasks involve discrete network decisions under physical constraints. QAOA-based methods have been studied for data-driven graph optimization \cite{jing2023datadriven}, optimal PMU placement \cite{10980373}, unit commitment benchmarking \cite{adler2025scaling}, and contingency analysis \cite{11250180}, indicating growing interest in quantum optimization for grid operation.

Power-system islanding is a particularly relevant target for advanced optimization methods. As renewable energy penetration increases, environment-dependent generation introduces additional uncertainty into power balance and post-disturbance resilience \cite{xu2024resilience,tang2023intentional}. These uncertainties can reduce operating margins and make it harder to contain disturbances once they begin to propagate through the network. Intentional islanding addresses this risk by separating a stressed grid into self-sustained islands, thereby limiting cascading failures, preserving critical loads, and supporting restoration \cite{ding2015mixed,babaei2023intentional}. Since islanding must determine a physically meaningful network partition under multiple operating requirements, it forms a constrained combinatorial problem with substantial computational complexity.

Classical studies have developed a range of solution strategies for controlled islanding. Earlier studies applied linear programming and ordered binary decision diagrams to construct splitting decisions for large-scale power systems \cite{kyriacou2017controlled,sun2003splitting}. More recent formulations incorporate richer network physics and stability requirements: Daniar \emph{et al.} use Benders decomposition to accommodate full AC power-flow constraints \cite{daniar2023benders}, while Li \emph{et al.} formulate a mixed-integer second-order cone program with connectivity and frequency-stability constraints \cite{li2023socp}. Alongside these mathematical-programming approaches, Priya and Roselyn develop a multi-layer constrained spectral-clustering method that accounts for generator coherency and real- and reactive-power disruption \cite{priya2024multilayer}. Despite this progress, the combinatorial search space grows rapidly with network size, motivating continued investigation of alternative optimization pipelines for physically meaningful islanding decisions.

Quantum methods have begun to address power system partitioning, although early formulations largely optimize graph structure without enforcing all physical requirements of controlled islanding \cite{hartmann2025partitioning}. A recent study narrows this gap through physics-aware constraint evaluation and structured classical post-processing, which repairs sampled partitions into feasible islanding decisions \cite{jiang2026regrid}. Nevertheless, its qubit requirements still grow rapidly with network size, limiting the scale of power systems that can be addressed on near-term quantum hardware. This scalability bottleneck motivates a more qubit-efficient formulation that retains the feasibility benefits of hybrid post-processing.

Reducing circuit width is therefore a prerequisite for extending QAOA-based islanding beyond small networks on near-term hardware. The qubit demand arises not only from representing the island assignments, but also from auxiliary variables introduced when operational constraints are embedded directly in the quantum model. On the one hand, a compact encoding can reduce the number of variables required to represent island assignments. On the other hand, classical Lagrangian updates can handle constraints that would otherwise introduce additional qubits. Complementing these qubit-saving strategies, post-processing converts measured samples into valid islanding decisions while preserving the reduced quantum resource requirements \cite{niroula2022constrained,ruan2023constrained}.

To address the qubit limitation of near-term quantum hardware, this study proposes a qubit-saving QAOA framework that combines problem-aware compact quantum encoding with a slack-free Lagrangian treatment of key feasibility requirements for controlled power-system islanding. This integrated design reduces the quantum resources needed to represent the optimization problem and its operational constraints. A constrained two-stage post-processing routine is then applied to repair sampled bitstrings and select feasible islanding solutions. The complete workflow is summarized in Fig.~\ref{fig:overall_diagram}. The proposed framework is validated on IEEE benchmark systems ranging from 9 to 89 buses, demonstrating that the reduced formulation lowers the qubit requirement while preserving high-quality feasible islanding decisions. Its applicability is further demonstrated across different quantum-provider backends, where power-system islanding serves as a common workload for cross-platform benchmarking. Landscape analysis further shows that the qubit-saving formulations produce smoother and more consistently scaled optimization surfaces, improving trainability alongside resource efficiency. The main contributions of this paper can be summarized as:
\begin{itemize}
    \item This study establishes an exact, power system-informed representation for qubit-efficient QAOA-based controlled islanding. By embedding known network structure into a physics-informed compact encoding of islanding decisions, the proposed encoding hierarchy substantially reduces the quantum register while preserving the feasible physical partitions and the optimal islanding objective.
    \item This study introduces a slack-free Lagrangian constraint-handling strategy that incorporates the operational requirements of controlled islanding without auxiliary quantum registers. The resulting formulation reduces constraint-induced qubit and Hamiltonian complexity while maintaining physical feasibility and solution quality within the hybrid optimization framework.
    \item The framework is validated on eight IEEE benchmark systems (9- to 89-bus) and benchmarked across available quantum-provider backends. The framework not only consistently produces feasible islanding solutions but also maintains stable cut quality across backend platforms and the tested noise conditions, demonstrating its scalability, hardware adaptability, and noise resilience.
    \item Factorial ablation and landscape analyses demonstrate the complementary benefits of compact encoding and slack-free constraint handling. The proposed formulation reduces quantum-resource and circuit complexity while simplifying the Hamiltonian structure and producing smoother, more consistently scaled optimization landscapes. These combined advantages improve hardware compatibility and parameter optimization, providing key enablers for scaling constrained quantum optimization to larger power-system applications.
\end{itemize}

The remainder of this paper is organized as follows. Section~\ref{formulation} introduces the controlled-islanding problem and proposes the compact assignment encodings. Section~\ref{quantum} presents the hybrid QAOA solution framework and its Lagrangian constraint handling and feasibility refinement. Section~\ref{sec:complexity} describes the quantum-resource, circuit, and computational-complexity characteristics of the proposed formulation. Section~\ref{sec:numerical} reports on the benchmark evaluation, factorial ablation, noise-resilience study, and optimization-landscape analysis. Section~\ref{conclusion} concludes the paper and outlines future research directions.

\section{Problem Formulation}
\label{formulation}

\subsection{Controlled-Islanding Setting}

Consider a physical power system represented by a weighted undirected graph
$G^0=(V^0,E^0,w^0)$, where $V^0$ is the bus set, $E^0$ is the transmission-line
set, and $w^0_{ij}$ is the disruption cost of line $(i,j)\in E^0$.  Let
$N_0=|V^0|$.  The bus set is partitioned into generation-capable buses $V_G^0$,
load buses $V_L^0$, and transit buses $V_T^0$.  Following standard
slow-coherency-based islanding practice~\cite{chow1982time,you2004slow}, the
coherency analysis predetermines a family of nonempty, pairwise-disjoint
coherent anchor groups
$\mathcal{P}_{\mathrm{coh}}^0=\{\mathcal{C}_1^0,\ldots,\mathcal{C}_K^0\}$,
where
\begin{equation}
K=\left|\mathcal{P}_{\mathrm{coh}}^0\right|
\label{eq:coherent_group_count}
\end{equation}
is the target number of islands.  Each anchor group contains at least one
generation-capable bus,
\begin{equation}
\mathcal{C}_c^0\cap V_G^0\neq\varnothing,\qquad c=1,\ldots,K.
\label{eq:anchor_contains_generator}
\end{equation}
For a transmission interface, the line weight is defined from the pre-islanding
directional active-power flows as
\begin{equation}
w^0_{ij}=\frac{|p_{ij}|+|p_{ji}|}{2},
\label{eq:physical_edge_weight}
\end{equation}
so cutting a high-weight edge corresponds to interrupting a larger power exchange.
This physical-network representation and the coherent anchor groups are
illustrated in Fig.~\ref{fig:overall_diagram}(1).

\subsection{Graph Reduction via Generator Coherency}
\label{sec:reduction}

Since all buses within the same anchor group $\mathcal{C}_c^0$ must be
co-assigned to a single island, any connected subgroup can be collapsed into one
super-bus without losing information about the partition.  This observation
motivates a graph reduction step that shrinks the optimization domain before any
binary variable is introduced.  A working graph $G=(V,E,w)$ is therefore formed
by contracting every anchor group whose induced subgraph in $G^0$ is connected.
This contraction into anchor super-buses is summarized in
Fig.~\ref{fig:overall_diagram}(2).
Disconnected groups are retained as separate working-graph buses so that a
connected island on the working graph always expands to a connected physical
island.  Let $\mathcal{M}\subseteq\{1,\ldots,K\}$ be the index set of contracted
anchor groups and $\phi:V^0\rightarrow V$ map each physical bus to its
working-graph representative.  The working-graph vertex set and size are
\begin{align}
V
&=
\Bigl(V^0\setminus\bigcup_{c\in\mathcal{M}}\mathcal{C}_c^0\Bigr)
\cup\{s_c:c\in\mathcal{M}\},
\notag\\
N
&=|V|=N_0-\sum_{c\in\mathcal{M}}\bigl(|\mathcal{C}_c^0|-1\bigr),
\label{eq:mandatory_merge_vertices}
\end{align}
and the working-graph edge weights aggregate all physical interfaces between two
representatives,
\begin{equation}
w_{ab}=
\sum_{\substack{u\in\phi^{-1}(a),\,v\in\phi^{-1}(b)\\(u,v)\in E^0}}
w^0_{uv}.
\label{eq:mandatory_merge_weight}
\end{equation}
The working-graph generator and load sets are
\begin{align}
V_G&=\{a\in V:\phi^{-1}(a)\cap V_G^0\neq\varnothing\},
\notag\\
V_L&=\{a\in V:\phi^{-1}(a)\cap V_L^0\neq\varnothing\}.
\label{eq:merged_gen_load_sets}
\end{align}
For each anchor group $c$, its image in the working graph is
\begin{equation}
\mathcal{A}_c=\{a\in V:\phi^{-1}(a)\cap\mathcal{C}_c^0\neq\varnothing\},
\qquad c=1,\ldots,K.
\label{eq:working_anchor_groups}
\end{equation}
A working-graph island $V_k\subseteq V$ expands to the physical island
\begin{equation}
\widetilde{V}_k=\phi^{-1}(V_k)=\bigcup_{a\in V_k}\phi^{-1}(a),
\label{eq:expanded_island}
\end{equation}
which is used for final physical validation.

\begin{lemma}[Cut preservation]
\label{lem:cut_preserve}
For any partition $\pi:V\rightarrow\{1,\ldots,K\}$, the working-graph cut
equals the cut of its expanded physical partition
$\tilde{\pi}(u)=\pi(\phi(u))$. The proof is given in~\cite{jiang2026regrid}.
\end{lemma}

\subsection{Working-Graph Islanding Formulation}
\label{sec:working_graph_formulation}

Let $\pi:V\rightarrow\{1,\ldots,K\}$ be an island assignment and
$V_k=\{i\in V:\pi(i)=k\}$.  Let
$\chi_{\mathrm{conn}}(G[S])$ be the connectivity indicator of the induced
subgraph $G[S]$, equal to 1 if $G[S]$ is connected and 0 otherwise.  The
controlled-islanding problem is
\begin{align}
\min_{\pi}\quad
&\sum_{(i,j)\in E}w_{ij}\,\mathbf{1}[\pi(i)\neq\pi(j)]
\label{eq:islanding_objective}\\
\text{s.t.}\quad
&\bigcup_{k=1}^{K}V_k=V,\qquad
V_k\cap V_\ell=\varnothing,\quad 1\leq k<\ell\leq K,
\label{con:onehot}\\
&|V_k|\geq N_{\min},\qquad k=1,\ldots,K,
\label{con:min_size}\\
&V_k\cap V_G\neq\varnothing,\quad V_k\cap V_L\neq\varnothing,
\qquad k=1,\ldots,K,
\label{con:gen_load}\\
&\pi(a)=\pi(a'),\quad a,a'\in\mathcal{A}_c,\quad c=1,\ldots,K,
\label{con:same_group}\\
&\pi(a)\neq\pi(b),\quad
a\in\mathcal{A}_c,\ b\in\mathcal{A}_d,\ c\neq d,
\label{con:different_group}\\
&\chi_{\mathrm{conn}}(G[V_k])=1,\qquad k=1,\ldots,K.
\label{con:connect}
\end{align}
Constraint~\eqref{con:onehot} enforces complete coverage and pairwise disjointness;
\eqref{con:min_size} avoids degenerate islands; \eqref{con:gen_load} ensures each
island has generation and load; \eqref{con:same_group} and
\eqref{con:different_group} enforce coherency; and \eqref{con:connect} requires
$G[V_k]$ to be connected, because the connectivity indicator equals 1 exactly
for connected induced subgraphs.
Because contracted anchor groups preserve the cut
value under~\eqref{eq:mandatory_merge_weight}, the cut in~\eqref{eq:islanding_objective}
equals the physical cut of $\{\widetilde{V}_k\}$.  This combinatorial problem is
NP-hard because it contains weighted minimum bisection as a simpler special case:
when only two equal-size islands are required, the task is exactly to minimize
the cut weight between them.

The graph-partitioning objective serves as a proxy for post-islanding power-flow
feasibility because minimizing severed active power reduces the inter-area
transfer that must be redistributed after separation, but it does not by itself
guarantee that each island can sustain a consistent operating point.  Final
partitions are therefore physically validated using power-flow analysis.

To embed the problem in a QAOA
implementation, the assignment encoding determines how the islanding structure
is represented by qubits and Hamiltonian terms.  The full one-hot register
treats all island labels and all buses as
independent decisions~\cite{jiang2026regrid}, although one label is implied by the partition constraint
and the coherent anchor labels are already fixed.  This paper proposes two
assignment encodings, $\mathcal{E}_1$ and
$\mathcal{E}_2$, that remove these two redundancies while preserving the
feasible partitions.
The two register reductions are compared visually in
Fig.~\ref{fig:overall_diagram}(3).
The two encodings are introduced in Sections~\ref{sec:reduced_encoding}
and~\ref{sec:sf_encoding}, respectively.

\subsection{Encoding $\mathcal{E}_1$: Reduced Assignment Encoding}
\label{sec:reduced_encoding}

Existing QAOA-based islanding formulations~\cite{jiang2026regrid} adopt a full
one-hot assignment encoding on the working graph, introducing $K$ binary
variables per bus, one per island label.  This yields $KN$ binary variables
in total and requires an explicit one-hot equality penalty in the QUBO.
Two structural observations about the islanding problem permit a systematic
reduction of this variable count.  First, the $K$-th island label carries no
independent information once the first $K-1$ are fixed, since every bus must
belong to exactly one island.  Second, the coherent anchor constraints already
determine the island label of every anchor bus without any binary search, so
allocating variables to those buses is redundant.  Exploiting these two
observations in sequence yields two encoding reductions proposed in this paper,
denoted $\mathcal{E}_1$ and $\mathcal{E}_2$.
Section~\ref{sec:reduced_encoding} introduces $\mathcal{E}_1$, the reduced
assignment encoding, which removes the per-bus label redundancy and reduces
the register from $KN$ to $(K-1)N$ variables.
Section~\ref{sec:sf_encoding} introduces $\mathcal{E}_2$, the symmetry-fixed
assignment encoding, which further eliminates anchor-bus variables and reduces
the register to $(K-1)(N-|\mathcal{A}|)$ variables.

$\mathcal{E}_1$ eliminates the one-hot redundancy by representing only the
first $K-1$ island labels explicitly.  For each working-graph bus $i\in V$,
introduce
\begin{equation}
y_{i,k}\in\{0,1\},\qquad k=1,\ldots,K-1,
\label{eq:reduced_explicit_vars}
\end{equation}
and define the island-membership indicator
\begin{equation}
Y_{i,k}=
\begin{cases}
y_{i,k}, & k=1,\ldots,K-1,\\[1mm]
1-\displaystyle\sum_{\ell=1}^{K-1}y_{i,\ell}, & k=K.
\end{cases}
\label{eq:reduced_inferred_var}
\end{equation}
Bus $i$ is assigned to island $k$ when $Y_{i,k}=1$ and to no other island.
The all-zero explicit assignment $y_{i,1}=\cdots=y_{i,K-1}=0$ represents
membership in island $K$, so island $K$ is the implicit island and needs
no explicit variable.

By construction, the $K$ indicators sum to one,
\begin{equation}
\sum_{k=1}^{K}Y_{i,k}
=\sum_{k=1}^{K-1}y_{i,k}+\Bigl(1-\sum_{\ell=1}^{K-1}y_{i,\ell}\Bigr)=1,
\qquad i\in V,
\label{eq:reduced_sum_to_one}
\end{equation}
so the partition equality in~\eqref{con:onehot} is satisfied identically for
every bitstring.  This removes the need for a one-hot equality penalty in the
QUBO.

For the indicator~\eqref{eq:reduced_inferred_var} to represent a valid
single-island assignment, the inferred value $Y_{i,K}$ must remain in
$\{0,1\}$.  This requires
\begin{equation}
\sum_{k=1}^{K-1}y_{i,k}\leq 1,\qquad i\in V.
\label{con:reduced_validity}
\end{equation}
If more than one explicit variable for bus $i$ equals one, then
$Y_{i,K}=1-\sum_{\ell}y_{i,\ell}$ becomes negative, which is not a valid
indicator.  The corresponding uniqueness QUBO term is
\begin{equation}
H_U
=\sum_{i\in V}\sum_{1\le k<\ell\le K-1}y_{i,k}y_{i,\ell},
\label{eq:reduced_encoding_penalty_def}
\end{equation}
which is nonzero exactly when two explicit labels are selected for the same bus.

The cut objective in~\eqref{eq:islanding_objective} can be rewritten using
$Y_{i,k}$.  Because $\mathbf{1}[\pi(i)\neq\pi(j)]=1-\sum_{k=1}^{K}Y_{i,k}Y_{j,k}$
when both indicators are valid, the objective becomes
\begin{equation}
\sum_{(i,j)\in E}w_{ij}\Bigl(1-\sum_{k=1}^{K}Y_{i,k}Y_{j,k}\Bigr).
\label{eq:reduced_cut_full}
\end{equation}
Dropping the edge-weight constant $\sum_{(i,j)\in E}w_{ij}$, which is
independent of the assignment, gives the cut Hamiltonian
\begin{equation}
H_f
=-\sum_{(i,j)\in E}w_{ij}\sum_{k=1}^{K}Y_{i,k}Y_{j,k}.
\label{eq:reduced_cut_qubo}
\end{equation}
Since each $Y_{i,k}$ is affine in $\{y_{i,k}\}$ and $y_{i,k}^2=y_{i,k}$ for
binary variables, every term $Y_{i,k}Y_{j,k}$ expands to degree at most two, so
$H_f$ is a valid QUBO.

The assignment register contains
\begin{equation}
n_q^{\mathcal{E}_1}=(K-1)N
\label{eq:reduced_qubit_count}
\end{equation}
binary variables.  Compared with the full one-hot encoding on the working graph
($KN$ variables) the saving is $N$ variables from eliminating one label per bus.
Compared with the full one-hot encoding on the physical graph ($KN_0$ variables),
the total saving is
\begin{equation}
\Delta_Q^{\mathcal{E}_1}
=KN_0-(K-1)N
=N_0+(K-1)\sum_{c\in\mathcal{M}}\bigl(|\mathcal{C}_c^0|-1\bigr),
\label{eq:re_total_saving}
\end{equation}
where the second term counts the variables eliminated by graph contraction.

This encoding is general with respect to island label assignments: it still
represents all $K!$ label-permutation equivalents of each feasible physical
partition.  The following encoding eliminates this residual symmetry.

\subsection{Encoding $\mathcal{E}_2$: Symmetry-Fixed Assignment Encoding}
\label{sec:sf_encoding}

\subsubsection{Label symmetry and its source}

$\mathcal{E}_1$ still contains a structural redundancy.  The
controlled-islanding objective~\eqref{eq:islanding_objective} and every
constraint~\eqref{con:onehot}--\eqref{con:connect} are invariant under a common
relabeling of the $K$ island indices: if $\pi$ is feasible with cut value $c^*$
and $\sigma:\{1,\ldots,K\}\rightarrow\{1,\ldots,K\}$ is any permutation, then
$\sigma\circ\pi$ has the same physical partition, the same cut value $c^*$, and
the same feasibility status.  Consequently, the assignment register of the
$\mathcal{E}_1$ holds $K!$ distinct bitstrings that all correspond to the same
unlabeled island partition.

The coherent anchor constraints~\eqref{con:same_group}
and~\eqref{con:different_group} provide a natural mechanism to break this
symmetry.  Constraint~\eqref{con:same_group} forces every bus in anchor group
$\mathcal{A}_c$ to the same island, and constraint~\eqref{con:different_group}
forces buses from different anchor groups to different islands.  Together they
impose a bijection between the $K$ anchor groups and the $K$ island labels.  Once
that bijection is fixed, no further label permutation is consistent with the
constraints.  The following proposition makes this precise.

\begin{proposition}[Canonical coherent-group labeling]
\label{prop:canonical_labeling}
Suppose the $K$ working-graph anchor groups $\mathcal{A}_1,\ldots,\mathcal{A}_K$
are pairwise disjoint and nonempty, there are $K$ target islands, and the
objective and feasibility constraints are invariant under a common permutation of
island labels.  Every feasible assignment $\pi$ has a label-relabeled
representative $\sigma\circ\pi$ satisfying
\begin{equation}
(\sigma\circ\pi)(a)=c,\qquad a\in\mathcal{A}_c,\quad c=1,\ldots,K.
\label{eq:canonical_anchor_assignment}
\end{equation}
This representative has the same cut value and satisfies all constraints.
\end{proposition}

\begin{proof}
For each anchor group $\mathcal{A}_c$, constraint~\eqref{con:same_group}
implies that there exists an island label $k_c\in\{1,\ldots,K\}$ such that
\begin{equation}
\pi(a)=k_c,\qquad a\in\mathcal{A}_c.
\label{eq:anchor_group_label}
\end{equation}
Constraint~\eqref{con:different_group} gives
\begin{equation}
c\neq d \quad\Longrightarrow\quad k_c\neq k_d .
\label{eq:anchor_group_labels_distinct}
\end{equation}
Hence the map
\begin{equation}
\kappa:\{1,\ldots,K\}\rightarrow\{1,\ldots,K\},\qquad
\kappa(c)=k_c
\end{equation}
is injective and therefore bijective.  Let $\sigma=\kappa^{-1}$, equivalently
\begin{equation}
\sigma(k_c)=c,\qquad c=1,\ldots,K.
\label{eq:anchor_inverse_permutation}
\end{equation}
Combining~\eqref{eq:anchor_group_label} and~\eqref{eq:anchor_inverse_permutation}
gives
\begin{equation}
(\sigma\circ\pi)(a)=\sigma(k_c)=c,\qquad
a\in\mathcal{A}_c,\quad c=1,\ldots,K,
\end{equation}
which proves~\eqref{eq:canonical_anchor_assignment}.  Since the objective and
all constraints are label-invariant, $\sigma\circ\pi$ is also feasible with
the same cut value.
\end{proof}

\subsubsection{Variable elimination via fixed anchor assignments}

Let
\begin{equation}
\mathcal{A}=\bigcup_{c=1}^{K}\mathcal{A}_c,\qquad
F=V\setminus\mathcal{A}
\label{eq:free_bus_set}
\end{equation}
denote the fixed-anchor and free-bus sets, with $|\mathcal{A}|$ anchor buses and
$|F|=N-|\mathcal{A}|$ free buses.  Proposition~\ref{prop:canonical_labeling}
shows that restricting to canonical representatives pins the island assignment of
every anchor bus to a known constant: bus $a\in\mathcal{A}_c$ is always in island
$c$.  No binary variable is needed to represent this assignment.  Binary variables
are therefore allocated only for free buses,
\begin{equation}
y_{i,k}\in\{0,1\},\qquad i\in F,\quad k=1,\ldots,K-1,
\label{eq:sf_explicit_vars}
\end{equation}
with the same reduced structure as~\eqref{eq:reduced_explicit_vars}: island $K$
remains implicit for free buses too.

The unified island-membership indicator is
\begin{equation}
Y^{\mathcal{E}_2}_{i,k}=
\begin{cases}
\mathbf{1}[k=c], & i\in\mathcal{A}_c,\\[1mm]
y_{i,k}, & i\in F,\ k=1,\ldots,K-1,\\[1mm]
1-\displaystyle\sum_{\ell=1}^{K-1}y_{i,\ell},
  & i\in F,\ k=K.
\end{cases}
\label{eq:sf_membership_indicator}
\end{equation}
For an anchor bus $a\in\mathcal{A}_c$, the indicator evaluates to
$Y^{\mathcal{E}_2}_{a,k}=\mathbf{1}[k=c]$, which is 1 for $k=c$ and 0 otherwise.
The sum-to-one property holds for all buses:
\begin{equation}
\sum_{k=1}^{K}Y^{\mathcal{E}_2}_{i,k}=1,\qquad i\in V,
\label{eq:sf_sum_to_one}
\end{equation}
because $\sum_{k=1}^{K}\mathbf{1}[k=c]=1$ for anchor buses, and the free-bus
case follows identically to~\eqref{eq:reduced_sum_to_one}.

\subsubsection{Structural satisfaction of constraints}

Substituting $Y^{\mathcal{E}_2}$ satisfies the coherency constraints and the
generator-presence part of~\eqref{con:gen_load} structurally, while leaving the
encoding-validity condition active only on free buses.

\textit{Coherency (constraints~\eqref{con:same_group}
and~\eqref{con:different_group}).}  For any two buses $a,a'\in\mathcal{A}_c$,
the indicator gives $Y^{\mathcal{E}_2}_{a,k}=Y^{\mathcal{E}_2}_{a',k}=\mathbf{1}[k=c]$,
so both are assigned to island $c$: constraint~\eqref{con:same_group} is
satisfied.  For $a\in\mathcal{A}_c$ and $b\in\mathcal{A}_d$ with $c\neq d$,
$Y^{\mathcal{E}_2}_{a,k}=\mathbf{1}[k=c]$ and $Y^{\mathcal{E}_2}_{b,k}=\mathbf{1}[k=d]$
assign $a$ and $b$ to different islands: constraint~\eqref{con:different_group} is
satisfied.  No coherency QUBO penalty is required.

\textit{Generator presence in~\eqref{con:gen_load}.}  By
assumption~\eqref{eq:anchor_contains_generator}, each physical anchor group
$\mathcal{C}_c^0$ contains at least one generator bus $g\in V_G^0$.  Its
working-graph representative $\phi(g)\in\mathcal{A}_c$ is a generator bus in
$V_G$.  Since every bus in $\mathcal{A}_c$ is fixed to island $c$, island $c$
always contains $\phi(g)\in V_G$: the generator-presence condition is structural.

\textit{Encoding validity for anchor buses.}  For anchor bus $a\in\mathcal{A}_c$,
\begin{equation}
\sum_{k=1}^{K-1}Y^{\mathcal{E}_2}_{a,k}
=\sum_{k=1}^{K-1}\mathbf{1}[k=c]
=\mathbf{1}[c\leq K-1]\leq 1,
\label{eq:sf_anchor_validity_auto}
\end{equation}
which is always at most one.  The validity condition is therefore automatically
satisfied for every anchor bus; it remains an active constraint only for free
buses:
\begin{equation}
\sum_{k=1}^{K-1}y_{i,k}\leq1,\qquad i\in F.
\label{con:sf_validity}
\end{equation}

\subsubsection{Cut Hamiltonian after substitution}

Substituting $Y^{\mathcal{E}_2}$ into the cut Hamiltonian
\begin{equation}
H_{\mathrm{cut}}^{\mathcal{E}_2}
=-\sum_{(i,j)\in E} w_{ij}
  \sum_{k=1}^{K}Y^{\mathcal{E}_2}_{i,k}Y^{\mathcal{E}_2}_{j,k}
\label{eq:sf_cut_hamiltonian}
\end{equation}
separates its terms into three edge classes. A fixed-anchor endpoint belongs to
some $\mathcal{A}_c$ and therefore has island label $c$. A free-bus endpoint
belongs to $F$, and binary variables represent its label. For later use, define
\begin{align}
E_{\mathcal{A}F}^{c}
&=\{(a,f)\in E:\ a\in\mathcal{A}_c,\ f\in F\}, \notag\\
E_{FF}
&=\{(i,j)\in E:\ i,j\in F\},
\label{eq:sf_edge_sets}
\end{align}
where each edge in $E_{\mathcal{A}F}^{c}$ is written with its fixed-anchor
endpoint first.

\textit{Anchor--anchor edges} $(i,j)\in E$ with $i\in\mathcal{A}_c$,
$j\in\mathcal{A}_d$:
\begin{equation}
\sum_{k=1}^{K}Y^{\mathcal{E}_2}_{i,k}Y^{\mathcal{E}_2}_{j,k}
=\sum_{k=1}^{K}\mathbf{1}[k=c]\,\mathbf{1}[k=d]
=\mathbf{1}[c=d].
\label{eq:sf_anchor_anchor_term}
\end{equation}
This is a constant: 1 if both endpoints belong to the same anchor group
(intra-group edge, never cut) and 0 if they belong to different anchor groups
(always cut, contributing $w_{ij}$ as a fixed additive constant).  Both
cases are variable-independent and can be dropped from the QUBO without affecting
the optimization.

\textit{Edges with one fixed-anchor endpoint and one free-bus endpoint.}
For $(a,f)\in E_{\mathcal{A}F}^{c}$,
\begin{equation}
\sum_{k=1}^{K}Y^{\mathcal{E}_2}_{a,k}Y^{\mathcal{E}_2}_{f,k}
=\sum_{k=1}^{K}\mathbf{1}[k=c]\,Y^{\mathcal{E}_2}_{f,k}
=Y^{\mathcal{E}_2}_{f,c},
\label{eq:sf_anchor_free_term}
\end{equation}
which is linear in the free-bus variables because
$Y^{\mathcal{E}_2}_{f,c}=y_{f,c}$ for $c<K$ and
$Y^{\mathcal{E}_2}_{f,K}=1-\sum_{\ell=1}^{K-1}y_{f,\ell}$.

\textit{Edges with two free-bus endpoints.}  For $(i,j)\in E_{FF}$, the term
$\sum_{k=1}^{K}Y^{\mathcal{E}_2}_{i,k}Y^{\mathcal{E}_2}_{j,k}$ is quadratic in
the free-bus variables, exactly as in $\mathcal{E}_1$.

Therefore, the symmetry-fixed cut Hamiltonian is
\begin{align}
H_f
={}&
-\sum_{c=1}^{K}\sum_{(a,f)\in E_{\mathcal{A}F}^{c}}
  w_{af}\,Y^{\mathcal{E}_2}_{f,c}
\notag\\
&-\sum_{(i,j)\in E_{FF}}
  w_{ij}\sum_{k=1}^{K}Y^{\mathcal{E}_2}_{i,k}Y^{\mathcal{E}_2}_{j,k}
\label{eq:sf_cut_qubo}
\end{align}
after dropping the constant anchor--anchor contributions.  This is a valid QUBO
over the free-bus variables only, with degree at most two after applying
$y_{i,k}^2=y_{i,k}$.

\subsubsection{Qubit count and total saving}

The assignment register contains
\begin{equation}
n_q^{\mathcal{E}_2}
=(K-1)|F|
=(K-1)(N-|\mathcal{A}|)
\label{eq:sf_qubit_count}
\end{equation}
binary variables.  Relative to $\mathcal{E}_1$, the additional
saving from fixing anchor labels is
\begin{equation}
\Delta n_q
=n_q^{\mathcal{E}_1}-n_q^{\mathcal{E}_2}
=(K-1)|\mathcal{A}|,
\label{eq:sf_extra_saving}
\end{equation}
one variable per anchor bus per explicit label.  The total saving relative to
full one-hot encoding on the physical graph is
\begin{equation}
\Delta_Q^{\mathcal{E}_2}
=KN_0-n_q^{\mathcal{E}_2}
=KN_0-(K-1)(N-|\mathcal{A}|).
\label{eq:sf_total_saving}
\end{equation}

\subsubsection{Exactness qualification}

Proposition~\ref{prop:canonical_labeling} shows that every feasible islanding
partition has a label-relabeled representative satisfying the anchor-fixed
rules of $\mathcal{E}_2$.  Theorem~\ref{thm:encoding_hierarchy} in
Appendix~\ref{app:encoding_hierarchy} formalizes the consequence:
$\mathcal{E}_2$ preserves the feasible solutions and optimum of
$\mathcal{E}_1$ while representing each feasible unlabeled partition once rather
than $K!$ times.

\begin{figure*}[!t]
\centering
\includegraphics[width=0.92\textwidth]{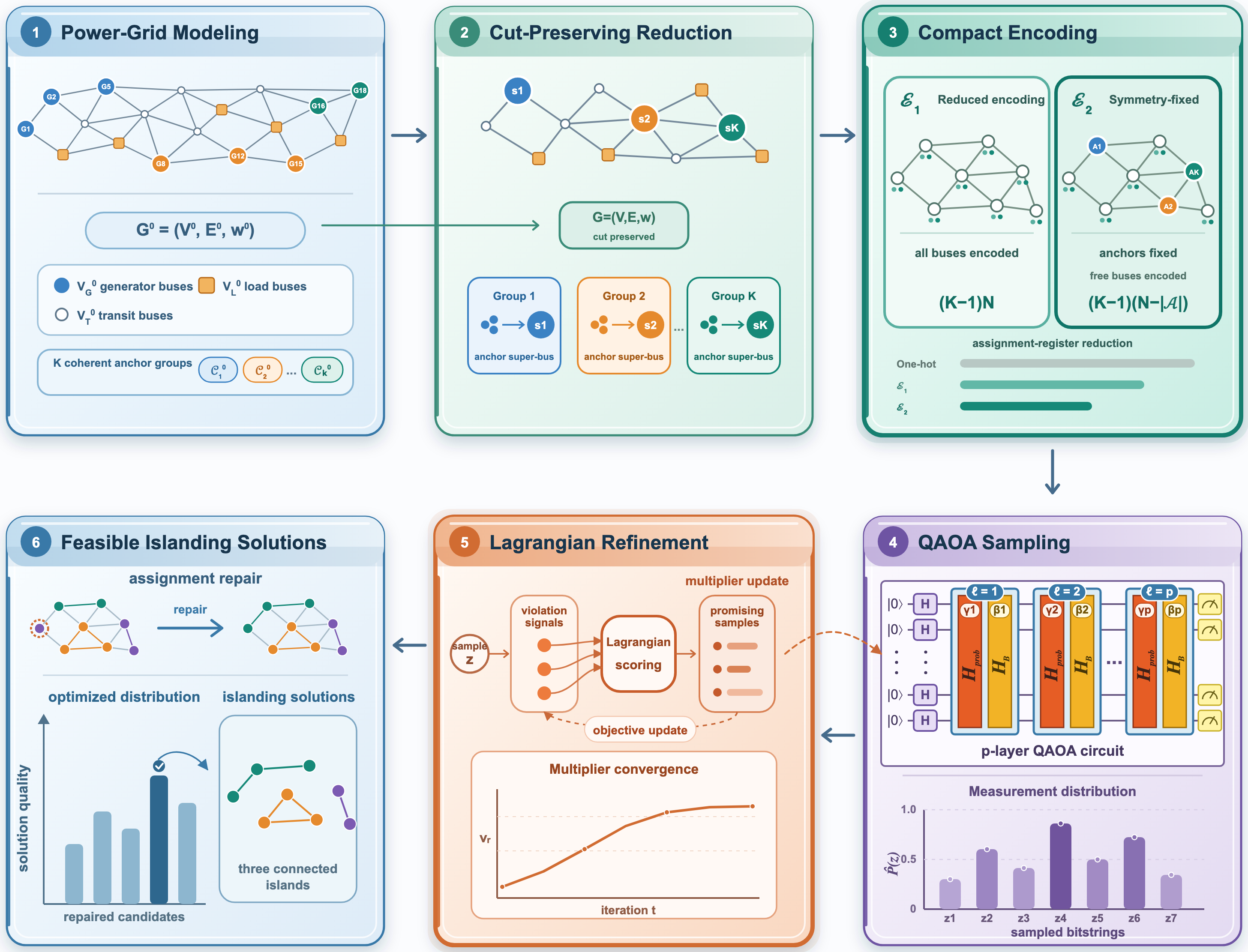}
\caption[Overall PACE-QAOA framework]{Overall PACE-QAOA workflow for controlled
power-system islanding. (1)~The original grid is modeled as a weighted graph with
generator, load, and transit buses and $K$ coherent anchor groups.
(2)~Each coherent group is contracted into an anchor super-bus to form a
cut-preserving working graph. (3)~The compact encodings $\mathcal{E}_1$ and
$\mathcal{E}_2$ represent bus-to-island assignments, with $\mathcal{E}_2$ fixing
anchor labels to further reduce the assignment register. (4)~A $p$-layer QAOA
circuit samples candidate assignments. (5)~Lagrangian refinement uses constraint
violations to score promising samples and update the multipliers.
(6)~Assignment repair converts the selected samples into feasible connected
islands. Blocks~3 and~5 highlight the proposed compact encoding and Lagrangian
refinement.}
\label{fig:overall_diagram}
\end{figure*}

\section{QAOA Solutions for Power System Islanding}
\label{quantum}

This section presents the quantum formulation on the working graph defined in
Section~\ref{sec:reduction}.  The solver combines the compact assignment
encoding with modular QUBO construction and a slack-free Lagrangian treatment
of selected inequalities evaluated on measured samples.

\subsection{QUBO Hamiltonian Construction}
\label{sec:reduced_qubo}

Let $\xi\in\{\mathcal{E}_1,\mathcal{E}_2\}$ denote the selected encoding, and
let
\begin{equation}
\mathcal{B}_{\xi}
=
\begin{cases}
V, & \xi=\mathcal{E}_1,\\
F, & \xi=\mathcal{E}_2,
\end{cases}
\label{eq:qaoa_variable_bus_set}
\end{equation}
be the set of buses whose island labels are represented by explicit binary
variables.  The QAOA register allocates variables
$y_{i,k}$ for $i\in\mathcal{B}_{\xi}$ and $k=1,\ldots,K-1$.  All QUBO terms are
formed from the decoded island-membership indicator
\begin{equation}
\widehat{Y}_{i,k} =
\begin{cases}
\mathbf{1}[k=c],
  & \xi=\mathcal{E}_2,\ i\in\mathcal{A}_c,\\[1mm]
y_{i,k},
  & i\in\mathcal{B}_{\xi},\ k=1,\ldots,K-1,\\[1mm]
1-\displaystyle\sum_{\ell=1}^{K-1}y_{i,\ell},
  & i\in\mathcal{B}_{\xi},\ k=K.
\end{cases}
\label{eq:qaoa_decoded_indicator}
\end{equation}
For $\mathcal{E}_1$, this reduces to the inferred-island indicator on all buses.
For $\mathcal{E}_2$, anchor-bus indicators are constants and only free-bus
labels remain variable.

The constraints converted to QUBO terms are written in terms of
$\widehat{Y}$.  The at-most-one uniqueness penalty for~\eqref{con:onehot} is
\begin{equation}
H_U
=
\sum_{i\in\mathcal{B}_{\xi}}
\sum_{1\le k<\ell\le K-1}
y_{i,k}y_{i,\ell}.
\label{eq:qubo_onehot}
\end{equation}
This term acts only on buses in $\mathcal{B}_{\xi}$.  For $K=2$, no pair of
explicit island variables exists, so $H_U=0$.

The cut objective~\eqref{eq:islanding_objective} is equivalent to
$\sum_{(i,j)\in E}w_{ij}(1-\sum_{k=1}^{K}\widehat{Y}_{i,k}
\widehat{Y}_{j,k})$ under valid assignments.
Since the additive constant $\sum_{(i,j)\in E}w_{ij}$ is independent of the
partition, the cut Hamiltonian takes the compact QUBO form
\begin{equation}
H_f
=
-\sum_{(i,j)\in E}
w_{ij}\sum_{k=1}^{K}\widehat{Y}_{i,k}\widehat{Y}_{j,k}.
\label{eq:qubo_cut_reduced}
\end{equation}
Under $\mathcal{E}_2$, fixed-anchor products reduce to constants or linear
terms in the free-bus variables. As each $\widehat{Y}_{i,k}$ is affine
in the explicit binary variables or is a constant, every bilinear product
remains at most quadratic after the binary identity $y^2=y$, so $H_f$
constitutes a valid QUBO under both encodings. The physical cut
interpretation of~\eqref{eq:qubo_cut_reduced} holds whenever the assignment
satisfies the validity condition on $\mathcal{B}_{\xi}$; a bitstring with
$\sum_{k=1}^{K-1}y_{i,k}>1$ for some $i\in\mathcal{B}_{\xi}$ yields an
algebraic rather than physical island indicator, and such states are
penalized by $\lambda_U H_U$ and corrected in post-processing via the
uniqueness violation~\eqref{eq:v_encoding_raw}.

For $\mathcal{E}_1$, the working-graph coherency component is
\begin{align}
H_G
={}&
\sum_{c=1}^{K}
\sum_{k=1}^{K}
\sum_{\substack{g,g'\in V_G^{(c)}\\g<g'}}
\left(\widehat{Y}_{g,k}
-\widehat{Y}_{g',k}\right)^2 \notag\\
&{}+
\sum_{k=1}^{K}
\sum_{\substack{c,c'=1\\c<c'}}^{K}
\sum_{g\in V_G^{(c)}}
\sum_{g'\in V_G^{(c')}}
\widehat{Y}_{g,k}
\widehat{Y}_{g',k}.
\label{eq:qubo_coh}
\end{align}
The first term keeps generators in the same coherent group together, and the
second prevents different coherent groups from occupying the same island.  For
$\mathcal{E}_2$, set
\begin{equation}
H_G=0,
\label{eq:qubo_coh_e2}
\end{equation}
since the fixed-anchor decoding already satisfies the coherency constraints.

The connectivity requirement~\eqref{con:connect} demands that each island
induce a connected subgraph, a condition that is not expressible as a
polynomial function of the binary assignment variables and therefore cannot
be embedded directly as a QUBO term.
To bias the phase separator toward connected partitions without introducing
auxiliary connectivity variables, the following local polynomial surrogate
is adopted:
\begin{equation}
H_c
=
\sum_{k=1}^{K}\sum_{i\in V}
\left(
\widehat{Y}_{i,k}
-
\sum_{j\in\mathcal{N}(i)}
\widehat{Y}_{i,k}\widehat{Y}_{j,k}
\right),
\label{eq:H_iso}
\end{equation}
where $\mathcal{N}(i)$ denotes the neighbor set of bus $i$.
Each summand is strictly positive when bus $i$ is assigned to island $k$
but none of its neighbors are, so $H_c$ penalizes isolated bus assignments
while remaining at most quadratic in the decision variables.
Satisfying $H_c = 0$ is only a necessary, not sufficient, condition
for~\eqref{con:connect}; accordingly, an exact DFS connectivity check is
applied to every measured sample as described in
Section~\ref{sec:constraint_placement}.

Thus, the QUBO constraint component used in the phase separator is
\begin{equation}
H_P
=\lambda_U H_U
{}+\alpha_G H_G
{}+\alpha_c H_c,
\label{eq:qubo_constraint_component}
\end{equation}
and the full problem Hamiltonian passed to the circuit is
\begin{equation}
H_{\mathrm{prob}}
=
H_f
{}+H_P,
\label{eq:H_prob}
\end{equation}
where $H_G=0$ under $\mathcal{E}_2$.  The assignment-qubit count is
\begin{equation}
n_q
=(K-1)|\mathcal{B}_{\xi}|
=
\begin{cases}
(K-1)N, & \xi=\mathcal{E}_1,\\
(K-1)(N-|\mathcal{A}|), & \xi=\mathcal{E}_2.
\end{cases}
\label{eq:qubit_count_reduced}
\end{equation}
Constraints~\eqref{con:min_size} and~\eqref{con:gen_load} are absent from
$H_{\mathrm{prob}}$; they involve inequality bounds on island cardinalities
whose QUBO embedding would require per-island binary slack registers, adding
$O(K\log N)$ auxiliary qubits and dense cross-variable coupling.  Instead,
they are evaluated exactly on every measured sample and enforced through the
slack-free Lagrangian formulation detailed in
Section~\ref{sec:constraint_placement}.

\subsection{Slack-Free Lagrangian Constraint Placement}
\label{sec:constraint_placement}

To preserve qubit efficiency, constraints~\eqref{con:min_size}
and~\eqref{con:gen_load} are not embedded in the QUBO via binary slack
registers. Instead, their satisfaction is assessed directly on every
measured sample and enforced through a classical augmented Lagrangian
that updates penalty multipliers between quantum optimization iterations, as
depicted in Fig.~\ref{fig:overall_diagram}(5).
This separation keeps the quantum register free of slack variables, so
the circuit encodes only the assignment variables of the selected
encoding and all inequality enforcement remains entirely classical.

Each measured bitstring $z\in\{0,1\}^{n_q}$ is decoded into island
membership indicators according to~\eqref{eq:qaoa_decoded_indicator}.
Island $K$ corresponds to the all-zero explicit assignment on
$\mathcal{B}_{\xi}$. A bitstring with more than one explicit island bit
equal to one for the same bus yields an invalid inferred indicator and
constitutes a multi-assignment violation. The island cardinalities used by the Lagrangian checks are
\begin{align}
n_k(z)&=\sum_{i\in V}\widehat{Y}_{i,k}(z),&
g_k(z)&=\sum_{i\in V_G}\widehat{Y}_{i,k}(z), \notag\\
\ell_k(z)&=\sum_{i\in V_L}\widehat{Y}_{i,k}(z).&&
\label{eq:island_cardinalities}
\end{align}
The corresponding slack-free hinge violations are
\begin{align}
v_{M,k}(z)&=\max\{0,\;N_{\min}-n_k(z)\}, \label{eq:v_minsize}\\
v_{G,k}(z)&=\max\{0,\;1-g_k(z)\}, \label{eq:v_gen}\\
v_{L,k}(z)&=\max\{0,\;1-\ell_k(z)\}. \label{eq:v_load}
\end{align}
For island $K$, these expressions use the inferred indicator in~\eqref{eq:qaoa_decoded_indicator}; no auxiliary binary variables are introduced. Aggregating over all islands gives
\begin{align}
v_M(z)&=\sum_{k=1}^{K}v_{M,k}(z),&
v_G(z)&=\sum_{k=1}^{K}v_{G,k}(z), \notag\\
v_L(z)&=\sum_{k=1}^{K}v_{L,k}(z).
\label{eq:aggregate_lag_violations}
\end{align}

\begin{proposition}[Slack-free inequality evaluation]
\label{prop:slack_free_equiv}
Let $g_r(z)\le 0$ denote any scalar inequality constraint evaluated on a decoded sample $z$, and let
$v_r(z)=\max\{0,g_r(z)\}$ be its hinge violation. Then
\begin{equation}
\label{eq:slack_free_equiv}
v_r(z)^2
=
\min_{s_r\ge 0}\left(g_r(z)+s_r\right)^2 .
\end{equation}
Thus, evaluating the hinge violation after measurement gives the same residual as the best nonnegative slack-variable representation, without allocating slack qubits in the QAOA register.
\end{proposition}
\begin{proof}
If $g_r(z)\le 0$, choosing $s_r=-g_r(z)$ makes the residual zero, which equals $v_r(z)^2$. If $g_r(z)>0$, the residual $\left(g_r(z)+s_r\right)^2$ is minimized over $s_r\ge 0$ at $s_r=0$, giving $g_r(z)^2=v_r(z)^2$. These two cases prove~\eqref{eq:slack_free_equiv}.
\end{proof}

The same raw decoded sample is also checked for reduced-assignment validity, exact DFS connectivity, and generator coherency. These violations are
\begin{align}
v_U(z)
&=
\sum_{i\in\mathcal{B}_{\xi}}
\max\left\{0,\sum_{k=1}^{K-1}y_{i,k}(z)-1\right\},
\label{eq:v_encoding_raw}\\
v_C(z)
&=
\sum_{k=1}^{K}
\max\{0,\omega(G[V_k(z)])-1\},
\label{eq:v_connectivity_raw}
\end{align}
where $\omega(G[V_k(z)])$ is the number of DFS-connected components in island $k$. Let $I_c(z)$ denote the set of islands occupied by the working-graph representatives of coherent generator group $c$. The coherency violation used by the post-checks is
\begin{equation}
v_{\mathrm{coh}}(z)
=
\sum_c \max\{0,|I_c(z)|-1\}
+
\sum_{c<c'}|I_c(z)\cap I_{c'}(z)|.
\label{eq:v_coh_raw}
\end{equation}
For $\mathcal{E}_2$, the fixed-anchor decoding makes
$v_{\mathrm{coh}}(z)=0$, and generator-presence violations are also zero under
the anchor assumption in~\eqref{eq:anchor_contains_generator}.  The raw
violation vector is therefore
\begin{equation}
\mathbf{v}(z)=
\bigl(
v_U(z),v_C(z),v_M(z),v_G(z),v_L(z),v_{\mathrm{coh}}(z)
\bigr),
\label{eq:raw_violation_vector}
\end{equation}
where components that are structural for the selected encoding are identically
zero and are omitted from the active multiplier set. In the implementation, each post-processing check returns a scaled penalty, and the Lagrangian solver divides by the effective scale to recover the unscaled violation values in~\eqref{eq:raw_violation_vector}.

The selected slack-free violations feed into an augmented Lagrangian.  Let
\[
\mathcal{R}\subseteq\{M,G,L\}
\]
index the active multiplier components of~\eqref{eq:raw_violation_vector}.
These are the only inequalities with no QUBO embedding: minimum island
size ($v_M$) from~\eqref{con:min_size} and generator and load presence
($v_G$, $v_L$) from~\eqref{con:gen_load}, with $v_G$ omitted under
$\mathcal{E}_2$ where it is identically zero.  Violations $v_U$ and
$v_{\mathrm{coh}}$ are excluded because $H_U$ and $H_G$ already penalize
them in the Hamiltonian.  Exact connectivity $v_C$ is handled by the
post-measurement repair step described later in this section, since
$H_c$ is only a sampling-bias surrogate and does not represent the exact
DFS violation.  For a measured sample $z$, the per-sample augmented
Lagrangian cost is
\begin{equation}
\label{eq:lagrangian_sample_cost}
\mathcal{L}(z;\boldsymbol{\nu})
=
c_f(z)
+
\sum_{r\in\mathcal{R}}\nu_r v_r(z)
+
\frac{\rho}{2}
\sum_{r\in\mathcal{R}}v_r(z)^2.
\end{equation}
Taking the expectation over the measurement distribution gives the
sampled Lagrangian objective
\begin{equation}
\label{eq:lagrangian_objective}
\widehat{\mathcal{L}}
(\boldsymbol{\gamma},\boldsymbol{\beta};\boldsymbol{\nu})
=
\sum_z
\widehat{P}_{\boldsymbol{\gamma},\boldsymbol{\beta}}(z)
\mathcal{L}(z;\boldsymbol{\nu}),
\end{equation}
which is returned to the outer-loop optimizer.
The empirical average violation for constraint $r$ at iteration $t$ is
\begin{equation}
\label{eq:avg_violation}
\bar v_r^{(t)}
=
\sum_z
\widehat{P}_{\boldsymbol{\gamma}^{(t)},\boldsymbol{\beta}^{(t)}}(z)
v_r(z),
\end{equation}
and the multipliers are updated by the projected gradient ascent rule
\begin{equation}
\label{eq:nu_update}
\nu_r^{(t+1)}
=
\min\left\{\nu_{\max},
\max\left(0,\nu_r^{(t)}+\eta\,\bar v_r^{(t)}\right)\right\},
\end{equation}
where $\eta>0$ is the step size and $\nu_{\max}$ caps the multiplier
magnitude to prevent numerical overflow. In the experiments,
$\nu_r^{(0)}=10$ for all active components, $\eta=2$, $\rho=1$, and
$\nu_{\max}=10^6$.
The multipliers affect only the classical sampled objective
in~\eqref{eq:lagrangian_objective} and do not introduce slack qubits or
additional phase-separator terms.

\begin{proposition}[Multiplier fixed-point characterization]
\label{prop:multiplier_fp}
Let $\nu_r^*\in[0,\nu_{\max}]$ be a fixed point of~\eqref{eq:nu_update}
and let $\bar v_r^*$ be the corresponding average violation.
Then either $\bar v_r^*=0$ or $\nu_r^*=\nu_{\max}$.
Moreover, if $\bar v_r^{(t)}\ge\delta>0$ for some window of iterations,
$\nu_r^{(t)}$ increases by at least $\eta\delta$ per step until
$\nu_{\max}$ is reached.
\end{proposition}
\begin{proof}
At a fixed point $\nu_r^{(t+1)}=\nu_r^{(t)}=\nu_r^*$.
If $\nu_r^*\in(0,\nu_{\max})$ the projection in~\eqref{eq:nu_update} is
inactive, so $\nu_r^*=\nu_r^*+\eta\bar v_r^*$, giving $\bar v_r^*=0$.
If $\nu_r^*=0$ then $\nu_r^*+\eta\bar v_r^*\le 0$; since
$v_r(z)\ge 0$ always, $\bar v_r^*=0$.
If $\nu_r^*=\nu_{\max}$ no constraint on $\bar v_r^*$ follows.
The monotone-increase claim follows because
$\nu_r^{(t+1)}\ge\max(0,\nu_r^{(t)}+\eta\delta)=\nu_r^{(t)}+\eta\delta$
whenever $\nu_r^{(t)}+\eta\delta\le\nu_{\max}$.
\end{proof}
\begin{remark}
Saturation at $\nu_{\max}$ signals that the QAOA circuit at depth $p$
and shot budget $S$ cannot drive $\bar v_r$ to zero.  Increasing $p$
or $S$ enlarges the reachable measurement distributions and can resolve
saturation, linking multiplier behavior directly to circuit-resource
sufficiency.
\end{remark}

\subsection{QAOA Circuit and Sampled Objective}
\label{sec:qaoa_implementation}

With $H_{\mathrm{prob}}$ fully assembled by~\eqref{eq:H_prob}, the QAOA circuit implements it as the phase separator of a $p$-layer variational ansatz over the $n_q$-qubit assignment register from~\eqref{eq:qubit_count_reduced}, as outlined in Fig.~\ref{fig:overall_diagram}(4). Each binary variable is mapped to a Pauli-$Z$ operator by
\begin{equation}
\label{eq:bin_qo_ising}
y_q=\frac{1-Z_q}{2},
\end{equation}
up to an additive constant. The circuit starts from
\begin{equation}
\label{eq:psi0}
|\psi_0\rangle
=
\bigotimes_{q=1}^{n_q}
\frac{|0\rangle+|1\rangle}{\sqrt{2}},
\end{equation}
uses the transverse-field mixer
\begin{equation}
\label{eq:H_B}
H_B=\sum_{q=1}^{n_q}X_q,
\end{equation}
and prepares the $p$-layer state
\begin{equation}
\label{eq:qaoa_state}
|\psi(\boldsymbol{\gamma},\boldsymbol{\beta})\rangle
=
\prod_{\ell=1}^{p}
e^{-i\beta_\ell H_B}
e^{-i\gamma_\ell H_{\mathrm{prob}}}
|\psi_0\rangle.
\end{equation}
In implementation, the linear and quadratic QUBO terms are compiled into the phase-separator circuit, while the mixer applies the standard transverse-field evolution across the assignment register.

For a candidate parameter vector $(\boldsymbol{\gamma},\boldsymbol{\beta})$, circuit measurements produce bitstrings $z$ with empirical probabilities $\widehat{P}_{\boldsymbol{\gamma},\boldsymbol{\beta}}(z)$. The modular non-Lagrangian objective is
\begin{equation}
\label{eq:modular_sample_objective}
\widehat{F}(\boldsymbol{\gamma},\boldsymbol{\beta})
=
\sum_z
\widehat{P}_{\boldsymbol{\gamma},\boldsymbol{\beta}}(z)
\left[
c_f(z)+P_{\mathrm{post}}(z)
\right],
\end{equation}
where $c_f(z)$ is the decoded cut weight and $P_{\mathrm{post}}(z)$ collects the active post-measurement penalties for reduced-encoding validity, DFS connectivity, minimum island size, generator presence, load presence, and coherency.

In the Lagrangian configuration,
$\widehat{\mathcal{L}}(\boldsymbol{\gamma},\boldsymbol{\beta};\boldsymbol{\nu})$
from~\eqref{eq:lagrangian_objective} replaces~\eqref{eq:modular_sample_objective}
as the outer-loop objective, with multipliers updated between iterations
according to~\eqref{eq:nu_update}.

After the Lagrangian cost is evaluated, the measured bitstrings are passed to the classical post-processing repair of~\cite{jiang2026regrid}, corresponding to Fig.~\ref{fig:overall_diagram}(6). Each raw bitstring $z$ is mapped to a repaired partition
\begin{equation}
\label{eq:repair_map}
\tilde{z} = \mathcal{P}(z),
\end{equation}
where $\mathcal{P}$ applies greedy single-bus reassignment descent on the QUBO-side energy $Q(z)$ followed by a DFS connectivity repair pass. The repair does not alter the Lagrangian objective in~\eqref{eq:lagrangian_objective}. Among all repaired candidates, the final solution is selected according to the lexicographic ranking
\begin{equation}
\label{eq:final_rank}
\operatorname{rank}(\tilde{z})
=
\Bigl(
\mathbb{I}[P_{\mathrm{post}}(\tilde{z})>0],\;
c_f(\tilde{z}),\;
{-}\widehat{P}_{\boldsymbol{\gamma},\boldsymbol{\beta}}(z)
\Bigr),
\end{equation}
The selected repaired solution is obtained directly from this lexicographic ranking.

\section{Resource and Complexity Analysis}
\label{sec:complexity}

In this section, the computational complexity of the proposed hybrid QAOA
framework is evaluated against full one-hot~\cite{jiang2026regrid} and
non-Lagrangian baselines.
The main comparison is between the two proposed Lagrangian formulations,
$(\mathcal{E}_1^{\mathrm{Lag}},\mathcal{E}_2^{\mathrm{Lag}})$, and three
baselines $(\mathcal{E}_{\mathrm{OH}}^{\mathrm{NL}},
\mathcal{E}_1^{\mathrm{NL}},\mathcal{E}_2^{\mathrm{NL}})$, where
OH denotes the full one-hot encoding and NL denotes direct non-Lagrangian
constraint embedding.  Table~\ref{tab:complexity_summary} summarizes the
qubit counts and computational complexity of these five cases.

\subsection{Notation}
\label{sec:complexity_notation}

Let $T$ denote the number of outer optimization iterations, $S$ the number
of measurement shots per iteration, and $\Delta_G$ the maximum vertex
degree of the working graph $G$.  Let $M_i$ denote the number of
nonconstant QUBO terms in the phase-separator Hamiltonian for formulation
$i$, and let $\Delta_{\mathrm{QUBO}}$ denote the maximum degree of the
QUBO interaction graph.  The per-island slack bit-widths for direct
non-Lagrangian formulations are
\begin{align}
  b_M &= \bigl\lceil\log_2(N-N_{\min}+1)\bigr\rceil, \notag\\
  b_G &= \bigl\lceil\log_2 G\bigr\rceil, \qquad
  b_L = \bigl\lceil\log_2 L\bigr\rceil,
  \label{eq:slack_widths}
\end{align}
where $G=|V_G|$ and $L=|V_L|$.  The total non-Lagrangian slack overhead is
$\Delta_{\mathrm{NL}}=K(b_M+b_G+b_L)$ for working-graph encodings and
$\Delta_{\mathrm{NL}}^0=K(b_M^0+b_G^0+b_L^0)$ for the physical-graph
baseline, where superscript~$0$ denotes physical-graph quantities.
The free-bus generator and load counts in $\mathcal{E}_2$ are
$G_F=|V_G\cap F|$ and $L_F=|V_L\cap F|$.
Unless explicitly stated otherwise, the bounds below are parameterized by
$E$ and $\Delta_G$ and do not assume sparsity.  The common power-network
specialization $E=O(N)$ is invoked only when simplifying edge-dependent
terms to linear-in-$N$ expressions; similarly, constant-depth statements
use the bounded-degree case $\Delta_G=O(1)$.

\begin{table}[!t]
\centering
\caption{Resource summary for the QAOA islanding formulations.}
\label{tab:complexity_summary}
\scriptsize
\setlength{\tabcolsep}{2.5pt}
\renewcommand{\arraystretch}{1.1}
\resizebox{0.78\columnwidth}{!}{%
\begin{tabular}{@{}clcccc@{}}
\toprule
\textbf{\#} & \textbf{Formulation} &
  \makecell{\textbf{Assign.}\\\textbf{qubits}} &
  \makecell{\textbf{Slack}\\\textbf{qubits}} &
  \makecell{\textbf{Total}\\\textbf{qubits}} &
  \makecell{\textbf{$M_i$}\\\textbf{(fixed $K$)}} \\
\midrule
(i)   & Full one-hot, w/o Lag.
      & $KN_0$   & $\Delta_{\mathrm{NL}}^0$
      & $KN_0+\Delta_{\mathrm{NL}}^0$ & $O(E_0+N_0^2)$ \\
(ii)  & $\mathcal{E}_1$, w/o Lag.
      & $(K-1)N$ & $\Delta_{\mathrm{NL}}$
      & $(K-1)N+\Delta_{\mathrm{NL}}$ & $O(E+N^2)$ \\
(iii) & $\mathcal{E}_2$, w/o Lag.
      & $(K-1)F$ & $\Delta_{\mathrm{NL}}$
      & $(K-1)F+\Delta_{\mathrm{NL}}$ & $O(E+F^2)$ \\
(iv)  & $\mathcal{E}_1$+Lag.
      & $(K-1)N$ & $0$
      & $(K-1)N$ & $O(E+N+A^2)$ \\
(v)   & $\mathcal{E}_2$+Lag.
      & $(K-1)F$ & $0$
      & $(K-1)F$ & $O(E+F+N)$ \\
\bottomrule
\end{tabular}
}
\end{table}

\begin{proposition}[Fixed-$K$ QUBO term complexity]
\label{prop:complexity}
Let $M_i$ denote the number of nonconstant QUBO terms in the
phase-separator Hamiltonian for formulation $i\in\{(i),\ldots,(v)\}$.
For fixed $K\geq 2$ and $A,G,L\leq N$,
\begin{alignat}{2}
  M_{\mathrm{full}}              &= O(E_0 + N_0^2), \label{eq:Mfull_fixedK}\\
  M_{\mathcal{E}_1,\mathrm{NL}} &= O(E   + N^2),   \label{eq:Mred_NL_fixedK}\\
  M_{\mathcal{E}_2,\mathrm{NL}} &= O(E   + F^2),   \label{eq:Msf_NL_fixedK}\\
  M_{\mathcal{E}_1,\mathrm{Lag}}&= O(E + N + A^2), \label{eq:Mred_Lag_fixedK}\\
  M_{\mathcal{E}_2,\mathrm{Lag}}&= O(E + F + N).   \label{eq:Msf_Lag_fixedK}
\end{alignat}
\end{proposition}

The following subsections derive these bounds for each formulation and detail
the corresponding circuit and algorithmic costs.

\subsection{Formulation~(i): Full One-Hot Encoding Without Lagrangian}
\label{sec:complexity_fi}

\paragraph{Circuit width.}
Under the full one-hot encoding every island $k\in\{1,\ldots,K\}$ is
represented by $N_0$ binary variables drawn from the physical graph.  Each
feasibility constraint requires a binary slack register, contributing
$\Delta_{\mathrm{NL}}^0=K(b_M^0+b_G^0+b_L^0)$ additional qubits.  The
total circuit width is
\[
  n_q^{(i)} = KN_0 + \Delta_{\mathrm{NL}}^0.
\]

\paragraph{Hamiltonian and gate complexity.}
Under the full one-hot representation, each island indicator enters the
Hamiltonian directly.
The phase-separator Hamiltonian consists of the cut term $H_f$, the one-hot
uniqueness penalty $H_U$, the coherency penalty $H_G$, and the direct slack
penalty blocks.  These terms contribute $O(E_0K)$, $O(N_0K^2)$, and
$O(KA_0^2)$ terms for $H_f$, $H_U$, and $H_G$, respectively.  For the
min-size constraint on island $k$, the squared slack penalty has support over
$N_0+b_M^0$ binary variables and therefore contributes
$O((N_0+b_M^0)^2)$ QUBO terms.  Summing over all $K$
islands and all three constraint types gives a total slack overhead of
$O(K(N_0+b_M^0)^2)=O(N_0^2)$ for fixed~$K$.  This dominates all
combinatorial contributions, giving
\[
  M_{\mathrm{full}} = O(E_0 + N_0^2).
\]
Each QUBO term compiles to $O(1)$ quantum gates, so the gate count per
QAOA layer is $O(E_0+N_0^2)$.

\paragraph{Circuit depth.}
The slack penalty for each island couples every bus variable to every other
bus variable and to the slack bits within the same block, making every
variable a neighbour of $O(N_0 + b_M^0) = O(N_0)$ others in the QUBO
interaction graph.  The maximum degree is therefore
$\Delta_{\mathrm{QUBO}}^{(i)} = O(N_0)$, giving a circuit depth of
$O(N_0)$ per layer and $O(p \cdot N_0)$ for the full $p$-layer circuit.

\paragraph{Post-processing complexity.}
The adopted classical repair step evaluates the active assignment moves in
each round.
Since $n_q^{(i)}=O(N_0)$ for fixed~$K$, the post-processing cost per sample
is $O\bigl((n_q^{(i)})^2\bigr)=O(N_0^2)$.

\paragraph{Total algorithm complexity.}
The dominant cost across $T$ outer iterations, $S$ shots, and $p$ QAOA
layers is
\[
  O\!\bigl(T \cdot S \cdot (p \cdot N_0^2 + N_0^2)\bigr)
  = O(T \cdot S \cdot p \cdot N_0^2),
\]
with circuit execution time proportional to $O(p \cdot N_0)$.

\subsection{Formulation~(ii): Reduced Encoding $\mathcal{E}_1$ Without Lagrangian}
\label{sec:complexity_fii}
\label{sec:inferred_indicator}

\paragraph{Circuit width.}
In the reduced encoding $\mathcal{E}_1$, the assignment variables are defined
on a working graph with $N\leq N_0$ buses.  Each bus is represented by
$K-1$ explicit binary variables, while membership in island $K$ is given by the
inferred indicator
\begin{equation}
  Y_{i,K} = 1 - \sum_{\ell=1}^{K-1} y_{i,\ell}.
  \label{eq:inferred_ind}
\end{equation}
The assignment register shrinks to $(K-1)N$ qubits.  Adding the slack
overhead $\Delta_{\mathrm{NL}}=K(b_M+b_G+b_L)$ gives
\[
  n_q^{(ii)} = (K-1)N + \Delta_{\mathrm{NL}}.
\]

\paragraph{Hamiltonian and gate complexity.}
Because $Y_{i,K}$ is a linear form over $K-1$ variables, any bilinear
product involving island $K$ expands as
\begin{equation}
  Y_{i,K}Y_{j,K}
  = 1 - \sum_\ell y_{i,\ell} - \sum_m y_{j,m}
    + \sum_{\ell,m} y_{i,\ell}y_{j,m},
  \label{eq:indicator_expansion}
\end{equation}
introducing $(K-1)^2$ quadratic monomials per edge for the inferred-island
cut term alone.  For the min-size constraint associated with the inferred
island $k=K$, \eqref{eq:inferred_ind} expands the support to all
$(K-1)N$ explicit variables.  The corresponding squared-penalty block
therefore contributes
\begin{equation}
  O\!\bigl(((K-1)N+b_M)^2\bigr)
  \label{eq:slack_split_red}
\end{equation}
terms.  For fixed~$K$ this is $O(N^2)$ and dominates all other
contributions, giving $M_{\mathcal{E}_1,\mathrm{NL}}=O(E+N^2)$ and a
per-layer gate count of $O(E+N^2)$.

\paragraph{Circuit depth.}
The inferred-island slack block couples all $(K-1)N + b_M$ variables to one
another, so each variable has $O((K-1)N + b_M) = O(N)$ neighbours in the
QUBO interaction graph for fixed~$K$.  The maximum degree is
$\Delta_{\mathrm{QUBO}}^{(ii)} = O(N)$, giving a circuit depth of $O(N)$
per layer and $O(p \cdot N)$ for the full circuit.  This is the same
asymptotic depth as Formulation~(i) with $N_0$ replaced by $N \leq N_0$.

\paragraph{Post-processing complexity.}
With $n_q^{(ii)}=O(N)$ for fixed~$K$, the post-processing cost per sample
is $O(N^2)$.

\paragraph{Total algorithm complexity.}
\[
  O\!\bigl(T \cdot S \cdot (p \cdot N^2 + N^2)\bigr)
  = O(T \cdot S \cdot p \cdot N^2),
\]
with circuit execution time proportional to $O(p \cdot N)$.
The qubit reduction from $KN_0$ to $(K-1)N+\Delta_{\mathrm{NL}}$ lowers the
circuit width, but the inferred-island slack block keeps both gate-count
and depth scaling at $O(N^2)$ and $O(N)$, respectively.

\subsection{Formulation~(iii): Symmetry-Fixed Encoding $\mathcal{E}_2$ Without Lagrangian}
\label{sec:complexity_fiii}

\paragraph{Circuit width.}
In the symmetry-fixed encoding $\mathcal{E}_2$, the canonical representative
sets the island label of each anchor bus $g\in\mathcal{A}_c$ as
$Y_{g,k}^{\mathcal{E}_2}=\mathbf{1}[k=c]$
(Proposition~\ref{prop:canonical_labeling}).  The assignment register is
therefore defined only over the $F=N-A$ free buses, and the circuit width becomes
\[
  n_q^{(iii)} = (K-1)F + \Delta_{\mathrm{NL}}.
\]

\paragraph{Hamiltonian and gate complexity.}
Constant anchor labels eliminate $H_G$ structurally and reduce
the one-hot validity penalty to $O(FK^2)$ terms involving only free buses.
The inferred-island slack block mirrors Formulation~(ii) with $N$ replaced
by $F$, contributing $O(((K-1)F+b_M)^2)=O(F^2)$ terms for fixed~$K$.
This dominates, giving $M_{\mathcal{E}_2,\mathrm{NL}}=O(E+F^2)$ and a
per-layer gate count of $O(E+F^2)$.

\paragraph{Circuit depth.}
The inferred-island slack block couples each free-bus variable to at most
$O(F)$ other variables, so $\Delta_{\mathrm{QUBO}}^{(iii)}=O(F)$.  The
corresponding circuit depth is $O(F)$ per layer and $O(p \cdot F)$ for the
full circuit.  Relative to Formulation~(ii), the depth scaling is reduced
from $O(N)$ to $O(F)=O(N-A)$.

\paragraph{Post-processing complexity.}
With $n_q^{(iii)}=O(F)$ for fixed~$K$, the post-processing cost per sample
is $O(F^2)$.  Since $F=N-A\leq N$, this is never worse than
Formulation~(ii) and can be substantially cheaper when many anchors are
present.

\paragraph{Total algorithm complexity.}
\[
  O\!\bigl(T \cdot S \cdot (p \cdot F^2 + F^2)\bigr)
  = O(T \cdot S \cdot p \cdot F^2),
\]
with circuit execution time proportional to $O(p \cdot F)$.
Both the quantum and post-processing costs scale in $F^2$ rather than
$N^2$, providing a uniform reduction across all phases of the algorithm.

\subsection{Formulation~(iv): Reduced Encoding $\mathcal{E}_1$ With Lagrangian}
\label{sec:complexity_fiv}

\paragraph{Circuit width.}
Replacing slack-variable blocks with the classical Lagrangian update
eliminates all slack qubits, reducing the circuit width to
\[
  n_q^{(iv)} = (K-1)N.
\]

\paragraph{Hamiltonian and gate complexity.}
The phase-separator Hamiltonian is determined by the cut term $H_f$, the
one-hot uniqueness penalty $H_U$, the connectivity surrogate $H_c$, and the
coherency penalty $H_G$.  Their term counts scale as $O(EK^2)$, $O(NK^2)$,
$O(EK^2+NK)$, and $O(K^2A^2)$, respectively.  For fixed~$K$, the resulting
Hamiltonian size is
\[
  M_{\mathcal{E}_1,\mathrm{Lag}} = O(E + N + A^2).
\]
The per-layer gate count is therefore $O(E+N+A^2)$, a qualitative reduction
from the $O(N^2)$ gate count of Formulation~(ii) whenever $A$ grows
sublinearly in $N$.

\paragraph{Circuit depth.}
With no slack blocks, the densest part of the QUBO interaction graph comes
from $H_G$, which connects every pair of anchor-bus variables
across all $A$ anchor buses.  Each anchor variable therefore has
$O(A \cdot (K-1)) = O(A)$ neighbours for fixed~$K$.  The cut and
connectivity terms $H_f$, $H_c$ connect bus $i$'s variables
only to its power-graph neighbours, contributing $O(\Delta_G)$ per
variable.  The maximum QUBO interaction-graph degree is thus
$\Delta_{\mathrm{QUBO}}^{(iv)} = O(\Delta_G + A)$, giving a circuit depth
of $O(\Delta_G + A)$ per layer and $O(p \cdot (\Delta_G + A))$ for the
full circuit, which is a substantial reduction from the $O(p \cdot N)$ depth of
Formulation~(ii).

\paragraph{Post-processing complexity.}
Although the Hamiltonian is cheaper, the circuit width remains $n_q^{(iv)}
=(K-1)N=O(N)$, so the greedy post-processing cost per sample is still
$O(N^2)$.  The post-processing bottleneck therefore dominates the quantum
cost when $p \cdot A^2 < N^2$, i.e., when $p < (N/A)^2$.

\paragraph{Lagrangian multiplier update.}
Each outer iteration evaluates $3K$ aggregate violations
(min-size, generator presence, load presence across all islands) and
performs one projected-ascent update per constraint.  This classical
multiplier-update overhead is $O(K\cdot N)$ per iteration.

\paragraph{Total algorithm complexity.}
\[
  O\!\Bigl(T \cdot \bigl[S \cdot (p \cdot (E+N+A^2) + N^2) + KN\bigr]\Bigr),
\]
with circuit execution time proportional to $O(p \cdot (\Delta_G + A))$.
For sparse graphs ($E=O(N)$) and small anchor sets ($A=o(N)$), the dominant
term is $O(T \cdot S \cdot N^2)$ from post-processing, not the quantum
circuit.

\subsection{Formulation~(v): Symmetry-Fixed Encoding $\mathcal{E}_2$ With Lagrangian}
\label{sec:complexity_fv}

\paragraph{Circuit width.}
Combining symmetry-fixing with the Lagrangian treatment eliminates both
slack qubits and anchor-variable overhead, giving the minimum circuit width
among all five formulations:
\[
  n_q^{(v)} = (K-1)F.
\]

\paragraph{Hamiltonian and gate complexity.}
Constant anchor labels set $H_G=0$ structurally, and the
generator-presence constraint is satisfied by
assumption~\eqref{eq:anchor_contains_generator} and therefore also absent
from the QUBO.  The active phase-separator Hamiltonian is
\[
  H_f + \lambda_U H_U + \alpha_c H_c,
\]
contributing $O(EK^2)$, $O(FK^2)$, and $O(EK^2+NK)$ terms, respectively.
For fixed~$K$ these collapse to $O(E)$, $O(F)$, and $O(E+N)$, yielding
\[
  M_{\mathcal{E}_2,\mathrm{Lag}} = O(E + F + N).
\]
On a sparse working graph with $E=O(N)$, the per-layer gate count is
$O(N)$, which is linear in the problem size.

\paragraph{Circuit depth.}
With $H_G$ absent and variables restricted to free buses, the
only inter-variable connections in the QUBO interaction graph come from
$H_f$ and $H_c$, which link bus $i$'s variables only to those
of its power-graph neighbours.  Each variable therefore has at most
$O(\Delta_G \cdot (K-1)) = O(\Delta_G)$ neighbours for fixed~$K$, giving
\[
  \Delta_{\mathrm{QUBO}}^{(v)} = O(\Delta_G).
\]
The circuit depth is $O(\Delta_G)$ per layer and $O(p \cdot \Delta_G)$ for
the full circuit.  Since $\Delta_G$ is a small constant for power networks,
this is effectively $O(p)$, independent of $N$, $F$, or $A$, and is the
minimum depth among all five formulations.  This directly explains the
dramatic circuit-depth reductions observed in the benchmark
(Table~\ref{tab:sf_bench}), e.g.\ depth 19 vs.\ 67 on IonQ for the
IEEE~9-bus system.

\paragraph{Post-processing complexity.}
With $n_q^{(v)}=O(F)$, the post-processing cost per sample is $O(F^2)$,
strictly smaller than the $O(N^2)$ cost of Formulation~(iv) whenever
$A\geq 1$.

\paragraph{Lagrangian multiplier update.}
The same projected-ascent logic as Formulation~(iv) applies, costing
$O(KN)$ per outer iteration.  Generator-presence multipliers are omitted
because the fixed anchor assigned to each island already contains a generator
by~\eqref{eq:anchor_contains_generator}.  Minimum-size and load-presence
violations are still evaluated after decoding, since the remaining free-bus
assignments determine whether each island has enough buses and at least one
load bus.

\paragraph{Total algorithm complexity.}
\[
  O\!\Bigl(T \cdot \bigl[S \cdot (p \cdot (E+F+N) + F^2) + KN\bigr]\Bigr),
\]
with circuit execution time proportional to $O(p \cdot \Delta_G)$.
For sparse graphs ($E=O(N)$) this becomes
$O\!\bigl(T\cdot [S\cdot(pN+F^2)+KN]\bigr)$, and the shot-dependent term
is $O(T \cdot S \cdot (pN + F^2))$ whenever the multiplier-update cost
$O(TKN)$ is asymptotically lower order.
Thus, among the five formulations considered, Formulation~(v) gives the
most favorable asymptotic resource profile, reducing the assignment register,
the phase-separator size, the circuit depth, and the post-processing cost.

\subsection{Concrete Qubit Requirements}
\label{sec:concrete_qubits}

The asymptotic reductions established above manifest as substantial
qubit savings on realistic benchmark systems.
Table~\ref{tab:concrete_qubits} reports the total qubit count required
by each formulation for eleven IEEE bus systems.  Values for the eight
systems up to the 89-bus case correspond to the executed benchmark
records in Tables~\ref{tab:lagrangian_bench} and~\ref{tab:sf_bench},
while the three larger systems are obtained by evaluating the same
closed-form expressions. Figure~\ref{fig:qubit_scaling_comparison}
visualizes these counts together with the 156-qubit hardware limit,
linking the executed cases to the larger-system resource projections.

Within each row, the qubit count decreases
monotonically from column~(i) to column~(v), reflecting the successive
elimination of $N_0-(K-1)N$ assignment qubits through graph contraction
and reduced encoding, of $\Delta_{\mathrm{NL}}$ slack qubits through the
Lagrangian treatment, and of the fixed-anchor assignment variables
through symmetry fixing.  For the IEEE~73-bus system, the proposed
$\mathcal{E}_2$+Lagrangian formulation reduces the qubit count from 267
under the full one-hot baseline to 54, a nearly fivefold reduction, and
for the IEEE~300-bus projection, from 1296 to 693.

\begin{table}[!t]
\centering
\caption{Total qubit requirements for the islanding formulations across
IEEE bus systems.}
\label{tab:concrete_qubits}
\footnotesize
\setlength{\tabcolsep}{3.5pt}
\renewcommand{\arraystretch}{0.95}
\begin{tabular}{@{}lrr
                r  
                r  
                r  
                r  
                r  
                @{}}
\toprule
\textbf{System} & $K$ & $N$ &
  \makecell{(i)\\Full one-hot\\w/o Lag.} &
  \makecell{(ii)\\$\mathcal{E}_1$\\w/o Lag.} &
  \makecell{(iv)\\$\mathcal{E}_1$\\+Lag.} &
  \makecell{(iii)\\$\mathcal{E}_2$\\w/o Lag.} &
  \makecell{(v)\\$\mathcal{E}_2$\\+Lag.} \\
\midrule
IEEE 9-bus    & 2 &   9 &   32 &  23 &   9 &  20 &   6 \\
IEEE 14-bus   & 2 &  14 &   43 &  29 &  14 &  29 &   9 \\
IEEE 24-bus   & 3 &  20 &  108 &  76 &  40 &  62 &  26 \\
IEEE 30-bus   & 2 &  29 &   86 &  55 &  29 &  50 &  24 \\
IEEE 39-bus   & 3 &  39 &  156 & 123 &  78 & 103 &  58 \\
IEEE 57-bus   & 2 &  52 &  140 &  78 &  52 &  76 &  50 \\
IEEE 73-bus   & 3 &  54 &  267 & 156 & 108 & 102 &  54 \\
IEEE 89-bus   & 3 &  89 &  321 & 232 & 178 & 208 & 154 \\
IEEE 118-bus  & 4 &  80 &  548 & 308 & 240 & 260 & 192 \\
IEEE 145-bus  & 5 & 100 &  830 & 485 & 400 & 465 & 380 \\
IEEE 300-bus  & 4 & 292 & 1296 & 968 & 876 & 785 & 693 \\
\bottomrule
\end{tabular}
\end{table}

\begin{figure*}[!t]
\centering
\includegraphics[width=0.84\textwidth]{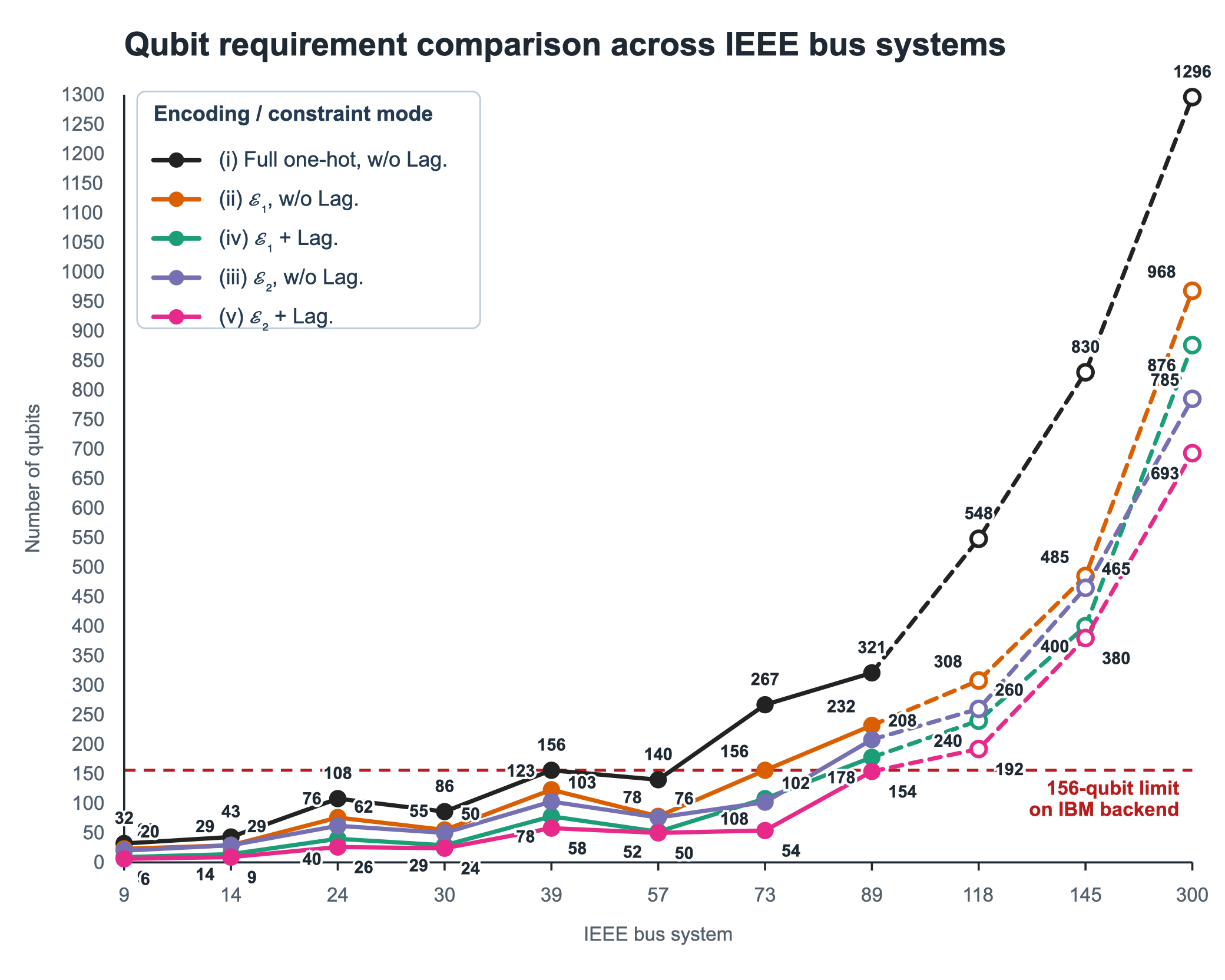}
\caption[Qubit requirements across IEEE bus systems]{Qubit requirements of the five islanding formulations
across IEEE bus systems. Black represents full one-hot QAOA. Orange and green represent
encoding~$\mathcal{E}_1$ without and with the Lagrangian treatment, respectively. Purple and magenta
represent encoding~$\mathcal{E}_2$ without and with the Lagrangian treatment, respectively. Filled markers
and solid lines indicate executed cases through the IEEE~89-bus system. Open markers and dashed lines
indicate formula-based projections for the IEEE~118-, 145-, and 300-bus systems. The red dashed line marks
the 156-qubit capacity of the available IBM backend.}
\label{fig:qubit_scaling_comparison}
\end{figure*}

Figure~\ref{fig:qubit_scaling_comparison} quantifies the cumulative qubit
reduction produced by the proposed formulation.  Before the binary variables
are introduced, the buses in each connected coherent generator group are
merged into a single super-bus because these buses must remain in the same
island.  This graph-reduction step decreases the number of buses represented in
the optimization model from $N_0$ to $N$ without changing the admissible
physical partitions.  Encoding~$\mathcal{E}_1$ then reduces the assignment
register further by representing only $K-1$ of the $K$ island labels.  The
Lagrangian treatment eliminates the slack register, while symmetry fixing in
encoding~$\mathcal{E}_2$ removes the assignment variables associated with the
anchor buses.  The ordering of the solid curves across the evaluated systems
confirms the incremental contribution of these reductions.  For the
IEEE~89-bus system, formulation~(v) requires 154
qubits, whereas formulations~(i)--(iv) require 321, 232, 208, and 178 qubits,
respectively.  Formulation~(v) is therefore the only formulation that remains
below the 156-qubit hardware limit for this system.  In contrast, the full
one-hot baseline reaches this limit at IEEE~39-bus.

The open markers and dashed colored segments beyond IEEE~89-bus represent
resource estimates obtained from the closed-form qubit-count expressions.  All
five projected counts exceed 156 qubits from IEEE~118-bus onward.  The proposed
reductions therefore extend the range supported by a fixed qubit budget,
although larger systems still require greater hardware capacity.
Formulation~(v) retains the lowest projected requirement, with 192, 380, and
693 qubits for the IEEE~118-, 145-, and 300-bus systems, respectively.  The
corresponding values for formulation~(i) are 548, 830, and 1296 qubits, yielding
absolute reductions of 356, 450, and 603 qubits.  The difference between
formulations~(iv) and~(v) isolates the contribution of symmetry fixing.  This
reduction is 20 qubits for IEEE~145-bus with five anchor buses and 183 qubits
for IEEE~300-bus with 61 anchor buses.  The slope and spacing of the dashed
curves are therefore governed by the number of islands, the contracted
working-graph size, and the number of fixed anchor buses in addition to the
nominal system size.  Collectively, these results establish that the combined
formulation provides the lowest qubit requirement among the five alternatives
and increases the system scale accessible under a finite hardware budget.

\section{Numerical Example}
\label{sec:numerical}

\subsection{Experimental Setup and Test Systems}

This section evaluates the proposed QAOA-based islanding framework on IEEE benchmark systems ranging
from 9 to 89 buses. These cases span
different network sizes and islanding structures, including both two-island and three-island settings, and
therefore provide a representative test bed for assessing the framework under
increasing combinatorial complexity. For each case, the coherent generator groups determine the target number
of islands~$K$. Each benchmark case is evaluated under both assignment encodings introduced in
Section~\ref{formulation}, the reduced encoding~$\mathcal{E}_1$ (Section~\ref{sec:reduced_encoding}) and the
encoding~$\mathcal{E}_2$ (Section~\ref{sec:sf_encoding}), and under each encoding the proposed
slack-free Lagrangian configuration is compared against a non-Lagrangian custom QUBO in which the same
feasibility terms are embedded directly as penalty components. The evaluation focuses on cut quality,
feasibility behavior, and quantum resource requirements across both encodings, and
Section~\ref{sec:landscape_trainability} further examines how the two encodings and the Lagrangian treatment
shape the sampled QAOA optimization landscape.

The experiments are implemented on the IBM Marrakesh backend from IBM Quantum and on the Rigetti, IQM, AQT, and
IonQ backends available through AWS Braket, as recorded in the benchmark outputs. For quantum implementation,
Table~\ref{tab:exp_setup} summarizes the QAOA configuration used for each
system, including the target island count~$K$, QAOA depth parameter~$p$, shot count~$S$, and classical
optimizer iteration budget~$I_{\max}$.

\begin{table}[!tbp]
\centering
\scriptsize
\setlength{\tabcolsep}{3pt}
\caption{QAOA benchmark settings.}
\label{tab:exp_setup}
\begin{tabular}{@{}lcccc@{}}
\toprule
\textbf{Case} & $K$ & $p$ & $S$ & $I_{\max}$ \\
\midrule
IEEE 9-bus  & 2 & 1 & 100  & 2 \\
IEEE 14-bus & 2 & 1 & 100  & 2 \\
IEEE 24-bus & 3 & 1 & 1000 & 2 \\
IEEE 30-bus & 2 & 1 & 300  & 2 \\
IEEE 39-bus & 3 & 2 & 2000 & 2 \\
IEEE 57-bus & 2 & 2 & 1500 & 2 \\
IEEE 73-bus & 3 & 3 & 2000 & 2 \\
\bottomrule
\end{tabular}
\end{table}

The total QUBO qubits required by each Lagrangian and non-Lagrangian configuration are reported in
Table~\ref{tab:concrete_qubits}, and the corresponding backend-level circuit depth, cut value, and
feasibility rate are reported in Table~\ref{tab:lagrangian_bench}.

\begin{table*}[!tbp]
\centering
\scriptsize
\setlength{\tabcolsep}{0pt}
\renewcommand{\arraystretch}{1.08}
\caption{Islanding solutions obtained by the proposed QAOA framework.}
\label{tab:islanding_partitions}
\begin{tabular}{@{}>{\raggedright\arraybackslash}p{0.97\textwidth}@{}}
\toprule
\textbf{IEEE cases and islanding solutions} \\
\midrule
\textbf{9-bus:} Island 1: 1, 4--5; Island 2: 2--3, 6--9 \\
\addlinespace[2pt]
\textbf{14-bus:} Island 1: 1--5; Island 2: 6--14 \\
\addlinespace[2pt]
\textbf{24-bus:} Island 1: 6, 10--14, 20, 23; Island 2: 3, 15--19, 21--22, 24; Island 3: 1--2, 4--5, 7--9 \\
\addlinespace[2pt]
\textbf{30-bus:} Island 1: 9--30; Island 2: 1--8 \\
\addlinespace[2pt]
\textbf{39-bus:} Island 1: 1--3, 25--30, 37--38; Island 2: 15--24, 33--36; Island 3: 4--14, 31--32, 39 \\
\addlinespace[2pt]
\textbf{57-bus:} Island 1: 1--5, 11, 13--23, 32--49, 56--57; Island 2: 6--10, 12, 24--31, 50--55 \\
\bottomrule
\end{tabular}
\end{table*}

\begin{figure*}[!t]
\centering
\setlength{\tabcolsep}{0pt}
\renewcommand{\arraystretch}{1.0}
\begin{tabular}{@{}c@{\hspace{1pt}}c@{}}
\includegraphics[width=0.48\textwidth]{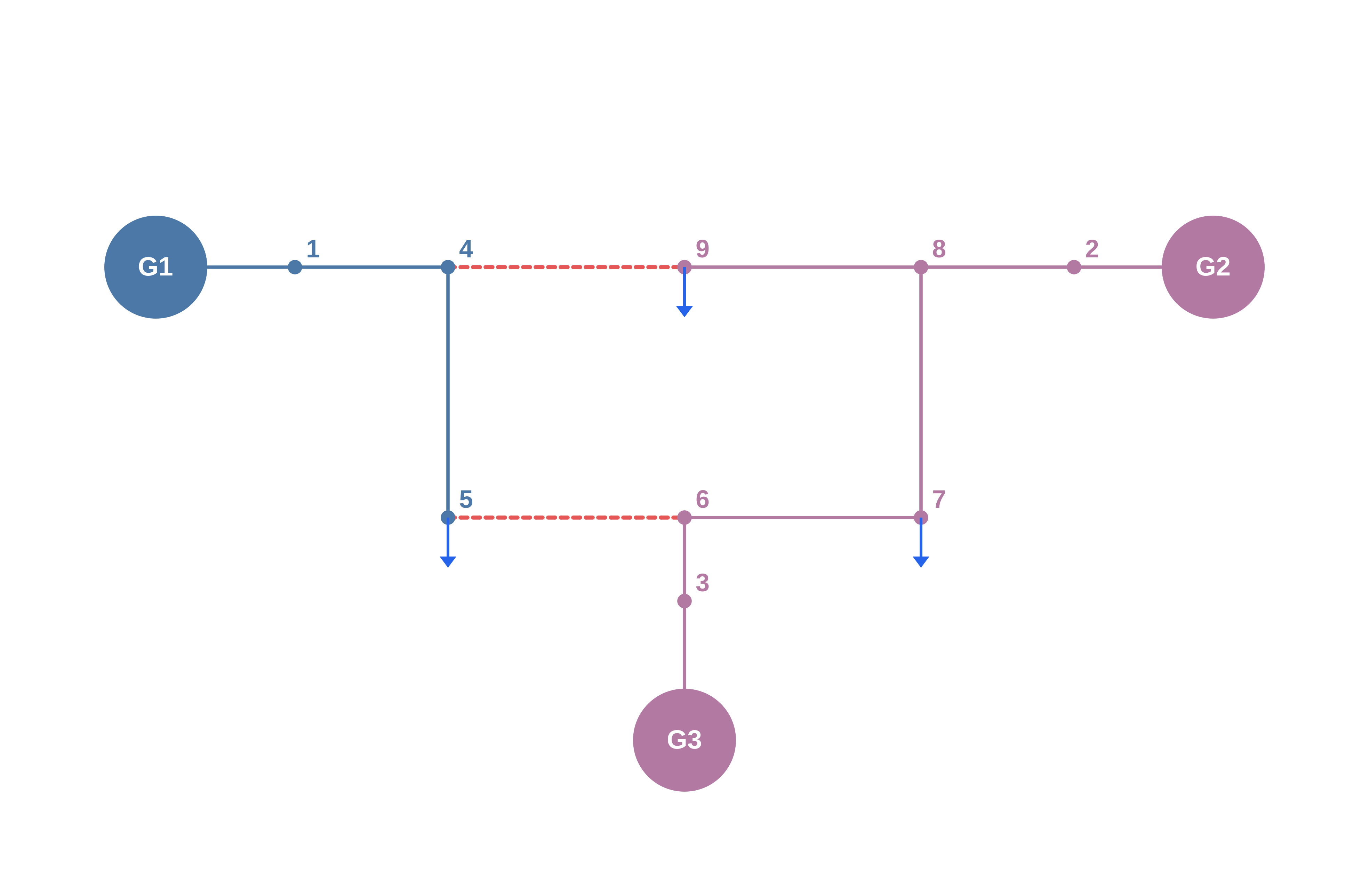} &
\includegraphics[width=0.48\textwidth]{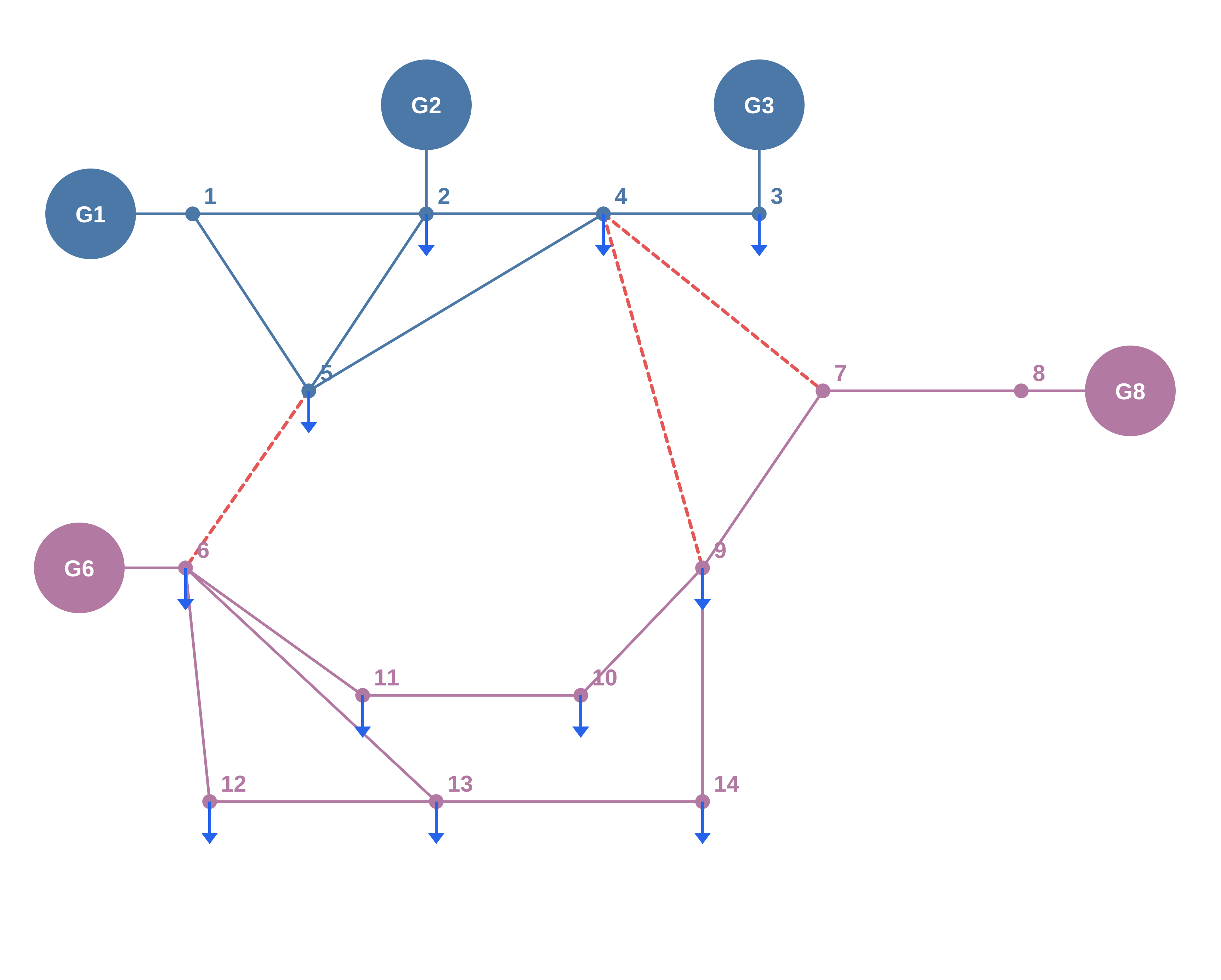} \\
\footnotesize (a) IEEE 9-bus & \footnotesize (b) IEEE 14-bus \\[4pt]
\includegraphics[width=0.48\textwidth]{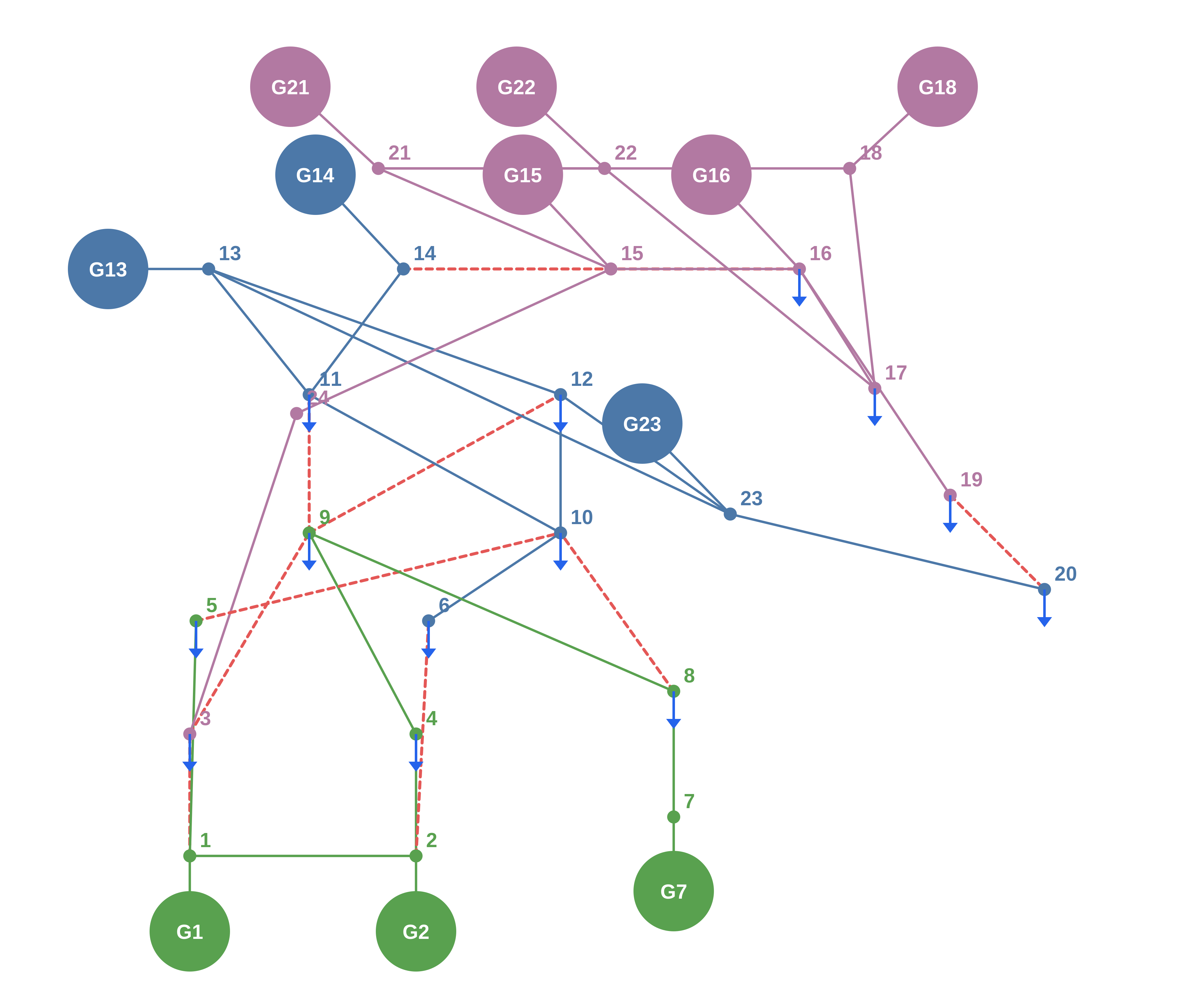} &
\includegraphics[width=0.48\textwidth]{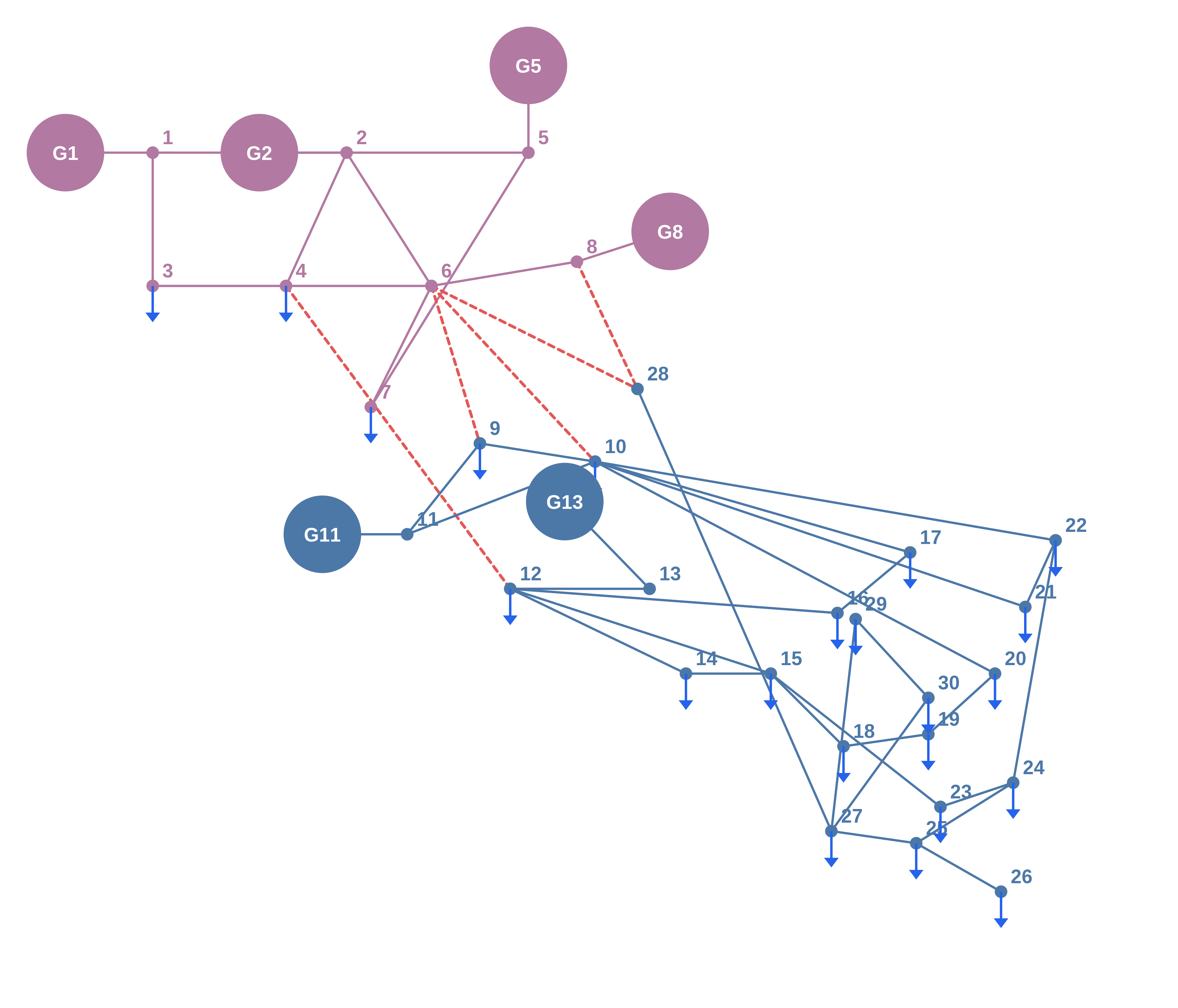} \\
\footnotesize (c) IEEE 24-bus & \footnotesize (d) IEEE 30-bus
\end{tabular}
\caption[QAOA islanding results on small IEEE systems]{Optimal islanding results generated by the proposed QAOA framework on the IEEE 9-, 14-, 24-, and 30-bus test systems. Bus colors indicate
the island assignment of each bus, and the red dashed lines denote the transmission interfaces removed to form the final
islands. The resulting islands consist mainly of neighboring buses and require only a few transmission lines
to be disconnected.}
\label{fig:island_viz_ieee_small}
\end{figure*}

\begin{figure*}[!t]
\centering
\setlength{\tabcolsep}{0pt}
\renewcommand{\arraystretch}{1.0}
\begin{tabular}{@{}c@{\hspace{1pt}}c@{}}
\includegraphics[width=0.48\textwidth]{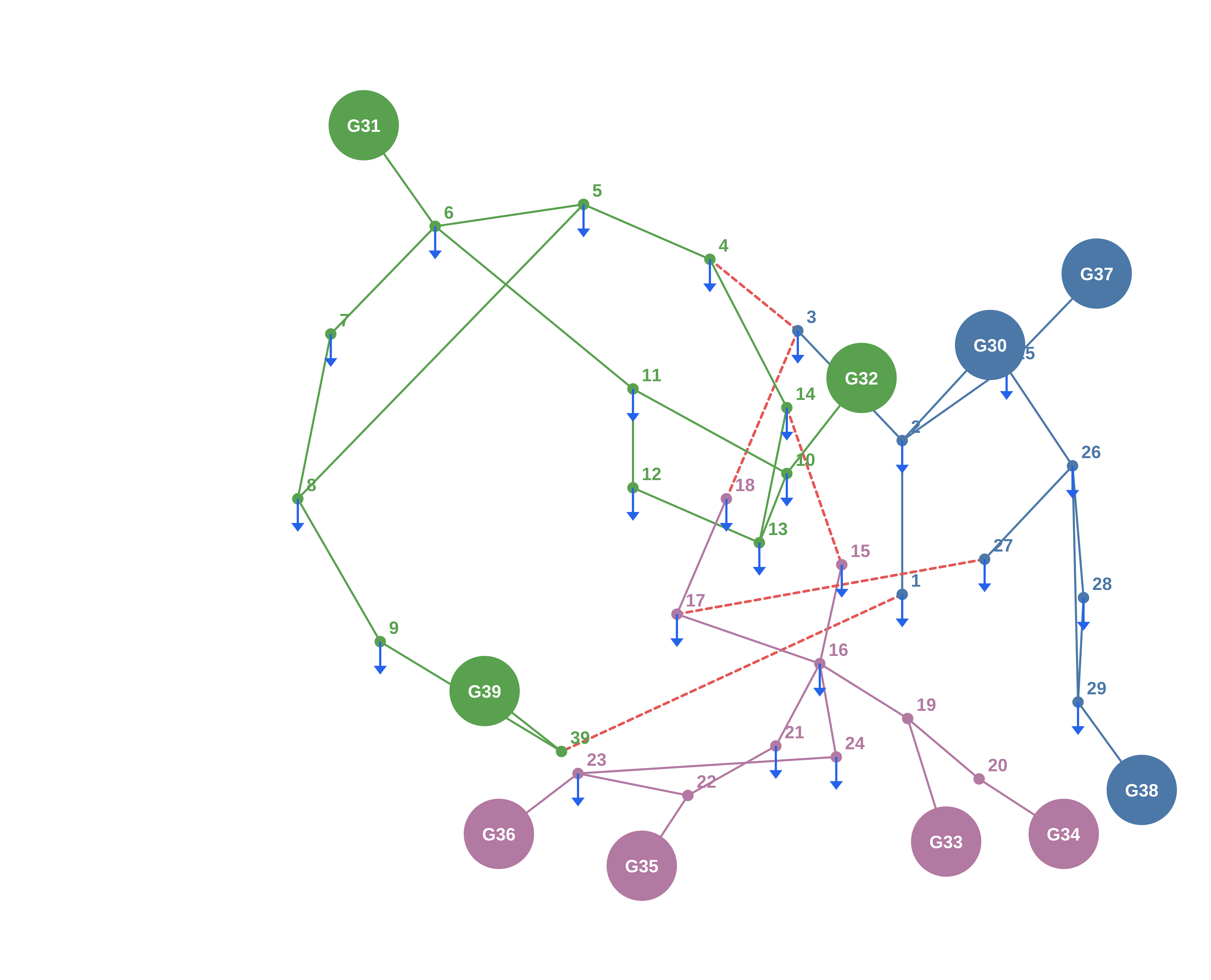} &
\includegraphics[width=0.48\textwidth]{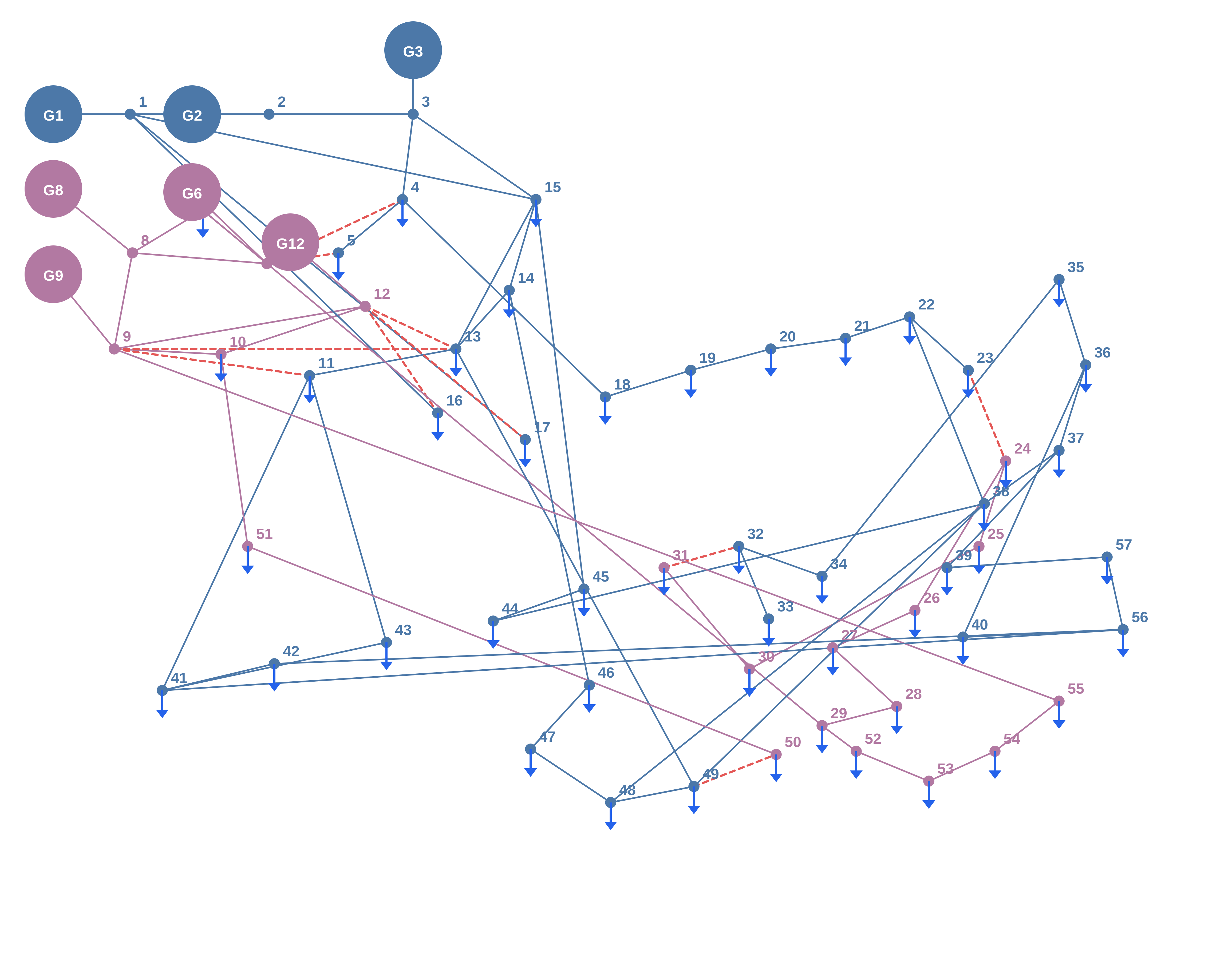} \\
\footnotesize (a) IEEE 39-bus & \footnotesize (b) IEEE 57-bus
\end{tabular}\par\vspace{4pt}
\begin{tabular}{@{}c@{\hspace{1pt}}c@{}}
\includegraphics[width=0.42\textwidth]{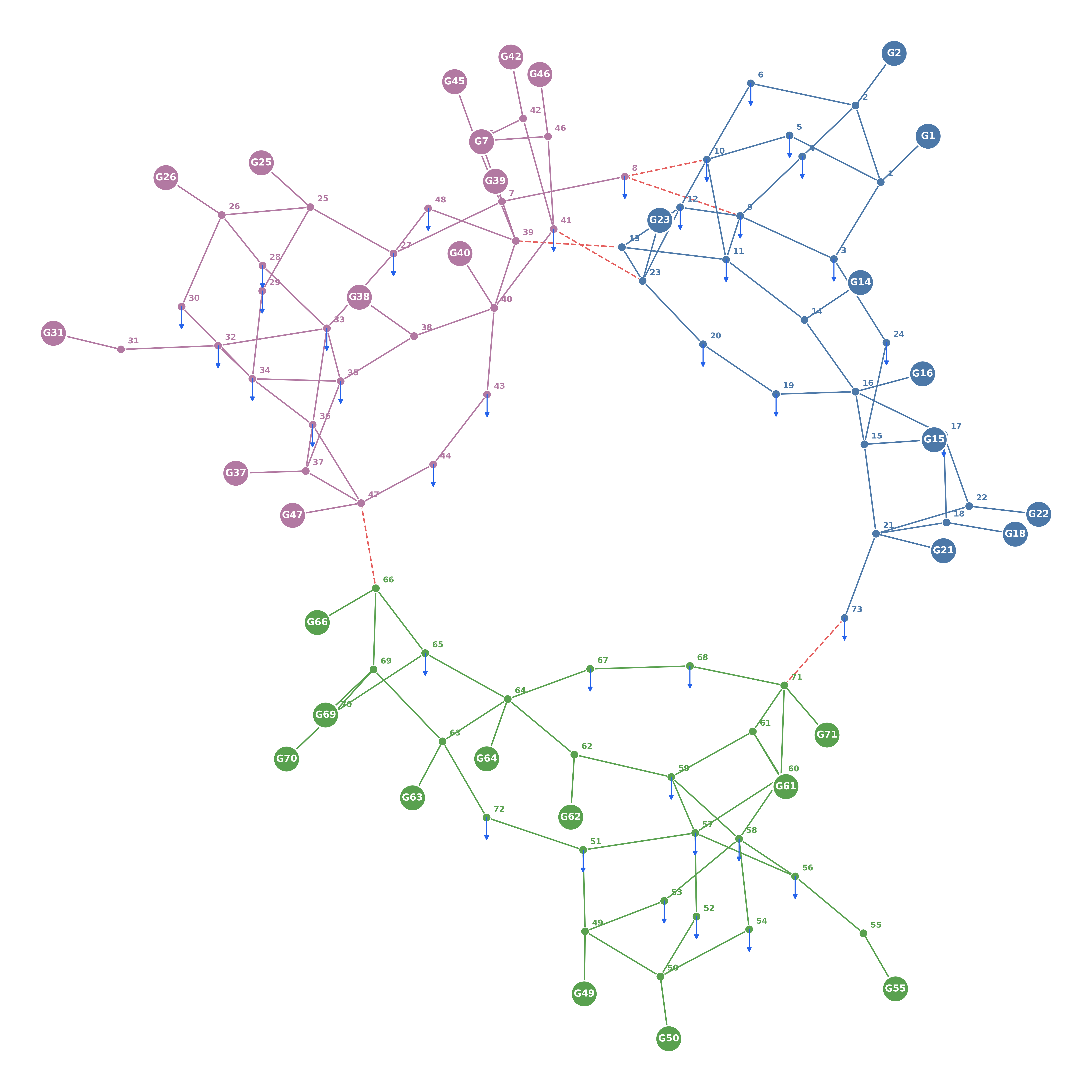} &
\includegraphics[width=0.54\textwidth]{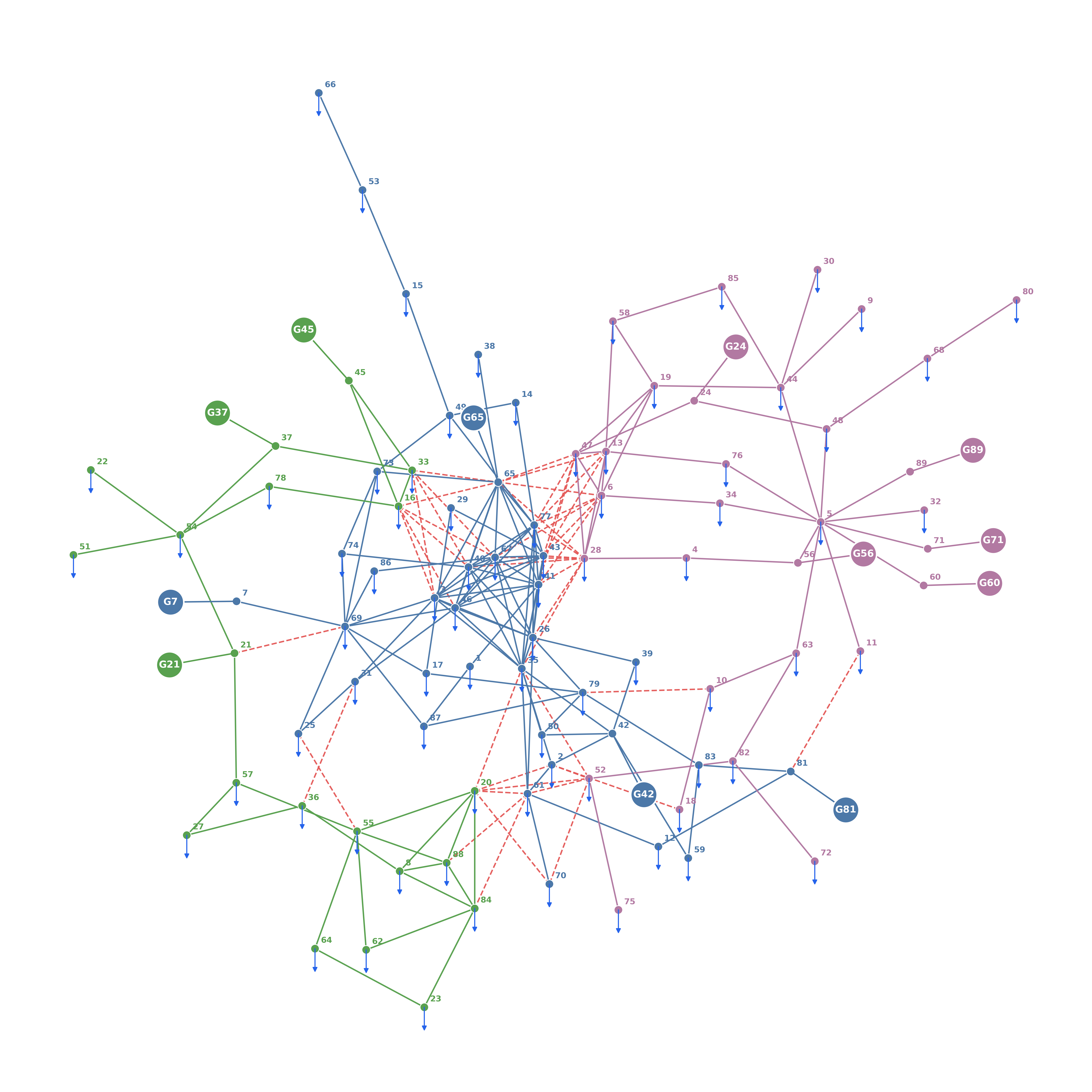} \\
\footnotesize (c) IEEE 73-bus & \footnotesize (d) IEEE 89-bus
\end{tabular}
\caption[QAOA islanding results on large IEEE systems]{Optimal islanding results generated by the proposed QAOA framework on the IEEE 39-, 57-, 73-, and 89-bus test systems. Bus colors indicate
the island assignment of each bus, and the red dashed lines denote the transmission interfaces removed to form the final
islands. The resulting islands consist mainly of neighboring buses and require only a few transmission lines
to be disconnected.}
\label{fig:island_viz_ieee_large}
\end{figure*}

\subsection{Benchmark Performance Comparison}

This subsection evaluates the proposed hybrid quantum-classical islanding workflow from two complementary
perspectives: the structure of the resulting islands and the quantum resources required to obtain them. The
representative islanding solutions obtained from the QAOA-based workflow are listed in
Table~\ref{tab:islanding_partitions}, and Figs.~\ref{fig:island_viz_ieee_small}
and~\ref{fig:island_viz_ieee_large} show the corresponding partitions on the IEEE visualization set. The
resource and runtime comparison for the reduced encoding $\mathcal{E}_1$ is reported in
Table~\ref{tab:lagrangian_bench}, covering the IEEE 9-, 14-, 24-, 30-, 39-, 57-, and 73-bus systems. The
corresponding comparison for $\mathcal{E}_2$ is reported in
Table~\ref{tab:sf_bench}, covering the same seven systems and the IEEE~89-bus system.

The island assignments in Table~\ref{tab:islanding_partitions} satisfy the prescribed island counts while
covering each bus exactly once. More importantly, the selected partitions retain coherent electrical areas
across both two-island and three-island cases. This is significant for controlled islanding because the
combinatorial encoding alone does not guarantee that a low-cut solution will have a physically interpretable
network structure.

The structural coherence of the resulting islands is evident from the topological views in
Figs.~\ref{fig:island_viz_ieee_small} and~\ref{fig:island_viz_ieee_large}. The bus colors form compact regions on the graph, and the red dashed interfaces
are concentrated on a small set of boundary branches. These patterns show that the selected solutions define
clear and physically interpretable island boundaries across all systems, from the IEEE 9-bus case to the IEEE
89-bus case.

The practical value of these physically interpretable partitions depends on whether they can be obtained
within the resource limits of available quantum hardware. Under $\mathcal{E}_1$, removing the explicit
penalty-variable overhead yields substantial QUBO reductions without sacrificing the best recovered
objective value. As summarized in
Table~\ref{tab:concrete_qubits} and visualized in Fig.~\ref{fig:qubit_scaling_comparison}, the variable counts decrease from 23 to 9 on IEEE~9-bus, 29 to 14 on
IEEE~14-bus, 76 to 40 on IEEE~24-bus, 55 to 29 on IEEE~30-bus, 123 to 78 on IEEE~39-bus, 78 to 52 on
IEEE~57-bus, and 156 to 108 on IEEE~73-bus. These changes correspond to reductions of $60.9\%$, $51.7\%$,
$47.4\%$, $47.3\%$, $36.6\%$, $33.3\%$, and $30.8\%$, respectively. The cut value remains unchanged on
six of the seven systems. On IEEE~24-bus, the Lagrangian runs obtain the lower cut value of $776.206$
across IBM, Rigetti, and IQM, compared with $824.125$ for the non-Lagrangian IBM run.
No $\mathcal{E}_1$ result is reported for IEEE~89-bus because its Lagrangian formulation requires 178
qubits, which exceeds the 156-qubit capacity of the available IBM backend. In contrast, $\mathcal{E}_2$
reduces this requirement to 154 qubits and makes the IEEE~89-bus run possible.

The benefit of this variable reduction extends beyond the abstract QUBO dimension because each removed
variable also reduces the burden placed on circuit construction and compilation. Accordingly, the smaller
Lagrangian QUBOs generally translate into shallower compiled IBM circuits. Relative to the non-Lagrangian mode,
the Lagrangian circuit depth decreases from 473 to 31 on IEEE~9-bus, 582 to 87 on IEEE~14-bus, 891 to 127
on IEEE~24-bus, 773 to 64 on IEEE~30-bus, 977 to 881 on IEEE~39-bus, 937 to 804 on IEEE~57-bus, and 1093
to 731 on IEEE~73-bus. The reductions are particularly large for the 9-, 14-, 24-, and 30-bus cases, ranging
from $85.1\%$ to $93.4\%$. The more moderate reductions for the larger cases indicate that circuit depth is
also shaped by QAOA depth and backend compilation. This backend dependence is visible in the Lagrangian
runs. For example, the IEEE~57-bus circuit has depth 804 on IBM, 181 on Rigetti, and 159 on IQM, although
all three runs use $p=2$ and $1500$ shots and return the same cut value.

Resource savings are meaningful only if the shallower circuits continue to produce valid island assignments
with sufficient frequency. The feasible-sample rates therefore complement the cut and depth comparisons by
measuring how often each formulation recovers a valid solution. On IBM, the $\mathcal{E}_1$ Lagrangian mode increases
the feasible-sample rate from $0.150$ to $0.380$ on IEEE~9-bus, $0.110$ to $0.210$ on IEEE~14-bus, $0.004$
to $0.007$ on IEEE~24-bus, $0.190$ to $0.230$ on IEEE~30-bus, $0.086$ to $0.104$ on IEEE~39-bus, and
$0.040$ to $0.057$ on IEEE~73-bus. The IEEE~57-bus case reaches a feasibility rate of $1.000$ in every
reported run. The cross-backend results show larger variation for the smaller systems. For IEEE~9-bus, the
Lagrangian feasible rate ranges from $0.150$ on IonQ to $0.380$ on IBM, with AQT, IQM, and Rigetti producing
$0.260$, $0.200$, and $0.240$, respectively. For IEEE~14-bus, IQM gives the highest rate of $0.270$, followed
by Rigetti at $0.230$, IBM at $0.210$, and IonQ at $0.150$. The three reported Lagrangian backends for
IEEE~24-bus produce rates between $0.007$ and $0.017$, which identifies this case as the most difficult one
for direct feasible sampling under $\mathcal{E}_1$. By contrast, the rates for IEEE~30-bus are tightly grouped
between $0.227$ and $0.253$, and all reported IEEE~57-bus runs reach $1.000$.

The consistent IBM improvements do not imply that the Lagrangian treatment affects every hardware platform
in the same way. Paired results from the other backends reveal a more nuanced feasibility pattern. On IEEE~9-bus, it raises
the rate on IBM and IQM, leaves Rigetti unchanged, and lowers IonQ from $0.180$ to $0.150$. On IEEE~14-bus,
it improves the rate on all four backends with paired results. These differences show that feasibility depends
on both the formulation and the backend-specific execution characteristics. In contrast, the best reported cut
for a given system is identical across all Lagrangian backends. This separation indicates that the backend
mainly affects the frequency of feasible observations, while the best recovered objective value remains stable.

The observed differences in feasibility cannot be explained by compiled circuit depth alone. Since native
gates, connectivity, and compilation procedures
differ across devices, depth values are most meaningful as within-backend resource indicators rather than as
direct measures of hardware quality.
Under $\mathcal{E}_1$, IBM gives the smallest Lagrangian depth for
IEEE~9-bus at 31, while IonQ gives the smallest depth for IEEE~14-bus at 25. The IEEE~24-bus depths are nearly
identical across IBM, IQM, and Rigetti at 127 to 130. For the deeper $p=2$ IEEE~57-bus circuits, the compiled
depth is 804 on IBM, 181 on Rigetti, and 159 on IQM. Because these depths result from backend-specific
compilation, they are suitable for assessing circuit complexity within each backend but not for directly
ranking the processors by accuracy or execution speed.

The $\mathcal{E}_2$ results in Table~\ref{tab:sf_bench} address the hardware limitation encountered by
$\mathcal{E}_1$ through a more compact assignment representation. Combining the backend results in
Table~\ref{tab:sf_bench} with the qubit counts in Table~\ref{tab:concrete_qubits} shows that the Lagrangian
treatment reduces the QUBO-variable counts from 20 to 6 on IEEE~9-bus, 29 to 9 on IEEE~14-bus, 62 to 26 on IEEE~24-bus,
50 to 24 on IEEE~30-bus, 103 to 58 on IEEE~39-bus, 76 to 50 on IEEE~57-bus, 102 to 54 on IEEE~73-bus,
and 208 to 154 on IEEE~89-bus. These changes represent reductions of $70.0\%$, $69.0\%$, $58.1\%$,
$52.0\%$, $43.7\%$, $34.2\%$, $47.1\%$, and $26.0\%$, respectively. Most importantly, the 154-qubit
Lagrangian model fits within the available 156-qubit IBM backend and enables execution of the IEEE~89-bus
case. Its 208-qubit non-Lagrangian counterpart remains beyond the hardware limit and therefore has no
companion run in Table~\ref{tab:sf_bench}. The corresponding IEEE~89-bus topology in
Fig.~\ref{fig:island_viz_ieee_large} shows that this hardware-compatible formulation produces three compact
islands separated by a limited set of boundary branches, thereby linking the qubit reduction to a physically
interpretable islanding solution.

The reduction in model dimension produces a corresponding decrease in compiled circuit depth on IBM.
Relative to the non-Lagrangian mode, the Lagrangian depth decreases from 241 to 25 on IEEE~9-bus, 467 to
77 on IEEE~14-bus, 784 to 94 on IEEE~24-bus, 685 to 62 on IEEE~30-bus, 848 to 720 on IEEE~39-bus,
849 to 544 on IEEE~57-bus, and 917 to 719 on IEEE~73-bus. The reductions exceed $83\%$ for the 9-, 14-,
24-, and 30-bus systems, while the larger systems show more moderate decreases of $15.1\%$ to $35.9\%$.
Compared with the Lagrangian $\mathcal{E}_1$ circuits on IBM, $\mathcal{E}_2$ further lowers depth for every
shared system, with particularly notable changes from 881 to 720 on IEEE~39-bus and from 804 to 544 on
IEEE~57-bus. This consistent reduction confirms that $\mathcal{E}_2$ removes circuit overhead beyond the
savings already provided by the Lagrangian formulation.

The principal performance distinction between the two encodings appears in the probability of obtaining
feasible samples. Under the Lagrangian mode, $\mathcal{E}_2$ reaches a feasibility rate of $1.000$ on every
reported backend for IEEE~9-bus, IEEE~30-bus, and IEEE~57-bus. This is a substantial improvement over
$\mathcal{E}_1$, whose corresponding rates range from $0.150$ to $0.380$ on IEEE~9-bus and from $0.227$
to $0.253$ on IEEE~30-bus. The contrast is strongest for IEEE~24-bus. Encoding $\mathcal{E}_2$ raises the
cross-backend feasibility range from $0.007$--$0.017$ under $\mathcal{E}_1$ to $0.411$--$0.550$ under
$\mathcal{E}_2$. On IBM, the same improvement is visible for IEEE~39-bus and IEEE~73-bus, where the rates
increase from $0.104$ to $0.762$ and from $0.057$ to $1.000$. These results indicate that $\mathcal{E}_2$
removes redundant states from the search space and makes valid island assignments substantially more
accessible to QAOA sampling.

Although $\mathcal{E}_2$ improves feasibility overall, the remaining variation across backends shows that
sampling behavior is still hardware dependent. For IEEE~14-bus, the Lagrangian rate ranges from $0.290$ on
IonQ to $0.520$ on IQM, with values of $0.300$, $0.390$, and $0.400$ on AQT, Rigetti, and IBM. For
IEEE~24-bus, Rigetti gives the highest rate of $0.550$, followed by IQM at $0.498$, IonQ at $0.458$, and IBM
at $0.411$. The paired Lagrangian and non-Lagrangian results also show that qubit reduction does not guarantee
a higher feasible rate on every device. On IEEE~14-bus, the Lagrangian rate increases on IQM and Rigetti but
decreases on IBM and IonQ. On IBM, the non-Lagrangian rate is also higher for IEEE~24-bus and IEEE~39-bus.
Nevertheless, the best cut value remains identical across formulations and backends for every paired
$\mathcal{E}_2$ case. The robust benefit of the Lagrangian treatment is therefore lower resource demand with
preserved objective quality, while its effect on feasible-sample frequency depends on the backend.

The cross-backend depth values reinforce the distinction between formulation complexity and hardware-specific
compilation. In the Lagrangian runs, the IEEE~9-bus depth ranges from 19 on AQT and IonQ to 34 on Rigetti,
and the IEEE~14-bus depth ranges from 25 on AQT and IonQ to 77 on IBM. The non-IBM depths for IEEE~24-bus
are closely grouped between 55 and 58, compared with 94 on IBM. For IEEE~57-bus, IQM and Rigetti compile
to depths of 135 and 157, compared with 544 on IBM, while all three backends achieve a feasibility rate of
$1.000$ and the same cut value. The IEEE~89-bus experiment extends the hardware demonstration to the largest
executed system using $p=7$, $10000$ shots, and $I_{\max}=3$. Its IBM circuit has depth 865 and achieves a
feasibility rate of $0.535$. Since only IBM results are available for the 39-, 73-, and 89-bus cases, these rows
demonstrate executable scale but do not support a complete comparison among backends.

\begin{table*}[!t]
\centering
\footnotesize
\setlength{\tabcolsep}{3pt}
\renewcommand{\arraystretch}{0.95}
\caption[Backend benchmark results under reduced assignment encoding]{Backend-level benchmark results under encoding $\mathcal{E}_1$, with and without Lagrangian terms.}
\label{tab:lagrangian_bench}
\begin{tabular}{@{}lllcccccc@{}}
\toprule
\textbf{Case} & \textbf{Backend} & \textbf{QUBO mode} & $p$ & $S$ & $I_{\max}$ &
\makecell{\textbf{Circuit}\\\textbf{depth}} &
$\cutval$ & \makecell{\textbf{Feas.}\\\textbf{rate}} \\
\midrule
\multirow{9}{*}{IEEE 9-bus}
 & IBM & Lagrangian & 1 & 100 & 2 & 31 & 1.238 & 0.380 \\ 
 & Rigetti & Lagrangian & 1 & 100 & 2 & 52 & 1.238 & 0.240 \\ 
 & IQM & Lagrangian & 1 & 100 & 2 & 44 & 1.238 & 0.200 \\ 
 & AQT & Lagrangian & 1 & 100 & 2 & 31 & 1.238 & 0.260 \\ 
 & IonQ & Lagrangian & 1 & 100 & 2 & 31 & 1.238 & 0.150 \\ 
 & IBM & w/o Lagrangian & 1 & 100 & 2 & 473 & 1.238 & 0.150 \\ 
 & Rigetti & w/o Lagrangian & 1 & 100 & 2 & 88 & 1.238 & 0.240 \\ 
 & IQM & w/o Lagrangian & 1 & 100 & 2 & 86 & 1.238 & 0.150 \\ 
 & IonQ & w/o Lagrangian & 1 & 100 & 2 & 85 & 1.238 & 0.180 \\ 
\midrule
\multirow{9}{*}{IEEE 14-bus}
 & IBM     & Lagrangian     & 1 & 100  & 2 &    87 &  88.241 & 0.210 \\
 & IonQ    & Lagrangian     & 1 & 100  & 2 &    25 &  88.241 & 0.150 \\
 & IQM     & Lagrangian     & 1 & 100  & 2 &    28 &  88.241 & 0.270 \\
 & Rigetti & Lagrangian     & 1 & 100  & 2 &    31 &  88.241 & 0.230 \\
 & IBM     & w/o Lagrangian & 1 & 100  & 2 &   582 &  88.241 & 0.110 \\
 & IonQ    & w/o Lagrangian & 1 & 100  & 2 &   109 &  88.241 & 0.080 \\
 & IQM     & w/o Lagrangian & 1 & 100  & 2 &   110 &  88.241 & 0.140 \\
 & Rigetti & w/o Lagrangian & 1 & 100  & 2 &   112 &  88.241 & 0.190 \\
\midrule
\multirow{4}{*}{IEEE 24-bus}
 & IBM & Lagrangian & 1 & 1000 & 2 & 127 & 776.206 & 0.007 \\ 
 & Rigetti & Lagrangian & 1 & 1000 & 2 & 130 & 776.206 & 0.017 \\ 
 & IQM & Lagrangian & 1 & 1000 & 2 & 128 & 776.206 & 0.017 \\ 
 & IBM & w/o Lagrangian & 1 & 1000 & 2 & 891 & 824.125 & 0.004 \\ 
\midrule
\multirow{5}{*}{IEEE 30-bus}
 & IBM & Lagrangian & 1 & 300 & 2 & 64 & 16.863 & 0.230 \\ 
 & Rigetti & Lagrangian & 1 & 300 & 2 & 79 & 16.863 & 0.247 \\ 
 & IQM & Lagrangian & 1 & 300 & 2 & 72 & 16.863 & 0.227 \\ 
 & IonQ & Lagrangian & 1 & 300 & 2 & 64 & 16.863 & 0.253 \\ 
 & IBM & w/o Lagrangian & 1 & 300 & 2 & 773 & 16.863 & 0.190 \\ 
\midrule
\multirow{2}{*}{IEEE 39-bus}
 & IBM & Lagrangian & 2 & 2000 & 2 & 881 & 228.993 & 0.104 \\ 
 & IBM & w/o Lagrangian & 2 & 2000 & 2 & 977 & 228.993 & 0.086 \\ 
\midrule
\multirow{4}{*}{IEEE 57-bus}
 & IBM & Lagrangian & 2 & 1500 & 2 & 804 & 128.239 & 1.000 \\ 
 & Rigetti & Lagrangian & 2 & 1500 & 2 & 181 & 128.239 & 1.000 \\ 
 & IQM & Lagrangian & 2 & 1500 & 2 & 159 & 128.239 & 1.000 \\ 
 & IBM & w/o Lagrangian & 2 & 1500 & 2 & 937 & 128.239 & 1.000 \\ 
\midrule
\multirow{2}{*}{IEEE 73-bus}
 & IBM & Lagrangian & 3 & 2000 & 2 & 731 & 284.543 & 0.057 \\ 
 & IBM & w/o Lagrangian & 3 & 2000 & 2 & 1093 & 284.543 & 0.040 \\ 
\bottomrule
\end{tabular}
\end{table*}

\begin{table*}[!t]
\centering
\footnotesize
\setlength{\tabcolsep}{3pt}
\renewcommand{\arraystretch}{0.95}
\caption[Backend benchmark results under encoding $\mathcal{E}_2$]{Backend-level benchmark results under encoding $\mathcal{E}_2$, with and without Lagrangian terms.}
\label{tab:sf_bench}
\begin{tabular}{@{}lllcccccc@{}}
\toprule
\textbf{Case} & \textbf{Backend} & \textbf{QUBO mode} & $p$ & $S$ & $I_{\max}$ &
\makecell{\textbf{Circuit}\\\textbf{depth}} &
$\cutval$ & \makecell{\textbf{Feas.}\\\textbf{rate}} \\
\midrule
\multirow{9}{*}{IEEE 9-bus}
 & IBM     & Lagrangian     & 1 & 100  & 2 &    25 &   1.238 & 1.000 \\
 & AQT     & Lagrangian     & 1 & 100  & 2 &    19 &   1.238 & 1.000 \\
 & IonQ    & Lagrangian     & 1 & 100  & 2 &    19 &   1.238 & 1.000 \\
 & IQM     & Lagrangian     & 1 & 100  & 2 &    28 &   1.238 & 1.000 \\
 & Rigetti & Lagrangian     & 1 & 100  & 2 &    34 &   1.238 & 1.000 \\
 & IBM     & w/o Lagrangian & 1 & 100  & 2 &   241 &   1.238 & 1.000 \\
 & IonQ    & w/o Lagrangian & 1 & 100  & 2 &    67 &   1.238 & 1.000 \\
 & IQM     & w/o Lagrangian & 1 & 100  & 2 &    68 &   1.238 & 1.000 \\
 & Rigetti & w/o Lagrangian & 1 & 100  & 2 &    70 &   1.238 & 1.000 \\
\midrule
\multirow{9}{*}{IEEE 14-bus}
 & IBM     & Lagrangian     & 1 & 100  & 2 &    77 &  88.241 & 0.400 \\
 & AQT     & Lagrangian     & 1 & 100  & 2 &    25 &  88.241 & 0.300 \\
 & IonQ    & Lagrangian     & 1 & 100  & 2 &    25 &  88.241 & 0.290 \\
 & IQM     & Lagrangian     & 1 & 100  & 2 &    28 &  88.241 & 0.520 \\
 & Rigetti & Lagrangian     & 1 & 100  & 2 &    31 &  88.241 & 0.390 \\
 & IBM     & w/o Lagrangian & 1 & 100  & 2 &   467 &  88.241 & 0.450 \\
 & IonQ    & w/o Lagrangian & 1 & 100  & 2 &   109 &  88.241 & 0.340 \\
 & IQM     & w/o Lagrangian & 1 & 100  & 2 &   110 &  88.241 & 0.420 \\
 & Rigetti & w/o Lagrangian & 1 & 100  & 2 &   112 &  88.241 & 0.270 \\
\midrule
\multirow{5}{*}{IEEE 24-bus}
 & IBM     & Lagrangian     & 1 & 1000 & 2 &   94 & 776.206 & 0.411 \\
 & IonQ    & Lagrangian     & 1 & 1000 & 2 &    55 & 776.206 & 0.458 \\
 & IQM     & Lagrangian     & 1 & 1000 & 2 &    56 & 776.206 & 0.498 \\
 & Rigetti & Lagrangian     & 1 & 1000 & 2 &    58 & 776.206 & 0.550 \\
 & IBM     & w/o Lagrangian & 1 & 1000 & 2 &  784 & 776.206 & 0.542 \\
\midrule
\multirow{6}{*}{IEEE 30-bus}
 & IBM     & Lagrangian     & 1 & 300  & 2 &   62 &  16.863 & 1.000 \\
 & IonQ    & Lagrangian     & 1 & 300  & 2 &    34 &  16.863 & 1.000 \\
 & IQM     & Lagrangian     & 1 & 300  & 2 &    42 &  16.863 & 1.000 \\
 & Rigetti & Lagrangian     & 1 & 300  & 2 &    49 &  16.863 & 1.000 \\
 & IBM     & w/o Lagrangian & 1 & 300  & 2 &  685 &  16.863 & 1.000 \\
 & IQM     & w/o Lagrangian & 1 & 300  & 2 &   215 &  16.863 & 1.000 \\
\midrule
\multirow{2}{*}{IEEE 39-bus}
 & IBM & Lagrangian     & 2 & 2000 & 2 &  720 & 228.993 & 0.762 \\
 & IBM & w/o Lagrangian & 2 & 2000 & 2 & 848 & 228.993 & 0.835 \\
\midrule
\multirow{4}{*}{IEEE 57-bus}
 & IBM     & Lagrangian     & 2 & 1500 & 2 &   544 & 128.239 & 1.000 \\
 & IQM     & Lagrangian     & 2 & 1500 & 2 &   135 & 128.239 & 1.000 \\
 & Rigetti & Lagrangian     & 2 & 1500 & 2 &   157 & 128.239 & 1.000 \\
 & IBM     & w/o Lagrangian & 2 & 1500 & 2 & 849 & 128.239 & 1.000 \\
\midrule
\multirow{2}{*}{IEEE 73-bus}
 & IBM & Lagrangian     & 3 & 2000 & 2 &  719 & 284.543 & 1.000 \\
 & IBM & w/o Lagrangian & 3 & 2000 & 2 & 917 & 284.543 & 1.000 \\
\midrule
{IEEE 89-bus}
 & IBM & Lagrangian     & 7 & 10000 & 3 & 865 & 3790.938 & 0.535 \\
\bottomrule
\end{tabular}
\end{table*}

The power-support checks for the smaller benchmark systems were established in the preceding REGRID-QAOA
study~\cite{jiang2026regrid}. To place the new IEEE~73- and 89-bus results in the same context, the present
analysis applies a consistent MATPOWER calculation to all eight systems. For each island, the active-power
margin is defined as $\Delta P=P_{g,\max}-P_d$, and the reactive-power margin is defined as
$\Delta Q=Q_{\max}-Q_d$. Positive values indicate that the aggregate generation capability exceeds the
corresponding island demand. The calculation uses the selected partitions associated with the backend results
and sequential bus renumbering consistent with the islanding code.

\begin{figure}[!tbp]
\centering
\includegraphics[width=\columnwidth]{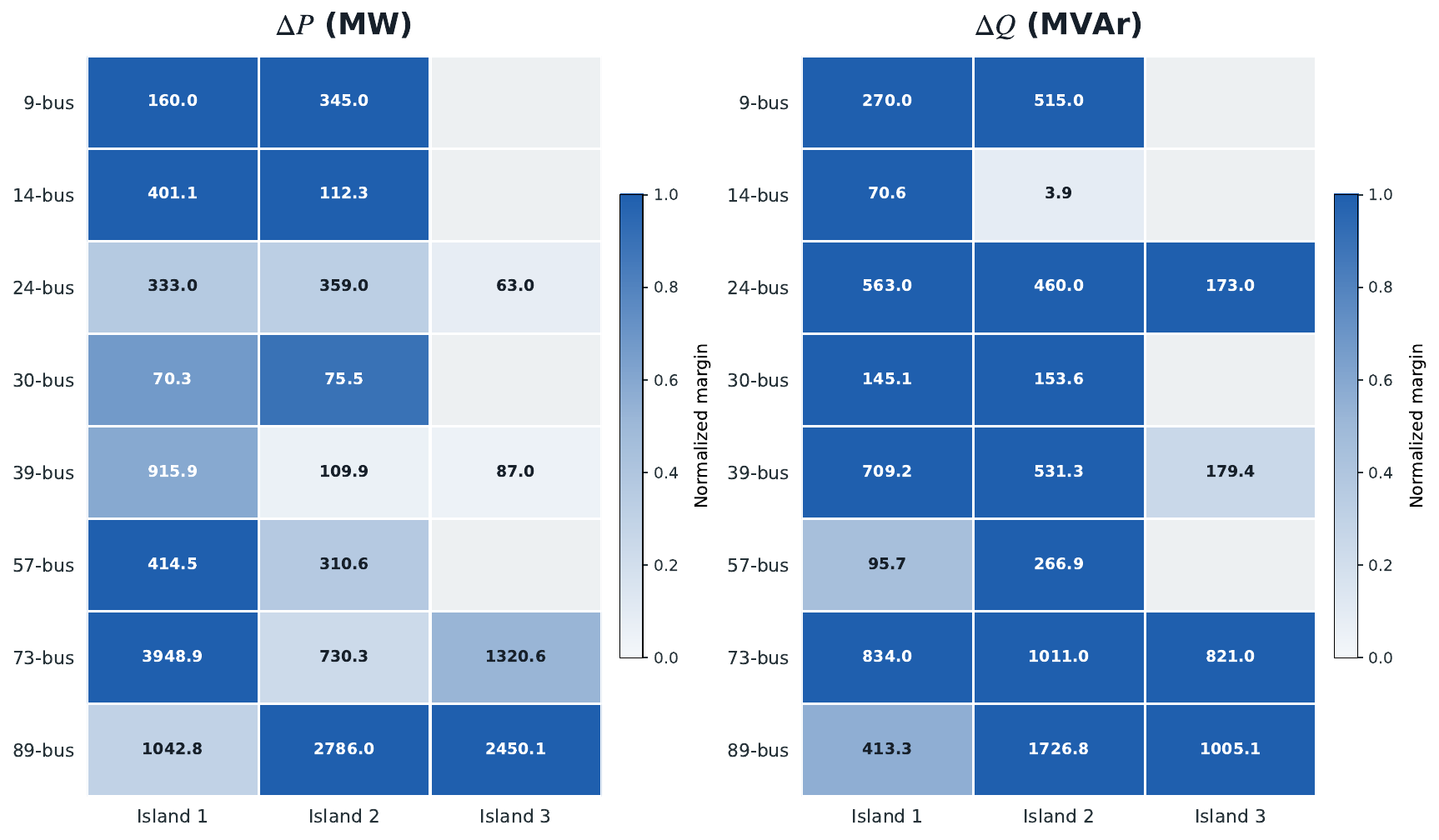}
\caption[Power-support margins across all benchmark systems]{Island-level active- and reactive-power support
margins across all benchmark systems. Each cell reports the exact margin in MW or MVAr, while the
white-to-blue color scale represents the corresponding margin normalized by island demand. Darker cells
indicate greater relative headroom. Gray cells denote an unavailable third island in a two-island system.}
\label{fig:power_margin_heatmap}
\end{figure}

The power-support margins complement the preceding topological and backend analyses by assessing whether
each selected island has sufficient aggregate generation capability. As summarized in
Fig.~\ref{fig:power_margin_heatmap}, every island retains positive active- and reactive-power margins. The most constrained active-power cases occur in IEEE~24-bus Island~3 and IEEE~30-bus
Island~1, with margins of $63.0$ and $70.3~\mathrm{MW}$. The tightest reactive-power case is IEEE~14-bus
Island~2, whose margin is only $3.9~\mathrm{MVAr}$. These small positive values identify the islands that are
most sensitive to load growth, generator unavailability, or modeling uncertainty. In contrast, the larger
active-power margins in IEEE~39- and 57-bus indicate greater aggregate support, although the normalized colors
show that absolute margin alone does not determine relative adequacy across systems of different sizes.

The positive margins obtained for the two largest systems demonstrate that the resource savings of the
proposed method translate into greater quantum-applicable problem scale without compromising island-level
power support. For IEEE~73-bus, the active-power margins are $3948.9$, $730.3$, and
$1320.6~\mathrm{MW}$, and the reactive-power margins are $834$, $1011$, and $821~\mathrm{MVAr}$. For
IEEE~89-bus, the corresponding active-power margins are $1042.8$, $2786.0$, and $2450.1~\mathrm{MW}$,
while the reactive-power margins are $413.3$, $1726.8$, and $1005.1~\mathrm{MVAr}$. More importantly,
$\mathcal{E}_2$ reduces the IEEE~89-bus formulation to 154 qubits, allowing it to run on the available
156-qubit IBM backend, while the resulting partition retains the coherent topology in
Fig.~\ref{fig:island_viz_ieee_large} and positive power-support margins in every island. In combination with
the hardware results in Table~\ref{tab:sf_bench}, this agreement between resource feasibility and physical
adequacy shows that the proposed method can scale controlled-islanding applications to larger power systems
on present quantum platforms.

\subsection{Ablation Study}
\label{sec:ablation_study}

The ablation study isolates how the principal formulation choices affect resource requirements,
computational efficiency, feasibility, and solution quality. The formulation labels are consistent with
those introduced in Table~\ref{tab:complexity_summary}. The additional variant (i)+Lag. combines the full
one-hot encoding of Formulation~(i) with Lagrangian constraint handling. Comparisons between formulations
that differ in only one design choice isolate the contributions of the encoding and constraint-handling
strategies. The encoding effect follows from comparing
Formulations~(ii) and~(iii) under direct QUBO penalties and Formulations~(iv) and~(v) under Lagrangian
constraint handling. The effect of the slack-free Lagrangian scheme is evaluated by comparing
Formulations~(ii) and~(iv) under $\mathcal{E}_1$ and Formulations~(iii) and~(v) under $\mathcal{E}_2$.
Formulation~(i) and (i)+Lag. provide the corresponding constraint-handling comparison for the full one-hot
encoding.

Table~\ref{tab:ablation_resource} summarizes the factorial ablation results across the benchmark power
systems. All variants for a given test system use the same system-specific QAOA layer depth, shot count, and
optimizer budget reported in Table~\ref{tab:exp_setup}. The table consolidates the six complete comparisons
for the IEEE 9-, 14-, 24-, 30-, 39-, and 57-bus systems. The Qubit count and QUBO terms columns give
the total logical-qubit and QUBO-term counts, respectively. Circuit depth denotes the depth of the
transpiled circuit. The two-qubit (2Q) gates and Total gates columns give the numbers of two-qubit gates and
all gates in the transpiled circuit, respectively. The Feas. rate column reports the fraction of samples
that are feasible after post-processing. Cut is the resulting cut value, and $t/t_q$ gives the overall
runtime and quantum-backend runtime in seconds, respectively.

For each system, the parenthetical ratios are obtained by normalizing the QUBO-term count, circuit depth,
and gate counts against the corresponding values for the proposed Formulation~(v). Formulation~(v) therefore
has a ratio of $1.00$, while values greater than $1.00$ quantify the resource overhead of the alternative
formulations.

\begin{table*}[!tbp]
\centering
\footnotesize
\setlength{\tabcolsep}{2.5pt}
\renewcommand{\arraystretch}{0.92}
\caption{Factorial ablation across benchmark power systems.}
\label{tab:ablation_resource}
\resizebox{\textwidth}{!}{%
\begin{tabular}{@{}ccrrrrrrrr@{}}
\toprule
\textbf{IEEE} & \textbf{Form.} & \makecell{\textbf{Qubit}\\\textbf{count}}
& \textbf{QUBO terms (ratio)} & \makecell{\textbf{Circuit}\\\textbf{depth (ratio)}} & \textbf{2Q gates (ratio)}
& \textbf{Total gates (ratio)}
& \makecell{\textbf{Feas.}\\\textbf{rate}} & \textbf{Cut} & $\boldsymbol{t/t_q}$ \\
\midrule
\multirow{6}{*}{9} & (i) & 32 & 205 (17.08) & 872 (34.88) & 1325 (63.10) & 5038 (37.32) & 0.170 & 1.238 & 40.1/4.0 \\
 & (i)+Lag. & 18 & 37 (3.08) & 336 (13.44) & 204 (9.71) & 865 (6.41) & 0.220 & 1.238 & 39.8/4.0 \\
 & (ii) & 23 & 151 (12.58) & 473 (18.92) & 120 (5.71) & 753 (5.58) & 0.270 & 1.238 & 34.2/3.0 \\
 & (iii) & 20 & 96 (8.00) & 241 (9.64) & 87 (4.14) & 527 (3.90) & 1.000 & 1.238 & 30.8/3.0 \\
 & (iv) & 9 & 21 (1.75) & 31 (1.24) & 57 (2.71) & 299 (2.21) & 0.350 & 1.238 & 32.3/2.0 \\
 & \textbf{(v)} & \textbf{6} & \textbf{12 (1.00)} & \textbf{25 (1.00)} & \textbf{21 (1.00)} & \textbf{135 (1.00)} & \textbf{1.000} & \textbf{1.238} & \textbf{28.7/2.0} \\
\addlinespace[2pt]
\multirow{6}{*}{14} & (i) & 44 & 402 (22.33) & 1712 (22.23) & 2528 (93.63) & 9690 (52.66) & 0.160 & 88.241 & 33.7/6.0 \\
 & (i)+Lag. & 24 & 74 (4.11) & 421 (5.47) & 316 (11.70) & 1353 (7.35) & 0.170 & 88.241 & 32.1/5.0 \\
 & (ii) & 32 & 312 (17.33) & 582 (7.56) & 257 (9.52) & 992 (5.39) & 0.110 & 88.241 & 41.4/4.0 \\
 & (iii) & 29 & 226 (12.56) & 467 (6.06) & 160 (5.93) & 458 (2.49) & 0.370 & 88.241 & 32.7/4.0 \\
 & (iv) & 12 & 31 (1.72) & 87 (1.13) & 77 (2.85) & 397 (2.16) & 0.210 & 88.241 & 32.9/2.0 \\
 & \textbf{(v)} & \textbf{9} & \textbf{18 (1.00)} & \textbf{77 (1.00)} & \textbf{27 (1.00)} & \textbf{184 (1.00)} & \textbf{0.360} & \textbf{88.241} & \textbf{32.0/2.0} \\
\addlinespace[2pt]
\multirow{6}{*}{24} & (i) & 96 & 1302 (15.14) & 2864 (30.47) & 8938 (36.04) & 33601 (28.05) & 0.025 & 860.528 & 125.3/6.0 \\
 & (i)+Lag. & 60 & 216 (2.51) & 1144 (12.17) & 1743 (7.03) & 6572 (5.49) & 0.031 & 817.378 & 121.8/6.0 \\
 & (ii) & 76 & 1105 (12.85) & 891 (9.48) & 1578 (6.36) & 3914 (3.27) & 0.010 & 776.206 & 83.2/6.0 \\
 & (iii) & 62 & 909 (10.57) & 784 (8.34) & 612 (2.47) & 2228 (1.86) & 0.524 & 776.206 & 91.5/5.0 \\
 & (iv) & 40 & 244 (2.84) & 127 (1.35) & 478 (1.93) & 2056 (1.72) & 0.015 & 776.206 & 74.0/3.0 \\
 & \textbf{(v)} & \textbf{26} & \textbf{86 (1.00)} & \textbf{94 (1.00)} & \textbf{248 (1.00)} & \textbf{1198 (1.00)} & \textbf{0.454} & \textbf{776.206} & \textbf{54.4/3.0} \\
\addlinespace[2pt]
\multirow{6}{*}{30} & (i) & 84 & 1531 (30.62) & 5165 (83.31) & 9695 (76.34) & 36324 (54.38) & 0.157 & 16.863 & 59.6/8.0 \\
 & (i)+Lag. & 58 & 137 (2.74) & 813 (13.11) & 1017 (8.01) & 4041 (6.05) & 0.180 & 16.863 & 55.0/8.0 \\
 & (ii) & 55 & 1067 (21.34) & 773 (12.47) & 531 (4.18) & 2038 (3.05) & 0.173 & 16.863 & 41.0/7.0 \\
 & (iii) & 50 & 851 (17.02) & 685 (11.05) & 402 (3.17) & 1700 (2.54) & 1.000 & 16.863 & 39.7/7.0 \\
 & (iv) & 29 & 79 (1.58) & 64 (1.03) & 343 (2.70) & 1455 (2.18) & 0.210 & 16.863 & 39.5/4.0 \\
 & \textbf{(v)} & \textbf{24} & \textbf{50 (1.00)} & \textbf{62 (1.00)} & \textbf{127 (1.00)} & \textbf{668 (1.00)} & \textbf{1.000} & \textbf{16.863} & \textbf{36.9/4.0} \\
\addlinespace[2pt]
\multirow{6}{*}{39} & (i) & 156 & 5306 (23.37) & 13344 (18.53) & 44154 (19.83) & 164783 (18.52) & 0.105 & 228.993 & 1017.5/10.0 \\
 & (i)+Lag. & 117 & 420 (1.85) & 4999 (6.94) & 8667 (3.89) & 31194 (3.51) & 0.129 & 228.993 & 939.9/9.0 \\
 & (ii) & 123 & 3895 (17.16) & 977 (1.36) & 6499 (2.92) & 26872 (3.02) & 0.101 & 228.993 & 681.4/8.0 \\
 & (iii) & 103 & 2121 (9.34) & 848 (1.18) & 4011 (1.80) & 24909 (2.80) & 0.882 & 228.993 & 605.6/8.0 \\
 & (iv) & 78 & 481 (2.12) & 881 (1.22) & 3418 (1.53) & 23754 (2.67) & 0.110 & 228.993 & 655.7/6.0 \\
 & \textbf{(v)} & \textbf{58} & \textbf{227 (1.00)} & \textbf{720 (1.00)} & \textbf{2227 (1.00)} & \textbf{8896 (1.00)} & \textbf{0.758} & \textbf{228.993} & \textbf{440.4/6.0} \\
\addlinespace[2pt]
\multirow{6}{*}{57} & (i) & 130 & 4122 (38.89) & 28350 (52.11) & 58024 (90.95) & 214315 (71.73) & 0.740 & 131.917 & 669.6/9.0 \\
 & (i)+Lag. & 104 & 192 (1.81) & 2107 (3.87) & 3753 (5.88) & 14212 (4.76) & 0.780 & 130.009 & 658.2/9.0 \\
 & (ii) & 78 & 2692 (25.40) & 937 (1.72) & 2946 (4.62) & 11210 (3.75) & 1.000 & 128.239 & 312.6/8.0 \\
 & (iii) & 76 & 2561 (24.16) & 849 (1.56) & 1349 (2.11) & 10818 (3.62) & 1.000 & 128.239 & 283.3/8.0 \\
 & (iv) & 52 & 122 (1.15) & 804 (1.48) & 982 (1.54) & 4200 (1.41) & 1.000 & 128.239 & 363.3/6.0 \\
 & \textbf{(v)} & \textbf{50} & \textbf{106 (1.00)} & \textbf{544 (1.00)} & \textbf{638 (1.00)} & \textbf{2988 (1.00)} & \textbf{1.000} & \textbf{128.239} & \textbf{276.5/6.0} \\
\bottomrule
\end{tabular}
}
\end{table*}

The two encoding strategies provide distinct and successive resource reductions under the same direct
constraint-penalty treatment. Encoding~$\mathcal{E}_1$ first contracts the physical network to the working
graph and replaces the $K$ one-hot variables at each retained bus with $K-1$ explicit variables. The final
island assignment is inferred from the remaining variables. Relative to the full one-hot Formulation~(i),
Formulation~(ii) reduces the logical-qubit count for all six systems by $20.8\%$ to $40.0\%$. The reduction
ranges from 96 to 76 qubits for IEEE~24-bus and from 130 to 78 qubits for IEEE~57-bus. The narrower
register is accompanied by a systematic reduction in compiled-circuit resources. Across the six systems,
the circuit depth is reduced by factors of 1.84 to 30.26, the two-qubit-gate count by factors of 5.66 to
19.70, and the total gate count by factors of 6.13 to 19.12. For IEEE~24-bus, the circuit depth decreases
from 2864 to 891, the two-qubit-gate count from 8938 to 1578, and the total gate count from 33,601 to 3914.
For IEEE~39-bus, the corresponding quantities decrease from 13,344 to 977, from 44,154 to 6499, and from
164,783 to 26,872. The reductions across every benchmark demonstrate that $\mathcal{E}_1$ removes a large
fraction of the assignment-variable interactions encountered during circuit compilation.

The reduction achieved by $\mathcal{E}_1$ also extends to the Hamiltonian representation. The QUBO-term
count decreases for every system, with reductions ranging from $15.1\%$ to $34.7\%$. The count decreases
from 1302 to 1105 for IEEE~24-bus, from 5306 to 3895 for IEEE~39-bus, and from 4122 to 2692 for
IEEE~57-bus. These reductions are smaller than those observed for circuit depth and gate count because the
inferred-island constraints still expand over the explicit assignment variables and slack registers.
Encoding $\mathcal{E}_1$ therefore reduces both logical width and QUBO size, while the direct penalty
construction retains a comparatively dense interaction structure. This distinction explains why the
compiled-circuit savings exceed the reduction in the number of QUBO terms.

The effect of $\mathcal{E}_1$ on solution quality is reflected jointly by the cut value and feasibility
rate. The cut value is preserved for four systems and improves from 860.528 to 776.206 for IEEE~24-bus and
from 131.917 to 128.239 for IEEE~57-bus. The feasibility rate increases for IEEE~9-, 30-, and 57-bus and
decreases for IEEE~14-, 24-, and 39-bus. The largest improvement occurs for IEEE~57-bus,
where the rate increases from 0.740 to 1.000. The decrease for IEEE~39-bus is limited to 0.105 to 0.101,
while the overall runtime decreases from 1017.5 to 681.4~s and the quantum-backend runtime decreases from
10 to 8~s. The overall runtime is lower for five systems, with the largest reduction occurring for
IEEE~57-bus from 669.6 to 312.6~s. IEEE~14-bus is the only exception, where the overall runtime increases
from 33.7 to 41.4~s despite a reduction in quantum-backend runtime from 6 to 4~s. Thus, $\mathcal{E}_1$
provides consistent Hamiltonian and circuit-resource savings with noninferior cut quality, although its
effect on the feasibility component of solution quality remains system dependent.

Encoding~$\mathcal{E}_2$ retains the working-graph and $K-1$ representation of $\mathcal{E}_1$ and further
fixes the island labels of the anchor buses. Assignment variables are then required only for the free buses,
which reduces the support of both the assignment terms and the inferred-island penalty blocks. This
additional reduction is isolated by comparing Formulations~(ii) and~(iii). Across the six systems,
$\mathcal{E}_2$ reduces the logical-qubit count by $2.6\%$ to $18.4\%$ and the QUBO-term count by $4.9\%$
to $45.5\%$ relative to $\mathcal{E}_1$. The circuit depth is reduced by factors of 1.10 to 1.96, the
two-qubit-gate count by factors of 1.32 to 2.58, and the total gate count by factors of 1.04 to 2.17. For
IEEE~24-bus, the qubit count decreases from 76 to 62, the QUBO-term count from 1105 to 909, the circuit
depth from 891 to 784, and the total gate count from 3914 to 2228. For IEEE~39-bus, the corresponding
quantities decrease from 123 to 103, from 3895 to 2121, from 977 to 848, and from 26,872 to 24,909. The
largest depth reduction occurs for IEEE~9-bus, where the depth decreases from 473 to 241 and the total gate
count decreases from 753 to 527.

Beyond its resource advantages, the symmetry-fixed encoding $\mathcal{E}_2$ provides a consistent
solution-quality improvement over $\mathcal{E}_1$ under the same direct constraint-penalty treatment. The
cut value is preserved for all six systems, while the feasibility rate increases for five systems and remains 1.000 for
IEEE~57-bus. The largest increases occur for IEEE~24-bus and IEEE~39-bus, where the feasibility rates rise
from 0.010 to 0.524 and from 0.101 to 0.882. IEEE~30-bus also reaches a feasibility rate of 1.000, compared
with 0.173 under $\mathcal{E}_1$. These improvements in solution quality are accompanied by a reduction in
overall runtime for five of the six systems. The overall runtime decreases from 681.4
to 605.6~s for IEEE~39-bus and from 312.6 to 283.3~s for IEEE~57-bus, while the quantum-backend runtime
remains 8~s in both comparisons. IEEE~24-bus is the only exception in overall runtime, which increases from
83.2 to 91.5~s even though the quantum-backend runtime decreases from 6 to 5~s. This difference indicates
that non-backend components can dominate small changes in the end-to-end runtime.

Taken together, the resource, solution-quality, and runtime results establish that $\mathcal{E}_1$ and
$\mathcal{E}_2$ provide complementary reductions through distinct structural mechanisms. Encoding
$\mathcal{E}_1$ provides the larger initial resource reduction through graph contraction and the $K-1$
assignment representation. Encoding
$\mathcal{E}_2$ builds on this structure by removing anchor-label variables and reducing the interaction
support to the free buses. Its incremental resource savings are smaller in some systems, particularly
IEEE~57-bus, but its solution-quality improvement is more uniform because the cut is preserved and the
feasibility rate never decreases. The remaining slack-based penalty blocks continue to dominate the
non-Lagrangian Hamiltonians under both encodings, which motivates the Lagrangian treatment examined next.

The Lagrangian treatment also reduces the resource requirements of both encoding strategies by eliminating the
slack registers and their associated penalty couplings. For $\mathcal{E}_1$, this effect is isolated by
comparing Formulations~(ii) and~(iv). The qubit count decreases from 23 to 9 for IEEE~9-bus, from 76 to 40
for IEEE~24-bus, and from 123 to 78 for IEEE~39-bus. The corresponding QUBO-term counts decrease from 151
to 21, from 1105 to 244, and from 3895 to 481. These changes represent reduction factors of 7.19, 4.53,
and 8.10. These factors are substantially larger than the corresponding qubit-count factors of 2.56, 1.90,
and 1.58. The principal advantage therefore arises not only from removing slack qubits but also from
eliminating the dense quadratic couplings generated by their penalty blocks. The compiled circuits reflect
this structural simplification. Circuit depth decreases from 473 to 31 for IEEE~9-bus and from 891 to 127
for IEEE~24-bus. The smaller decrease from 977 to 881 for IEEE~39-bus partly reflects the greater problem
complexity of partitioning the larger network into three islands. The remaining assignment and cut
interactions retain a comparatively dense structure, while hardware connectivity and transpilation further
constrain their scheduling. Consequently, the reduction in Hamiltonian size does not translate
proportionally into circuit depth for this case.

The solution-quality results indicate that the resource reduction does not alter the attainable cut value.
The cut is unchanged in all six comparisons, while the feasibility rate improves for five systems and
remains 1.000 for IEEE~57-bus. For IEEE~9-bus, the feasibility rate increases from 0.270 to
0.350. This consistent pattern shows that removing the slack-penalty structure improves the concentration
of post-processed samples on feasible solutions without changing the best partition quality. The
quantum-backend runtime decreases for all
six systems, confirming that the Hamiltonian simplification produces a consistent execution-time benefit.
The overall runtime decreases for five systems because non-backend processing can partially offset the
quantum-runtime reduction. For IEEE~24-bus, $t/t_q$ decreases from $83.2/6$ to $74.0/3$~s. The unchanged
cuts, uniformly nondecreasing feasibility rates, and lower runtimes demonstrate that Lagrangian handling
under $\mathcal{E}_1$ removes costly constraint structure while preserving the underlying solution space.

For $\mathcal{E}_2$, the corresponding effect follows from comparing Formulations~(iii) and~(v). The qubit
count decreases from 20 to 6 for IEEE~9-bus, from 62 to 26 for IEEE~24-bus, and from 103 to 58 for
IEEE~39-bus. The QUBO-term count decreases by factors of 8.00, 10.57, and 9.34 for the same systems. The
circuit-depth ratio of Formulation~(iii) to Formulation~(v) ranges from 1.18 for IEEE~39-bus to 11.05 for
IEEE~30-bus, while the two-qubit-gate ratio ranges from 1.80 to 5.93 across the six systems. The reductions
under both encodings confirm that removing slack-based penalty blocks is the principal source of the
Lagrangian resource advantage. Formulation~(v) largely preserves solution quality relative to
Formulation~(iii). The cut values are unchanged for every system, while the feasibility rate remains
unchanged for IEEE~9-, 30-, and 57-bus and decreases moderately for the other three systems. The rate
remains positive in every case. Overall, Lagrangian handling under $\mathcal{E}_2$ provides consistent
resource and runtime reductions while preserving the cut value and maintaining feasible post-processed
solutions across all systems.

The comparison between Formulations~(iv) and~(v) provides a direct assessment of the encoding contribution
under Lagrangian constraint handling. Replacing $\mathcal{E}_1$ with $\mathcal{E}_2$ reduces the qubit count
from 9 to 6 for IEEE~9-bus, from 40 to 26 for IEEE~24-bus, and from 78 to 58 for IEEE~39-bus. The QUBO-term
count decreases from 21 to 12, from 244 to 86, and from 481 to 227. Circuit depth decreases from 31 to 25,
from 127 to 94, and from 881 to 720. These results establish that the two design choices provide
complementary benefits. Lagrangian handling removes the slack-register overhead common to both encodings,
whereas $\mathcal{E}_2$ supplies an additional reduction by fixing anchor labels and restricting assignment
variables to the free buses. This encoding change improves solution quality by preserving the cut value for
all six systems and increasing the feasibility rate for five systems while maintaining a rate of 1.000 for
IEEE~57-bus. Under the same Lagrangian treatment, $\mathcal{E}_2$ therefore improves solution quality and
reduces logical and circuit resources as well as overall runtime without increasing quantum-backend runtime.

Applying the Lagrangian treatment to the full one-hot encoding also reduces
resource requirements by removing the slack registers and the dense quadratic couplings generated by their
penalty terms. For IEEE~24-bus, the transition from Formulation~(i) to (i)+Lag. reduces the qubit count from
96 to 60, the QUBO-term count from 1302 to 216, the circuit depth from 2864 to 1144, and the two-qubit-gate
count from 8938 to 1743. For IEEE~39-bus, the corresponding quantities decrease from 156 to 117, from 5306
to 420, from 13,344 to 4999, and from 44,154 to 8667. From a complexity perspective, the Lagrangian
treatment removes the slack-qubit contribution and its dense penalty blocks, but the $KN_0$ one-hot
assignment register remains. It therefore reduces Hamiltonian and circuit complexity without providing the
working-graph contraction and anchor-based width reductions of $\mathcal{E}_1$ and $\mathcal{E}_2$. The
full one-hot comparison confirms that removing slack-based penalties improves solution quality and lowers
overall runtime without increasing quantum-backend runtime, even when the assignment encoding remains
unchanged.

The overall advantage of the proposed framework is quantified by comparing Formulation~(v) with the full
one-hot Formulation~(i). Across the six systems, Formulation~(v) uses 61.5--81.3\% fewer logical qubits. The
full formulation requires 15.14--38.89 times as many QUBO terms, 18.53--83.31 times the circuit depth,
19.83--93.63 times as many two-qubit gates, and 18.52--71.73 times as many total gates. The IEEE~39-bus
case is particularly consequential. The full formulation occupies 156 qubits, exactly reaching the
available hardware scale, whereas the proposed formulation requires only 58 qubits. For the largest
complete comparison, IEEE~57-bus, the proposed formulation reduces the register from 130 to 50 qubits
and the total gate count from 214,315 to 2,988. Formulation~(v) improves solution quality across all six
systems by preserving or improving the cut value and achieving a higher feasibility rate. It also
has the lowest overall runtime for every system. The IEEE~39-bus feasibility rate increases from 0.105 to
0.758 as $t/t_q$ decreases from $1017.5/10$ to $440.4/6$~s. The IEEE~57-bus cut improves from 131.917 to
128.239, the feasibility rate increases from 0.740 to 1.000, and $t/t_q$ decreases from $669.6/9$ to
$276.5/6$~s. Taken together, the paired ablations establish the causal contribution of each design choice.
Compact encoding reduces assignment width, symmetry fixing removes redundant label degrees of freedom, and
slack-free Lagrangian handling removes the dense auxiliary penalty structure. Their combination in
Formulation~(v) provides the most favorable balance of hardware resources, runtime, and solution quality.

\subsection{Noise Resilience Analysis}
\label{sec:noise_resilience}

Table~\ref{tab:noise_resilience} evaluates the noise resilience of the proposed Formulation~(v) across ideal
simulation, the calibrated IBM FakeMarrakesh (FM) noise model~\cite{ibm_fakemarrakesh_2026}, and real IBM
Quantum hardware. The simulator and hardware evaluations use the system-specific QAOA and post-processing
configurations reported in Table~\ref{tab:exp_setup}, and the IBM results correspond to
Table~\ref{tab:ablation_resource}. The Aer configuration used in this study supports up to 32 logical
qubits, which restricts the analysis to the IEEE 9- through 30-bus systems under Formulation~(v).

\begin{table*}[!tbp]
\centering
\scriptsize
\setlength{\tabcolsep}{3pt}
\renewcommand{\arraystretch}{1.05}
\caption{Noise-resilience comparison of Formulation~(v) under ideal simulation, FakeMarrakesh noise, and IBM hardware execution.}
\label{tab:noise_resilience}
\resizebox{\textwidth}{!}{%
\begin{tabular}{@{}lrrrrrrrrrrrrrrr@{}}
\toprule
\textbf{IEEE} & \multicolumn{3}{c}{\textbf{Feasible-sample rate}} &
\multicolumn{3}{c}{\textbf{Cut value}} & \multicolumn{3}{c}{\textbf{Circuit depth}} &
\multicolumn{3}{c}{\textbf{2Q gates}} & \multicolumn{3}{c}{\textbf{Total gates}} \\
\cmidrule(lr){2-4}\cmidrule(lr){5-7}\cmidrule(lr){8-10}\cmidrule(lr){11-13}\cmidrule(l){14-16}
\textbf{case} & \textbf{Ideal} & \textbf{FM} & \textbf{IBM} &
\textbf{Ideal} & \textbf{FM} & \textbf{IBM} &
\textbf{Ideal} & \textbf{FM} & \textbf{IBM} &
\textbf{Ideal} & \textbf{FM} & \textbf{IBM} &
\textbf{Ideal} & \textbf{FM} & \textbf{IBM} \\
\midrule
9-bus  & 1.000 & 1.000 & 1.000 & 1.238   & 1.238   & 1.238   &23  & 23  & 25 & 21  & 21  & 21  & 132  & 132  & 135 \\
14-bus & 0.490 & 0.400 & 0.360 & 88.241  & 88.241  & 88.241  & 72  & 72  & 77 & 24  & 24  & 27  & 171  & 171  & 184 \\
24-bus & 0.475 & 0.510 & 0.454 & 776.206 & 776.206 & 776.206 & 91 & 91 & 94 & 235 & 235 & 248 & 1117 & 1117 & 1198 \\
30-bus & 1.000 & 1.000 & 1.000 & 16.863  & 16.863  & 16.863  & 57 & 57 & 62 & 124 & 124 & 127 & 653  & 653  & 668 \\
\bottomrule
\end{tabular}
}
\end{table*}

As shown in Table~\ref{tab:noise_resilience}, Formulation~(v) yields identical cut values under ideal
simulation, FM noise, and IBM hardware for all four systems. Hardware execution also preserves the
feasible-sample concentration observed under calibrated noise. The IBM rate remains 1.000 for IEEE~9- and
30-bus and is only 0.040 and 0.056 below the FM rate for IEEE~14- and 24-bus, respectively. The additional
reduction from FM simulation to real hardware is therefore at most 5.6\% and averages only 2.4\% across the
four systems. Together with the unchanged cut values, this limited
degradation indicates that hardware noise slightly alters the sampling distribution without compromising
the final solution quality of the proposed framework.

The ideal and FM executions use identical compiled circuits, which isolates the effect of the calibrated
noise model from circuit-structure differences. The IBM circuits remain closely comparable, with depths
that differ from the simulator values by only 3.3--8.8\%. Their 2Q and total-gate counts remain
within approximately 11.1\% and 7.1\% of the simulator values, respectively. The consistency of the
compiled resources and cut values indicates that the observed feasibility variations arise from sampling
and device effects rather than a change in the optimization problem or circuit scale.

Collectively, the ideal, calibrated-noise, and IBM hardware results demonstrate that the proposed hybrid
framework is resilient to the tested quantum noise conditions. Device noise changes the frequency of
feasible samples but does not alter the final cut value obtained after constrained post-processing.

\subsection{QAOA Optimization Landscapes and Hamiltonian Complexity}
\label{sec:landscape_trainability}

The qubit reductions established in Section~\ref{sec:complexity} also change the numerical structure of the
QAOA optimization problem. This connection is examined using complete five-formulation landscape records for
the IEEE 9-, 14-, 24-, 30-, 39-, and 57-bus systems obtained on IBM Quantum hardware. The comparison follows
the formulation order in Table~\ref{tab:complexity_summary}, with Formulations~(i)--(iii) representing the
full or direct non-Lagrangian constructions and Formulations~(iv)--(v) representing the compact Lagrangian
constructions under $\mathcal{E}_1$ and $\mathcal{E}_2$. Taken together,
Figs.~\ref{fig:gradient_samples}--\ref{fig:hamiltonian_complexity} reveal how reductions in Hamiltonian
complexity reshape the gradient scale and two-dimensional geometry of the QAOA parameter landscape.

\begin{figure}[!t]
\centering
\includegraphics[width=\columnwidth]{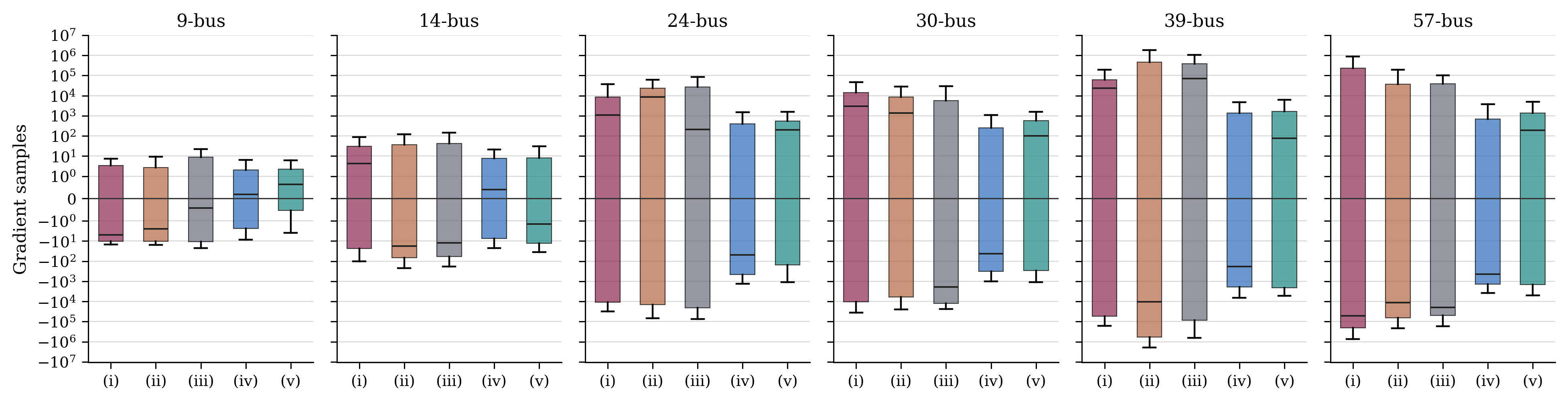}
\caption[Sampled gradient distributions]{Distribution of sampled finite-difference gradients for the five
formulations across IEEE benchmark systems. The horizontal line marks zero gradient, and formulation labels
(i)--(v) follow Table~\ref{tab:complexity_summary}.}
\label{fig:gradient_samples}
\end{figure}

The gradient distributions in Fig.~\ref{fig:gradient_samples} show that formulation choice becomes
increasingly important as system size grows. On the 9- and 14-bus systems, all five formulations occupy
relatively similar gradient scales. Beginning with the 24-bus case, the distributions for
Formulations~(i)--(iii) broaden by several orders of magnitude and develop extreme positive and negative
samples. This growth is most pronounced on the 39- and 57-bus systems, where the direct formulations extend
into the $10^{5}$--$10^{6}$ range. Formulations~(iv) and~(v) retain substantially narrower distributions on
the same systems, while still exhibiting signed variation around zero. The compact formulations therefore
reduce extreme changes in the objective without collapsing the sampled surface to a constant value. Among
them, Formulation~(v) generally maintains the most controlled gradient scale as the network size increases.

\begin{figure}[!t]
\centering
\includegraphics[width=\columnwidth]{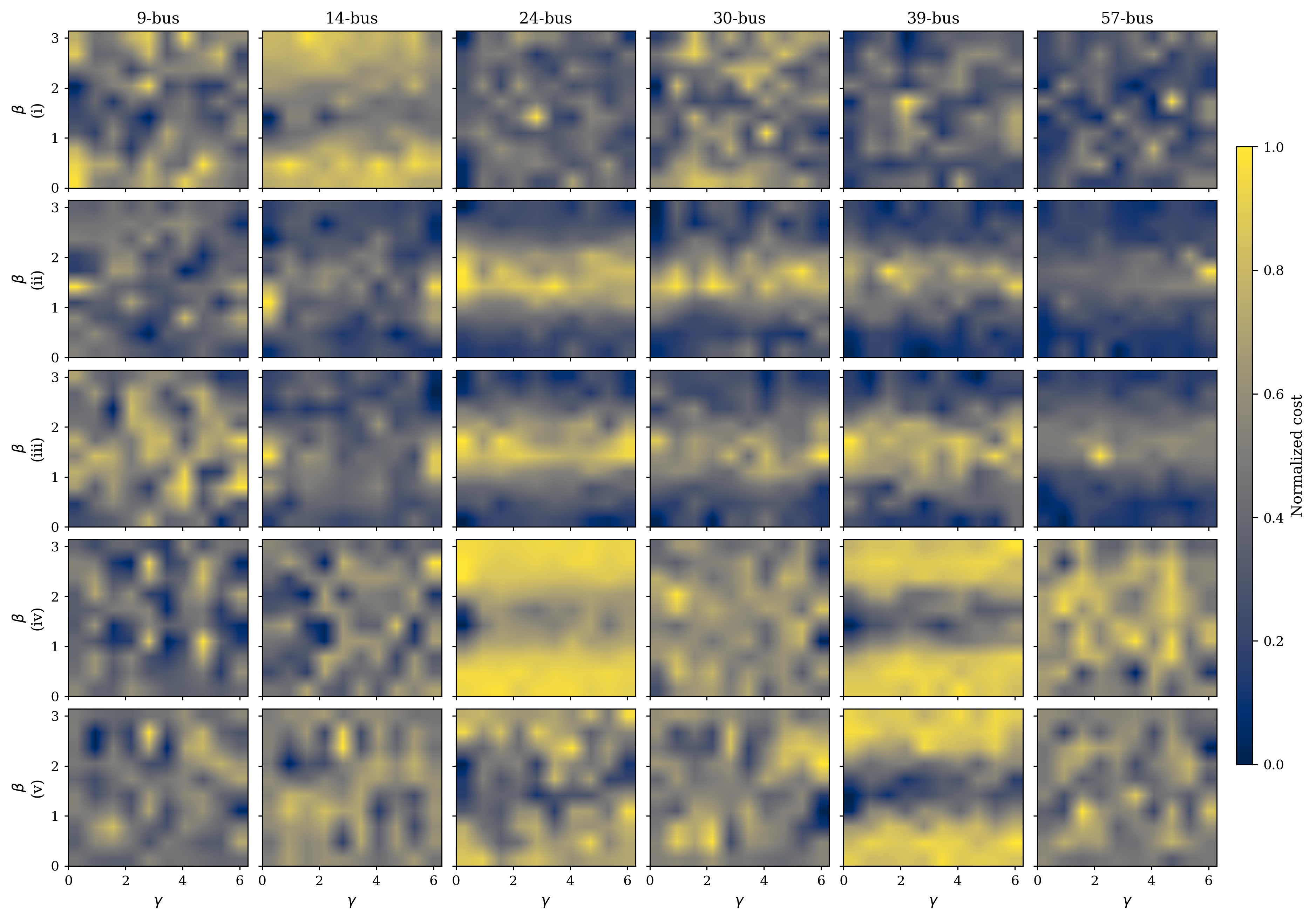}
\caption[Normalized QAOA cost landscapes]{Normalized QAOA cost landscapes over $(\gamma,\beta)$ for all five
formulations and the IEEE benchmark systems. Each row corresponds to one formulation and each column to one benchmark
system, using the same normalization within each panel.}
\label{fig:landscape_gallery_all5}
\end{figure}

Beyond these aggregate gradient statistics, the full cost landscapes in Fig.~\ref{fig:landscape_gallery_all5}
reveal how the scale differences translate into structure across the $(\gamma,\beta)$ parameter domain. Because
each panel is normalized to its own cost range, the comparison isolates the shape of the landscape from its
absolute objective magnitude. The direct formulations
develop increasingly dominant ridges, broad bands, and abrupt transitions on the 24-, 30-, 39-, and 57-bus
systems. Such structures concentrate large cost changes into narrow regions of the parameter plane, producing
sharp barriers between basins that make the outcome of a local search highly sensitive to initialization and
step size and that offer little reliable gradient information away from the dominant ridges.

In contrast, the Lagrangian formulations~(iv) and~(v) produce markedly more favorable optimization surfaces.
They retain clear parameter dependence but spread the variation smoothly over a moderate set of broad basins
and bands, so that a descent direction obtained near almost any initialization points toward the same
low-cost region. This regularity is especially pronounced for Formulations~(iv) and~(v) on the 24- and 39-bus
cases, where the dominant large-scale structure remains a single identifiable basin rather than the highly
fragmented extrema seen in the direct formulations. As a result, the compact Lagrangian encodings are
substantially easier to optimize, since they widen the range of effective initial parameters, reduce the risk
of the optimizer stalling on a spurious ridge, and make the modest one-to-three-layer QAOA circuits used in
the benchmarks sufficient to locate high-quality parameters. Formulation~(v), which combines symmetry fixing with
the slack-free Lagrangian treatment, exhibits the smoothest and most consistently navigable landscape across
system sizes. Together with the qubit reductions of Table~\ref{tab:concrete_qubits}, this shows that the
proposed encodings improve trainability at the same time as they lower circuit width, so their resource
savings do not come at the cost of a harder optimization problem.

\begin{figure}[!t]
\centering
\includegraphics[width=\columnwidth]{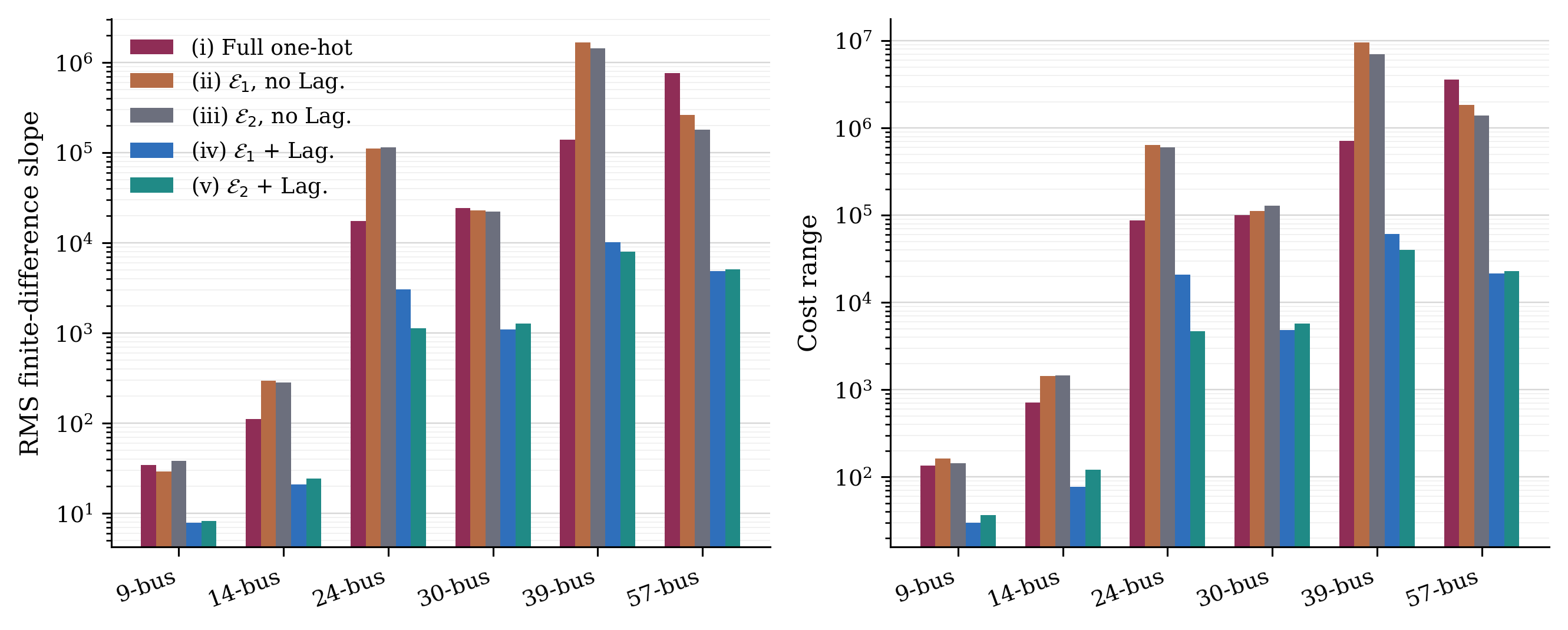}
\caption[Landscape-shape characteristics]{Landscape-shape characteristics for the five formulations, measured by RMS
finite-difference slope and cost range. The direct non-Lagrangian formulations often produce larger cost
ranges and steeper landscapes, whereas the Lagrangian formulations keep these quantities closer to the scale
of the proposed compact Hamiltonian.}
\label{fig:landscape_shape_diagnostics}
\end{figure}

While the normalized panels of Fig.~\ref{fig:landscape_gallery_all5} expose these geometric differences
qualitatively, Fig.~\ref{fig:landscape_shape_diagnostics} measures them directly through the RMS
finite-difference slope and cost range of each landscape before normalization. These two quantities turn the
visual contrast between the direct and Lagrangian formulations into concrete numbers that scale with system
size. For the 9- and 14-bus cases, the RMS slopes and cost ranges remain
within a comparatively narrow interval across the formulations. The separation becomes substantial from
24 buses onward. On the 39-bus system, for example, the direct $\mathcal{E}_1$ and $\mathcal{E}_2$
formulations reach RMS slopes above $10^{6}$, whereas the two Lagrangian formulations remain near $10^{4}$.
Their cost ranges show an even larger separation, with the direct formulations approaching $10^{7}$ and the
Lagrangian formulations remaining below $10^{5}$. A similar pattern appears for IEEE~57-bus, where the
Lagrangian cost ranges are near $2\times10^{4}$ while the direct formulations lie between approximately
$10^{6}$ and $4\times10^{6}$. Formulation~(v) provides the lowest or nearly lowest slope and range in most
cases. These reductions limit the dynamic range that must be handled by the classical optimizer and make the
parameter response more consistently scaled across network sizes.

\begin{figure}[!t]
\centering
\includegraphics[width=\columnwidth]{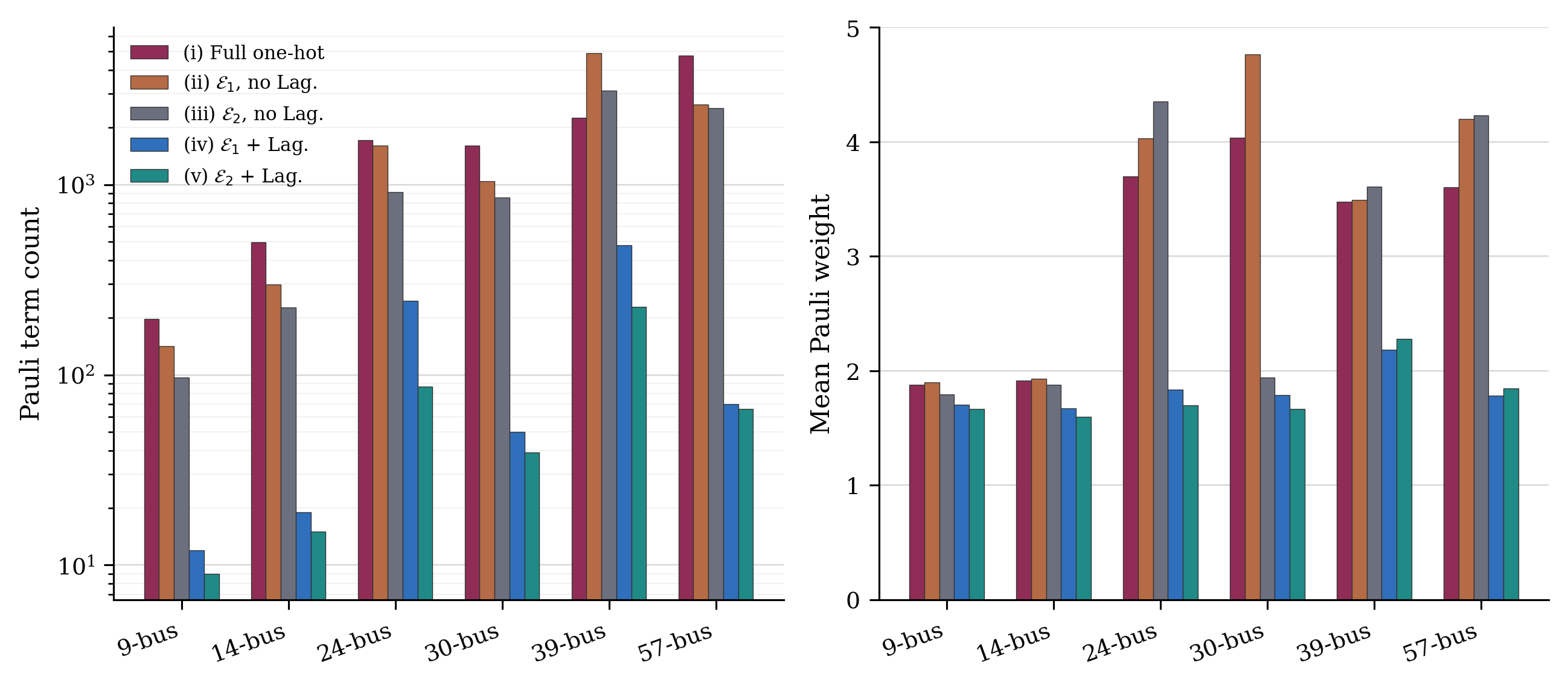}
\caption[Hamiltonian complexity comparison]{Hamiltonian complexity comparison for the five formulations. The
left panel reports Pauli-term count and the right panel reports mean Pauli weight. The proposed Lagrangian
formulations reduce the term count by removing explicit slack-variable penalty blocks, while Formulation~(v)
also benefits from encoding $\mathcal{E}_2$.}
\label{fig:hamiltonian_complexity}
\end{figure}

The preceding differences in slope and cost range originate in the structure of the cost Hamiltonian, which
Fig.~\ref{fig:hamiltonian_complexity} characterizes through the Pauli-term count and mean Pauli weight of each
formulation. Formulations~(iv) and~(v) contain
substantially fewer Pauli terms than
Formulations~(i)--(iii) for every system. The difference grows with network size. On the 57-bus system, the
three direct formulations contain on the order of $10^{3}$ terms, while the Lagrangian formulations require
only several tens. Formulation~(v) gives the smallest term count in nearly every case because $\mathcal{E}_2$
removes additional assignment redundancy. The same formulations also reduce mean Pauli weight. For the
24- to 57-bus systems, the direct formulations typically have mean weights between approximately 3.5 and
4.8, whereas Formulations~(iv) and~(v) remain near 1.7--2.3. The resulting Hamiltonians therefore contain
fewer interactions and involve fewer qubits per term on average.

The results in Figs.~\ref{fig:gradient_samples}--\ref{fig:hamiltonian_complexity} are consistent with a
single underlying mechanism that links resource efficiency to optimization behavior. Removing explicit
penalty-variable blocks reduces the Pauli-term count and interaction weight, which in turn suppresses the
extreme slope and cost-range growth observed in the direct formulations. Encoding $\mathcal{E}_2$ strengthens this effect by further
reducing redundant assignment structure. The proposed formulation therefore improves quantum applicability
in two complementary ways. It lowers the qubit and circuit requirements demonstrated in Tables~
\ref{tab:lagrangian_bench} and~\ref{tab:sf_bench}, and it produces a more controlled parameter landscape for
the classical optimization loop. This combination is important for scaling QAOA-based controlled islanding
beyond small benchmark systems on present quantum hardware.

\section{Conclusion}
\label{conclusion}

In this study, a scalable quantum optimization framework has been established for controlled power system
islanding that retains its essential combinatorial and physical requirements within the resource limits of
near-term quantum hardware. The proposed resource-efficient island-assignment representation captures the
essential partitioning decisions with substantially fewer quantum variables. In addition, a slack-free
Lagrangian treatment complements the compact representation by eliminating auxiliary constraint registers
while preserving the optimization and physical requirements of controlled islanding. Collectively, these
advances reduce quantum resource requirements while maintaining solution quality and operational
feasibility. For fixed $K$, the resulting formulation reduces the phase-separator complexity from
$O(E_0+N_0^2)$ to $O(E+F+N)$ and, on sparse working graphs, achieves linear per-layer gate complexity
$O(N)$ with circuit depth $O(p\Delta_G)$. This complexity advantage improves hardware compatibility and
optimization trainability and extends the scale of executable islanding problems without requiring deeper
circuits or larger sampling budgets.

The proposed hybrid quantum design framework produces high-quality, feasible islanding solutions across
IEEE benchmark systems ranging from 9 to 89 buses using shallow circuits and practical sampling budgets,
with power-flow validation confirming the operational feasibility of the resulting islands. Benchmarking
across multiple quantum-provider backends further demonstrates that the proposed framework can be adapted
to different hardware architectures, supporting its broad applicability to quantum optimization in power
systems. Across ideal simulation, calibrated FakeMarrakesh noise, and IBM hardware, the framework preserves
the cut quality with only small variations in feasible-sample rates, demonstrating resilience to the tested
device noise. The factorial ablation confirms that compact encoding and slack-free Lagrangian handling provide
complementary reductions in quantum resources, circuit complexity, and runtime while maintaining solution
quality. Across the six complete ablation cases, the proposed formulation uses 61.5--81.3\% fewer logical
qubits than the full one-hot formulation and reduces the QUBO-term count and circuit depth by factors of
15.14--38.89 and 18.53--83.31, respectively. Consistent with these resource-level findings, the landscape
analysis shows that the associated Hamiltonian simplification produces smoother and more consistently
scaled QAOA cost surfaces, thereby demonstrating that resource efficiency is accompanied by improved
parameter-optimization behavior.

Future work will strengthen the quantum component through adaptive-depth QAOA, cross-instance parameter
transfer, error-mitigation strategies, and circuit compilation tailored to device connectivity. These
developments will be combined with extensions to dynamic and uncertainty-aware islanding and scalable
hybrid treatment of richer stability and AC power-flow requirements.

\appendix
\section{Proof of the Encoding Hierarchy Theorem}
\label{app:encoding_hierarchy}

This appendix establishes formally that $\mathcal{E}_2$ provides an exact
symmetry reduction of $\mathcal{E}_1$ together with a strict reduction in the
assignment-register size.  Under the controlled-islanding structure of
Section~\ref{formulation}, every feasible physical partition retains a unique
canonical representative, the optimal cut is preserved, and the $K!$
label-equivalent representations of each feasible physical partition are
reduced to one.

\begin{theorem}[Encoding hierarchy]
\label{thm:encoding_hierarchy}
Under the assumptions of Proposition~\ref{prop:canonical_labeling}, with
$K\geq 2$, define the cut function
\begin{equation}
f_{\mathrm{cut}}(\pi)
=\sum_{(i,j)\in E}w_{ij}\,\mathbf{1}[\pi(i)\neq\pi(j)].
\label{eq:appendix_cut_function}
\end{equation}
Let $\Pi_{\mathrm{feas}}$ denote the nonempty set of all assignments
$\pi:V\to\{1,\ldots,K\}$ satisfying
constraints~\eqref{con:onehot}--\eqref{con:connect}, and define its canonical
subset
\begin{equation}
\Pi_{\mathrm{can}}
=
\bigl\{\pi\in\Pi_{\mathrm{feas}}:
\pi(a)=c\text{ for every }a\in\mathcal{A}_c,
\ c=1,\ldots,K\bigr\}.
\label{eq:canonical_feasible_set}
\end{equation}
The following hold.
\begin{enumerate}[(i)]
\item \emph{Algebraic soundness of $\mathcal{E}_1$.}
      For every bitstring $\{y_{i,k}\}$, the indicators $Y_{i,k}$ defined
      in~\eqref{eq:reduced_inferred_var} satisfy
      $\sum_{k=1}^{K}Y_{i,k}=1$ for every $i\in V$.  If the bitstring also
      satisfies the validity condition~\eqref{con:reduced_validity}, then
      $Y_{i,k}\in\{0,1\}$ for every $i\in V$ and $k=1,\ldots,K$.
\item \emph{Completeness of $\mathcal{E}_2$.}
      Every $\pi\in\Pi_{\mathrm{feas}}$ has a relabeled representative
      $\pi'\in\Pi_{\mathrm{can}}$, with
      $f_{\mathrm{cut}}(\pi')=f_{\mathrm{cut}}(\pi)$.
\item \emph{Strict assignment-register reduction and optimality preservation.}
      The assignment-register sizes satisfy
      $n_q^{\mathcal{E}_2}<n_q^{\mathcal{E}_1}$ whenever
      $|\mathcal{A}|\geq 1$, and
      \begin{equation}
      \min_{\pi\,\in\,\Pi_{\mathrm{can}}}f_{\mathrm{cut}}(\pi)
      =
      \min_{\pi\,\in\,\Pi_{\mathrm{feas}}}f_{\mathrm{cut}}(\pi).
      \label{eq:opt_preserved}
      \end{equation}
\item \emph{Symmetry reduction.}
      Among feasible labeled assignments, every feasible physical partition
      has exactly $K!$ representatives under $\mathcal{E}_1$ and exactly one
      canonical representative under $\mathcal{E}_2$.
\end{enumerate}
\end{theorem}

\begin{proof}
\textbf{(i)}\ By construction,
\begin{equation}
\sum_{k=1}^{K}Y_{i,k}
=\sum_{k=1}^{K-1}y_{i,k}
 +\Bigl(1-\sum_{k=1}^{K-1}y_{i,k}\Bigr)=1.
\label{eq:sum_to_one_proof}
\end{equation}
This identity holds without invoking the validity condition.  Now suppose
that~\eqref{con:reduced_validity} also holds.  Since each
$y_{i,k}\in\{0,1\}$, at most one explicit variable equals one.  If exactly one
equals one for island $k_0$,
then $Y_{i,k_0}=1$ and $Y_{i,K}=0$, so all indicators are in $\{0,1\}$.  If
all explicit variables are zero, then $Y_{i,K}=1$ and all others are zero.
In both cases $Y_{i,k}\in\{0,1\}$ for every $k$.

\textbf{(ii)}\ Let $\pi\in\Pi_{\mathrm{feas}}$.  By
constraint~\eqref{con:same_group}, all
buses in $\mathcal{A}_c$ are assigned to one island; by
constraint~\eqref{con:different_group}, no two distinct anchor groups share
an island.  The induced map from anchor groups to island labels is therefore
injective from a set of size $K$ to a set of size $K$, hence bijective.
Let $\sigma:\{1,\ldots,K\}\to\{1,\ldots,K\}$ be the permutation that sends the
island occupied by $\mathcal{A}_c$ to label $c$.  Define $\pi'=\sigma\circ\pi$.
By label invariance of the objective~\eqref{eq:islanding_objective} and all
constraints~\eqref{con:onehot}--\eqref{con:connect},
$\pi'\in\Pi_{\mathrm{feas}}$ and
$f_{\mathrm{cut}}(\pi')=f_{\mathrm{cut}}(\pi)$.  Moreover,
$\pi'(i)=c$ for all $i\in\mathcal{A}_c$ and $c=1,\ldots,K$.  Therefore
$\pi'\in\Pi_{\mathrm{can}}$ by~\eqref{eq:canonical_feasible_set}.

\textbf{(iii)}\ The assignment-register sizes satisfy
\begin{equation}
n_q^{\mathcal{E}_1}-n_q^{\mathcal{E}_2}
=(K-1)N-(K-1)(N-|\mathcal{A}|)
=(K-1)|\mathcal{A}|>0
\label{eq:qubit_gap}
\end{equation}
because $K\geq2$ and $|\mathcal{A}|\geq 1$, establishing the strict reduction.
Equivalently, the ambient binary assignment space decreases from
$2^{(K-1)N}$ to $2^{(K-1)(N-|\mathcal{A}|)}$, a reduction factor of
$2^{(K-1)|\mathcal{A}|}$.  This count concerns the assignment register and is
independent of any additional registers introduced by a particular constraint
treatment.

For the cut objective, the definition~\eqref{eq:canonical_feasible_set} gives
$\Pi_{\mathrm{can}}\subseteq\Pi_{\mathrm{feas}}$.  Hence
\begin{equation}
\min_{\pi\in\Pi_{\mathrm{can}}}f_{\mathrm{cut}}(\pi)
\;\geq\;
\min_{\pi\in\Pi_{\mathrm{feas}}}f_{\mathrm{cut}}(\pi).
\label{eq:lower_bound}
\end{equation}
By part~(ii), every $\pi\in\Pi_{\mathrm{feas}}$ has a representative
$\pi'\in\Pi_{\mathrm{can}}$ with
$f_{\mathrm{cut}}(\pi')=f_{\mathrm{cut}}(\pi)$,
so
\begin{equation}
\min_{\pi\in\Pi_{\mathrm{can}}}f_{\mathrm{cut}}(\pi)
\;\leq\;
\min_{\pi\in\Pi_{\mathrm{feas}}}f_{\mathrm{cut}}(\pi).
\label{eq:upper_bound}
\end{equation}
Combining~\eqref{eq:lower_bound} and~\eqref{eq:upper_bound}
gives~\eqref{eq:opt_preserved}.

\textbf{(iv)}\ Let $S_K$ denote the symmetric group on the $K$ island labels.
Any feasible physical partition into $K$ islands admits exactly $K!$ feasible
labelings $\pi_\sigma=\sigma\circ\pi$ for $\sigma\in S_K$, because feasibility
is invariant under a common permutation of the island labels.  Each distinct
$\sigma$ produces a different valid assignment under $\mathcal{E}_1$, so the
partition has exactly $K!$ feasible representatives under $\mathcal{E}_1$.  In
$\mathcal{E}_2$, the canonical
condition $\pi'(i)=c$ for $i\in\mathcal{A}_c$ uniquely identifies $\sigma$ as
the permutation constructed in part~(ii), so exactly one of the $K!$ labelings
satisfies the anchor-fixing constraint.  Hence $\mathcal{E}_2$ retains exactly
one canonical representation of the same physical partition.
\end{proof}

\begin{acks}
This work is supported by the Office of Naval Research under award
N00014-22-1-2504 and the National Science Foundation under awards
OAC-2417773 and ECCS-2413237/2413238.
\end{acks}
\bibliographystyle{ACM-Reference-Format}
\bibliography{References}
\end{document}